\documentclass[11pt]{article}
\usepackage[a4paper, top=2.5cm, bottom=2.5cm, left=2.5cm, right=2.5cm]
{geometry}

\usepackage[utf8]{inputenc} % allow utf-8 input
\usepackage[T1]{fontenc}    % use 8-bit T1 fonts
\usepackage{hyperref}       % hyperlinks
\usepackage{url}            % simple URL typesetting
\usepackage{booktabs}       % professional-quality tables
\usepackage{amsfonts}       % blackboard math symbols
\usepackage{nicefrac}       % compact symbols for 1/2, etc.
\usepackage{microtype}      % microtypography
\usepackage{xcolor}         % colors

\usepackage{graphicx}
\usepackage{amsmath}
\usepackage{amsthm}
\usepackage{amsfonts}
\usepackage{amssymb}
\usepackage{natbib}
\usepackage{makecell}
\usepackage{algorithmic}
\usepackage[linesnumbered,algoruled,boxed,lined]{algorithm2e}
\usepackage{xcolor}
\usepackage{subcaption}
\usepackage{soul}
\usepackage{multirow}
\usepackage{hyperref}
\usepackage{booktabs}

\newtheorem{definition}{Definition}
\newtheorem{proposition}{Proposition}

\newtheorem{theorem}{Theorem}

\newtheorem{lemma}{Lemma}

\newtheorem{observation}{Observation}

\usepackage{makecell}
\usepackage{enumitem}
\usepackage{microtype}
\usepackage{multirow}
\usepackage{bm}
\usepackage{bigstrut}
\usepackage{subcaption}
\usepackage{xcolor}

\title{Online Fair Allocations with Binary Valuations and Beyond\footnote{This is the second version. We have simplified some algorithms and optimized the whole structure.}}

\author{
Yuanyuan Wang$^1$ \hspace{15pt}
Tianze Wei$^2$ \hspace{15pt}\\ 
$^1$School of Mathematical Science, The Ocean University of China \\ \small
\texttt{wyy8088@stu.ouc.edu.cn} \\ 
$^2$Department of Computer Science, City University of Hong Kong \\ \small
\texttt{t.z.wei-8@my.cityu.edu.hk} \\ 
}
\date{}

\begin{document}

\maketitle

\begin{abstract}
In an online fair allocation problem, a sequence of indivisible items arrives online and needs to be allocated to offline agents immediately and irrevocably.
In our paper, we study the online allocation of either goods or chores. 
We employ popular fairness notions, including envy-freeness up to one item (EF1) and maximin share fairness (MMS) to capture fairness, and utilitarian social welfare (USW) to measure efficiency. 
For both settings of items, we present a series of positive results regarding the existence of fair and efficient allocations with widely studied classes of additive binary and personalized bi-valued valuation/cost functions. 
Furthermore, we complement our results by constructing counterexamples to establish our results as among the best guarantees possible.
\end{abstract}

\section{Introduction}
The fair division of indivisible items is a prominent topic in algorithmic game theory and artificial intelligence, with practical applications \cite{budish2011combinatorial,goldman2015spliddit}. Most prior work has focused on \textit{offline} settings, where the information of all items is known a priori. In practice, items are often online, and waiting for all resources to appear can significantly impede progress. We must proactively allocate resources in a timely manner for better outcomes.
A paradigmatic example is the food banks problem \cite{lee2019webuildai}, where food arrives online without prior knowledge of future items. 
Since the food is perishable, it must be allocated upon arrival to some agent. Motivated by these real-world scenarios, we define and discuss the following problem:

(\textit{Online Fair Division})
There are $n$ agents that are fixed and known in advance. The items arrive online (the number of items is \textit{unknown}). Upon the arrival of an item, its values to all agents are revealed, and we must allocate it to some agent or discard it immediately and irrevocably. At all times, the allocation among agents must be fair.

If the online algorithm knows additional information about all future items beforehand, it somewhat reduces the uncertainty of the online model, allowing for more effective decision-making and improved outcomes. \cite{zhou2023multi,benade2024fair} study the normalized valuations, i.e., online algorithms know the total value of future items. \cite{benade2022dynamic} studies the online fair allocation, foreseeing the rank of a new arrival item with respect to the past items under different assumptions of agents' valuation distributions. And there are other studies about additional information \cite{he2019achieving,elkind2024temporal,cookson2024temporal}. For an independent interest, we also examine the online setting with additional information, with relevant results presented in the Appendix.

In numerous practical applications, such as vaccine allocation \cite{rey2023vaccine},
the online algorithm operates without prior knowledge of future items. \cite{aleksandrov2015online} studies the online fair division with additive binary valuations and shows that an envy-free up to one item (EF1) allocation exists. However, there are a few non-trivial approximation fairness results of online goods allocation beyond additive binary valuations and online chores allocation. In our paper, we aim to address this gap and study the online allocation model without any information about future items. We discuss the goods and chores settings, respectively.
Our primary research question is:

\begin{quote}

\textit{
Can we design deterministic online algorithms that are fair and efficient, for valuation functions beyond additive binary, in the online allocation problem?
}
\end{quote}

% the online algorithm knows nothing about future items beforehand.

\subsection{Our Contributions}

We study two settings: goods and chores, in the online fair indivisible items allocation model.

In the goods setting, for different classes of valuations, including binary and additive personalized bi-valued, our main results are summarized in Table \ref{tab: summary_results_goods}.
We introduce the non-wastefulness (NW) notion in this setting.
First, we show the inapproximability of EF1 and MMS with general additive valuations, even for only two agents with additive personal tri-valued valuations. 
Then, we study some restricted valuation functions, such as submodular binary valuation functions.
%This means that we can only focus our efforts on restricted valuation functions.

% \begin{minipage}[t]{0.45\textwidth} 

 \begin{table}[t]
 \centering
    \begin{tabular}{c|c|c|c|c}
    \toprule
       Valuations & \makecell[c]{Number \\of agents} & \makecell[c]{Fairness/\\Efficiency} & LB & UB \\
      \midrule
      \multirow{2}{*}{\makecell[c]{Additive\\(personalized tri-valued)}}&\multirow{2}{*}{$n$}& EF1 & \multicolumn{2}{c}{$\times$ (Thm \ref{the: trivalued_ef1_mms_goods})} \bigstrut\\\cline{3-5}
         ~ & ~ & MMS & \multicolumn{2}{c}{$\times$ (Thm \ref{the: trivalued_ef1_mms_goods})}\bigstrut\\\cline{1-5}
         \multirow{3}{*}{ \makecell[c]{Submodular\\ binary}}& \multirow{3}{*}{$n$}  & EF1  & \makecell[c]{$\frac{1}{2}$\\(Thm \ref{the: indivisible_submodular+binary})} & \makecell[c]{$\frac{1}{2}$\\(Thm \ref{the: martroid_rank_negative_results})}\bigstrut\\\cline{3-5}
         ~ & ~& MMS  & \makecell[c]{$\frac{1}{2}$\\(Thm \ref{the: indivisible_submodular+binary})} &  \makecell[c]{$\frac{1}{2}$\\(Thm \ref{the: martroid_rank_negative_results})}\bigstrut\\\cline{3-5}
        ~ & ~ & USW  & \makecell[c]{$\frac{1}{2}$\\(Thm \ref{the: indivisible_submodular+binary})} & \makecell[c]{$\frac{1}{2}$\\(Thm \ref{the: martroid_rank_negative_results})} \bigstrut\\\cline{1-5}
       \multirow{4}{*}{\makecell[c]{Additive personalized \\bi-valued}}&\multirow{4}{*}{$2$}& EF1 & \makecell[c]{$\frac{1}{2}$\\(Thm \ref{the: two_agents_bivalued_ef1_mms})} & \makecell[c]{$\frac{1}{2}$\\(Thm \ref{the: ef1_bi_valued_impossiblity_results})}\bigstrut\\\cline{3-5} 
        % ~ & ~ & ~ & \multicolumn{2}{c}{$1^{\clubsuit}$ (Thm \ref{the: two_agents_bivalued_save_item})}\bigstrut\\\cline{3-5}
         ~ & ~ & MMS & \makecell[c]{$\frac{1}{3}$\\(Thm \ref{the: two_agents_bivalued_ef1_mms})}& \makecell[c]{$\frac{1}{3}$\\(Thm \ref{the: ef1_bi_valued_impossiblity_results})}\bigstrut\\\cline{3-5}
          ~ & ~ & USW & \multicolumn{2}{c}{$\times$ (Prop \ref{prop: bivalued_no_efficiency_goods})}\bigstrut\\ \cline{1-5}
         \multirow{3}{*}{\makecell[c]{Identical additive\\ binary\\+ personalized bi-valued}}&\multirow{3}{*}{$n$}& EF1 & \multicolumn{2}{c}{1 (Thm \ref{the: binary_bivalued_goods})}\bigstrut\\\cline{3-5}
        ~ & ~ & MMS & \multicolumn{2}{c}{1 (Thm \ref{the: binary_bivalued_goods})}\bigstrut\\ \cline{3-5}
        ~ & ~ & USW & \multicolumn{2}{c}{$\times$ (Prop \ref{prop: binary_bivalued_no_efficiency_goods})}\\
        \bottomrule
    \end{tabular} 
    \caption{Goods setting. LB: lower bound. UB: upper bound. "$\times$" indicates inapproximability.} 
    % and $\clubsuit$ indicates that each arriving item has a deadline.}
    \label{tab: summary_results_goods}
\end{table}
% \hfill

% Unless specified otherwise, we assume that the valuation or cost functions are additive.

% In the goods setting, we first show the inapproximability of EF1 and MMS with general additive valuations, even for only two agents with additive tri-valued valuations. 
% This means that we can only focus our efforts on restricted valuation functions.
% This means that we can only focus our efforts on additive bi-values.
%we first build the upper bounds of the approximations of EF1 and MMS, where we show the inapproximability of EF1 and MMS for tri-valued valuations.
%Then, we start studying the submodular binary valuations, and then focus on the additive personalized bi-valued valuation function and its combination with additive binary valuation functions.
%More specifically,
\noindent
\begin{itemize}[left=0pt]%
  \item For submodular binary valuation functions, we propose the Marginal-Greedy Algorithm (Algorithm \ref{alg: indivisible_submodular+binary}) to compute an allocation that simultaneously satisfies 
  NW, $\frac{1}{2}$-EF1, $\frac{1}{2}$-MMS, and $\frac{1}{2}$-max-USW, and we show that these approximation ratios are all tight.
  %Further, we show that no deterministic algorithm can guarantee an NW allocation that is $\alpha$-EF1, $\alpha$-MMS, or $\alpha$-max-USW for any $\alpha>\frac{1}{2}$.
    \item For two agents with additive personalized bi-valued valuation functions, we design the Adapted Envy-Graph Procedure (Algorithm \ref{alg: two_agents_ef1_mms}) to compute an NW allocation satisfying $\frac{1}{2}$-EF1 and $\frac{1}{3}$-MMS simultaneously.
    Further, we show that no deterministic algorithm can guarantee an NW allocation that is $\alpha$-EF1 or $\beta$-MMS for any $\alpha>\frac{1}{2}$, $\beta>\frac{1}{3}$, and no deterministic algorithm can compute an NW allocation where $\alpha$-EF1 (or $\alpha$-MMS) is compatible with $\beta$-max-USW for any $\alpha, \beta>0$. 
    \item 
    % When $(n-1)$ agents have identical binary valuation functions and one agent has the personalized bi-valued valuation function,
    When one agent has the additive personalized bi-valued valuation function and the remaining agents have identical additive binary valuation functions,
    we design the Adapted-Picking Algorithm (Algorithm \ref{alg: binary_bivalued_ef1_mms}) to compute an allocation satisfying NW, EF1, and MMS.
    Further, we show that no deterministic algorithm can compute an NW allocation that is EF1 (or MMS) and $\alpha$-max-USW for any $\alpha>0$.
\end{itemize}

 \begin{table}[tb]  
 \centering
    \begin{tabular}{c|c|c|c|c}
    \toprule
       Costs & \makecell[c]{Number \\ of agents} & \makecell[c]{Fairness/\\Efficiency}& LB & UB \\
       \midrule
       \makecell[c]{Additive \\(personalized tri-valued)}&$n$& EF1 & \multicolumn{2}{c}{$\times$ (Thm \ref{the: trivalued_ef1_chores})} \bigstrut\\\cline{1-5}
      \multirow{3}{*}{ \makecell[c]{Additive\\ binary}}&\multirow{3}{*}{$n$}& EF1 & \multicolumn{2}{c}{1 (Thm \ref{the: indivisible_binary_chore})}\bigstrut\\\cline{3-5}
        ~ & ~ & MMS & \multicolumn{2}{c}{1 (Thm \ref{the: indivisible_binary_chore})}\bigstrut\\\cline{3-5}
        ~ & ~ & USC & \multicolumn{2}{c}{1 (Thm \ref{the: indivisible_binary_chore})}\bigstrut\\\cline{1-5}
         \multirow{3}{*}{ \makecell[c]{Supermodular\\ binary}}& \multirow{3}{*}{$n$}  & EF1  & \multicolumn{2}{c}{$\times$ (Thm \ref{the: supermodular_binary_negative_results})}   \bigstrut\\\cline{3-5}
         ~ & ~& MMS &  \multicolumn{2}{c}{$\times$ (Thm \ref{the: supermodular_binary_negative_results})}  \bigstrut\\\cline{3-5}
        ~ & ~ & USC & \multicolumn{2}{c}{$\times$ (Thm \ref{the: supermodular_binary_negative_results})}  \bigstrut\\\cline{1-5}
        \multirow{4}{*}{\makecell[c]{Additive personalized \\ bi-valued}}&\multirow{4}{*}{$2$}& EF1 & \makecell[c]{$2$\\(Thm \ref{the: ef1_bi_valued_impossiblity_results_chores})}&\makecell[c]{$2$\\(Thm \ref{the: two_agents_bivalued_ef1_mms_chores})}\bigstrut\\\cline{3-5} 
        % ~ & ~ & ~ & \multicolumn{2}{c}{$1^{\clubsuit}$ (Thm \ref{the: two_agents_bivalued_save_item_chores})}\bigstrut\\\cline{3-5}
         ~ & ~ & MMS & \makecell[c]{$\frac{3}{2}$\\(Thm \ref{the: ef1_bi_valued_impossiblity_results_chores})}& \makecell[c]{$\frac{5}{3}$\\(Thm \ref{the: two_agents_bivalued_ef1_mms_chores})}\bigstrut\\\cline{3-5}
          ~ & ~ & USC & \multicolumn{2}{c}{$\times$ (Prop \ref{prop: bivalued_no_efficiency_chores})}\bigstrut\\ \cline{1-5}
        \multirow{3}{*}{\makecell[c]{ Additive binary\\+ personalized bi-valued}}&\multirow{3}{*}{$n$}& EF1  & \multicolumn{2}{c}{1 (Thm \ref{the: binary_bivalued_chores}) }\bigstrut\\\cline{3-5}
        ~ & ~  & MMS & \multicolumn{2}{c}{1 (Thm \ref{the: binary_bivalued_chores}) }\bigstrut\\\cline{3-5}
        ~ & ~ & USC & \multicolumn{2}{c}{$\times$ (Prop \ref{prop: binary_bivalued_no_efficiency_chores})}\bigstrut\\
        \bottomrule
    \end{tabular} 
    \caption{Chores setting. LB: lower bound. UB: upper bound. "$\times$" indicates inapproximability.}
    \label{tab: summary_results_chores}
\end{table}

In the chores setting, we introduce the completeness constraints.
We first show the inapproximability of EF1 even for two agents with additive personalized tri-valued cost functions, which blocks our way of studying more general cost functions.
Next, we investigate the restricted cost functions, such as additive binary cost functions.
%We study additive binary and supermodular binary cost functions. Then, we consider the additive personalized bi-valued cost function and its combination with additive binary cost functions. More specifically,

\begin{itemize}[left=0pt]
     % \item When agents have tri-valued cost functions, we show that no deterministic algorithm can guarantee a complete allocation that is $\alpha$-EF1 for any $\alpha \geq 1$.

    \item For additive binary cost functions, we propose the Compelled-Greedy Algorithm (Algorithm \ref{alg: indivisible_suppermodular+binary}) to compute a complete allocation that simultaneously satisfies EF1, MMS, and min-USC.
    For supermodular binary cost functions, we prove that no deterministic algorithm can compute a complete allocation that is $\alpha$-EF1 (or $\alpha$-MMS or $\alpha$-min-USC) for any $\alpha \geq 1$.

    \item When two agents have additive personalized bi-valued cost functions, we design the Adapted Chores Envy-Graph Procedure (Algorithm \ref{alg: two_agents_ef1_mms_chores}) to compute an allocation that satisfies completeness, $2$-EF1, and $\frac{5}{3}$-MMS.
    Further, we show that no deterministic algorithm can guarantee a complete allocation that is $\alpha$-EF1 for any $\alpha<2$ or $\beta$-MMS for any $\beta<\frac{3}{2}$, and no deterministic algorithm can guarantee a complete allocation that is $\alpha$-EF1 (or $\beta$-MMS) and $\gamma$-min-USC for any $\alpha, \gamma \geq 1$ and $\beta<2$.

    \item When one agent has the additive personalized bi-valued cost function and the remaining agents have additive binary cost functions,
    % When $(n-1)$ agents have binary cost functions and one agent has the personalized bi-valued cost function, 
    we design the Adapted-Chores-Picking Algorithm (Algorithm \ref{alg: binary_bivalued_ef1_mms_chores}) to compute an allocation satisfying completeness, EF1, and MMS.
    Further, no deterministic algorithm can compute a complete allocation that is $\alpha$-EF1 (or MMS) and $\beta$-min-USC for any $\alpha, \beta \geq 1$.
\end{itemize}
% In the last section, we discuss two special online item settings. 
% One is that each arriving item can wait for some periods without being allocated immediately, which can be found in the online matching literature \cite{ashlagi2018maximum}.
% In this setting, we prove that an EF1 allocation for two agents with personalized bi-valued valuation/cost functions exists.
% The other one is the monotone online instance, which originates from the online MMS allocation with normalized valuations \cite{zhou2023multi}.
% In this setting, we show that the Marginal-Greedy Algorithm computes a non-wasteful EF1 allocation when items are goods, and the Round-Robin algorithm computes a complete EF1 allocation when items are chores.

\subsection{Related Work}
\paragraph{Online Fair Allocation.} 
In the majority of the literature, there are two kinds of online arrivals:
The former kind is that agents arrive over time \cite{kash2014no,friedman2015dynamic,friedman2017controlled,li2018dynamic},
and the latter is that items arrive over time, like \cite{aleksandrov2015online}, which is also the topic of our paper. 
%Our paper is closely related to the growing literature on online fair division; please see the survey of \cite{aleksandrov2020online} for an overview.
%There are two research directions in this area.
%One is the study of allocating static items to online agents \cite{kash2014no,friedman2015dynamic,friedman2017controlled,li2018dynamic}.
%The other one analyzes settings where the agents are static and the items arrive over time as we do.
Our work is mainly related to the study of
\cite{aleksandrov2015online} and \cite{hosseini2024class}, where items arrive one by one without any future information known in advance. 
\cite{aleksandrov2015online} consider the food bank problem when all agents only have the additive binary valuation functions and design two simple mechanisms to guarantee ex-ante envy-freeness.
In our paper, we study beyond additive binary valuation functions in the goods setting.
\cite{hosseini2024class} study the class fairness in online matching for indivisible and divisible items, and present a series of positive results.
The model of indivisible items in their paper can be seen as a special case of our model with submodular binary valuation functions in our paper, which will be discussed later.
In other words, we study a more general setting and derive the approximation results that strengthen their results.

\paragraph{Offline Fair Allocation.}
%For the offline fair division, please see the survey of \cite{amanatidis2023fair} for an overview.
When agents have additive valuation functions, there is a rich body of literature that discusses EF1 or MMS. Please refer to the survey by \cite{amanatidis2023fair} for an overview.
Here, we highlight some papers most relevant to our work on offline fair division with submodular binary valuation functions or supermodular binary cost functions.
When agents have submodular binary valuations in the goods setting,
\cite{benabbou2021finding} study the compatibility of fairness and efficiency and show that a maximum utilitarian social welfare allocation that is EF1 can be efficiently computed, and  \cite{barman2021existence} consider the fairness notion - MMS and prove that a maximum utilitarian social welfare allocation that satisfies MMS can be computed in polynomial time.
When agents have supermodular binary cost functions in the chores setting,
\cite{barman2023suppermodular} show that an EF1 or MMS allocation minimizing the social cost can be computed in polynomial time.

\section{Preliminaries}

For each natural number $s \in \mathbb{N}$, let $[s]= \{1, \ldots,s\}$.
Let $T = \{e_1, \ldots, e_{t}\}$ be the set of \textit{indivisible} items arriving one by one and $N =\{1,\ldots,n\}$ be the set of offline agents.
When an online item $e_t$ arrives,
its values or costs to all agents are revealed.
When items are goods, each agent is endowed with a valuation function $v_i: 2^{T} \rightarrow \mathbb{R}_{\geq 0}$.
%\textcolor{blue}{We assume that $v_i$ is normalized, that is $v_i(\emptyset)=0$.}
%We simplify notion by using $v_i(e)$ to denote $v_i(\{e\})$.
%Note that we define $v_i(\emptyset)=0$.
%A valuation function $v_i$ is additive if $v_i(S) = \sum_{e \in S}v_i(e)$ for any $S \subseteq T$, or called matroid rank if $v_i$ is submodular, where $v_i(A \cup \{e\}) -v_i(A) \geq v_i(B \cup \{e\})-v_i(B)$ for any $A \subseteq B \subseteq T$ and $e \notin B$, and the marginal value of each good is zero or one. 
When items are chores, each agent is endowed with a cost function $c_i: 2^{T} \rightarrow \mathbb{R}_{\geq 0}$. 
We assume that  $v_i(\emptyset)=0$ ($c_i(\emptyset)=0$), and monotone, $v_i(X)\le v_i(Y)$ ($c_i(X)\le c_i(Y)$) for all $X\subseteq Y\subseteq T$.
For any item $e\in T$ and any subset $X \subseteq T$, we denote by $\Delta^i_X(e)=v_i(X\cup\{e\})-v_i(X)$ or $c_i(X\cup\{e\})-c_i(X)$ the marginal value or cost of an item $e\in T$ in the set $X$.
%Given an item set $M$, we say that
%\begin{itemize} 
 $v_i$ is \textit{additive} if the value of any set of items $S$ is the sum of the value of each item in $S$, that is, $v_i(X)=\sum_{e\in X} v_i(e)$.
   $v_i$ is \textit{submodular} if $v_i(X)+v_i(Y)\ge v_i(X\cup Y)+v_i(X\cap Y)$ for any two sets $X,Y\subseteq T$.
  $v_i$ is \textit{supermodular} if $v_i(X)+ v_i(Y)\le v_i(X\cup Y)+v_i(X\cap Y)$ for any two sets $X,Y\subseteq T$. 
    % \item $f$ is \textit{bi-valued} if the marginal value of any item $e\in M$ in any set of items $S\subseteq M$ is $0$ or $1$, that is, $\Delta_e(S)\in \{a,b\}$, $a>b$.
%\end{itemize}
There are similar arguments for $c_i$, and we omit it.
Let $T^k$ be the set of items that have arrived by round $k \in [t]$.
% and $\Pi_{n}(T^k)$ be the set of all $n$-partitions of $T^k$.
We denote an allocation by ${\bf A}^k = (A_1^k, \ldots, A_n^k)$, where $A_i^k \subseteq T^k$ is the bundle of items that are allocated to agent $i$, and $A_i^{k} \cap A_j^k =\emptyset$ for any $i, j \in N$. Let $\Pi_{n}(T^k)$ be the set of all allocations of $T^k$.
Collectively, $(T, N, (v_i)_{i\in N})$ forms an instance of \textit{online goods fair allocation} and $(T, N, (c_i)_{i\in N})$ forms an instance of \textit{online chores fair allocation}.
For a singleton set $\{e\}$, we will use $v_i(e)$ ($c_i(e)$) as a shorthand for  $v_i(\{e\})$ ($c_i(\{e\})$). 

\paragraph{Deterministic online algorithms.}
% \footnote{Randomization significantly alleviates the difficulty of the online setting, e.g., if each agent has the same probability to get the allocated item, it is easy to achieve ex-ante fairness. Thus, we do not consider it in our paper.}
%In this paper, we study deterministic online algorithms, which must make decisions ``online" without knowing future information.
%More concretely, 
In the online setting, the items in $T$ arrive one by one in an arbitrary order, where $|T|$ is unknown.
When an item arrives, the values (costs) of this item for all agents are revealed.
The online algorithm must make
an immediate and irrevocable decision to allocate the item to one agent or discard it \textit{deterministically}.

Our objective is to design deterministic algorithms to compute allocations guaranteeing the desired approximation fairness or efficiency, including EF1, MMS, and USW/USC (Definitions \ref{def: ef1_goods}, \ref{def: mms_goods}, \ref{def: usw}, \ref{def: ef1_chores}, \ref{def: mms_chores}, and \ref{def: usc}), as well as constraints, including non-wastefulness (Definition \ref{def: non-wastefulness}) and completeness (Definition \ref{def: completeness}), at the end of each round.
% In short, this type of algorithm can process information one by one in series and guarantee the desired approximation (approximate EF1, MMS, USW, USC) and restriction (non-wastefulness and completeness) at the end of each step.
%Note that in our online algorithms design, we do not need to know in advance the number of items that will arrive, which means that showing the desired guarantees at the end implies that they hold at every step.
In contrast, our impossibility results will hold even if the desired guarantees are required to hold only at the end.
For a better understanding of our online model, let us demonstrate the deterministic online allocation procedure through the following example.

\begin{table}[tb]
\centering
    \begin{tabular}{c|c|c|c|c}
    \toprule
      & $e_1$ & $e_2$ & $e_3$ & $e_4$    \\
      \midrule
        agent $1$  & \textcolor{red}{6} & $8$& 10 & $6$  \\
        \midrule
        agent 2 & 13 & \textcolor{red}{8} & \textcolor{red}{5} & 0  \\
        \bottomrule
    \end{tabular} 
\caption{Example of the deterministic online allocation.}
\label{tab: det_example}

\end{table}

\paragraph{Example.}
Suppose that there are two agents who have additive valuation functions, and the arriving items are goods.
The value of each item is shown in Table \ref{tab: det_example}, and note that the value of each item is only revealed after it arrives.
Now we have a deterministic online algorithm named $\mathsf{OALG}$.
In round $t=1$, $\mathsf{OALG}$ allocates $e_1$ to agent 1, and the allocation ${\bf A}^1$ is EF1, MMS, and $\frac{6}{13}$-max-USW. 
In round $t=2$, $\mathsf{OALG}$ allocates $e_2$ to agent 2, and the allocation ${\bf A}^2$ is EF1, MMS, and $\frac{2}{3}$-max-USW.
In round $t=3$, $\mathsf{OALG}$ allocate $e_3$ to agent 2, and the allocation $\textbf{A}^3$ is $\frac{3}{4}$-EF1, $\frac{3}{5}$-MMS, and $\frac{19}{31}$-max-USW.
In round $t=4$, $\mathsf{OALG}$ discards $e_4$, and the allocation $\textbf{A}^4$ is $\frac{3}{4}$-EF1, $\frac{3}{7}$-MMS, and $\frac{19}{37}$-max-USW. 
Then, no further item arrives and we say that $\mathsf{OALG}$ achieves $\min\{\frac{3}{4},1\} = \frac{3}{4}$-EF1, $\min\{ \frac{3}{7}, \frac{3}{5}, 1\} = \frac{3}{7}$-MMS, and $\min\{\frac{6}{13}, \frac{19}{37},\frac{19}{31}, \frac{2}{3}\}$$=\frac{6}{13}$-max-USW.

\section{Allocation of Goods}
\label{sec: goods_setting}

Throughout this section, we focus on the goods setting.
We first introduce some definitions of
fairness and efficiency notions that are used in this section.

\begin{definition}[Envy-freeness up to One Good]
\label{def: ef1_goods}
For any $\alpha \in [0,1]$, an allocation ${\bf A}^{k}$ is $\alpha$-approximate envy-free up to one good ($\alpha$-EF1) if, for every pair of agents $i, i^{\prime} \in N$ with $A_{i^{\prime}}^k \neq \emptyset$, it holds that $v_i(A_i^k) \geq \alpha \cdot v_i(A_{i^{\prime}}^k \setminus \{e\})$ for some $e \in A_{i^{\prime}}^k$.  
\end{definition}
\begin{definition}[Maximin Share Fairness]
\label{def: mms_goods}
For any agent $i$, her maximin share $\mathsf{MMS}_{i}^k$ by round $t$ is defined as:
\[
\mathsf{MMS}_{i}^k = \max_{B^k \in \Pi_{n}(T^k)} \min_{i^{\prime} \in N}
v_i(B_{i^{\prime}}^k),
\]
%where $\Pi_{n}(T^k)$ is the set of allocations of $T^k$.
For any $\alpha \in [0, 1]$, an allocation ${\bf A}^k$ is $\alpha$-approximate maximin share fair ($\alpha$-MMS) if for any agent $i \in N$, it holds that $v_i(A_i^k) \geq \alpha \cdot \mathsf{MMS}_i^k$.
\end{definition}
\begin{definition}[Utilitarian Social Welfare]
\label{def: usw}
The utilitarian social welfare of allocation ${\bf A}^k$ is given by $USW({\bf A}^k) = \sum_{i \in N}v_i(A_i^k)$.
For any $\alpha \in [0,1]$, an allocation ${\bf A}^k$ is $\alpha$-max-USW
if $USW({\bf A}^k) \geq \alpha \cdot \max_{{\bf B}^k \in \Pi_{n}(T^k)}USW({\bf B}^k)$. %where $\Pi_{n}(T^k)$ is the set of allocations of $T^k$. 
\end{definition}

Besides the above definitions, we also pay attention to the following restriction, which is added in the goods setting.

% \begin{definition}[Non-Wastefulness for Goods]
% An allocation $A$ is non-wasteful (NW) if for any arriving good $e$, it is allocated to agent $i \in N$ where $\Delta_{e}^{i}(A_i) >0$.
% \end{definition}

\begin{definition}[Non-Wastefulness]
\label{def: non-wastefulness}
An allocation ${\bf A}^k$ is non-wasteful (NW) if (1) $\Delta_{A_i^k \setminus \{e\}}^i(e)>0$ for any agent $i\in N$ and any good $e \in A^k_i$, and (2) $\Delta_{A_i^k}^i({e}^{'})=0$ for any agent $i \in N$ and any discarded good $e'\in T^k$. 
\end{definition}
% Note that it is necessary to impose NW in the goods setting, as introduced in our online model, we allow some items to be discarded, and without this constraint, achieving EF1 would be trivial, by simply discarding all online items.
% Additionally, the existence of such allocations warrants further investigation.
%it is necessary to study its existence.
%When agents have additive valuation functions, for any arriving item, as long as it is allocated to one agent who values it positively, an NW allocation can be found.
Note that imposing NW in the goods setting is necessary, as our online model permits item discarding (see our example for illustration). Without this constraint, achieving EF1 would be trivial by simply discarding all online items.
For agents with additive valuation functions, an NW allocation can always be found as long as each arriving item is allocated to at least one agent who values it positively.
Below, we present a positive result where an NW allocation always exists for submodular binary valuation functions that we study in the subsequent parts, and complement it with a negative result for general submodular valuation functions.
%we show that for submodular binary valuation functions, where for each agent $i\in N$, her valuation function $v_i$ is submodular and the marginal value of any item is 0 or 1, an NW allocation always exists.
\begin{lemma}
\label{submodular_binary_NW}
For submodular binary valuation functions, an NW allocation always exists.
\end{lemma}

\begin{proof}
\begin{algorithm}[t]
\caption{Greedy Algorithm}
\label{alg: indivisible_submodular+binary_greedy}
\KwIn{An instance $(T, N, (v_i)_{i\in N})$ with submodular binary valuation functions}
\KwOut{A NW allocation ${\bf A}$}
Initialize: 
Let ${\bf A} = (\emptyset, \ldots, \emptyset)$;

\SetKwProg{Def}{when}{ do}{}
\Def{item $e \in T$ arrives}{

\For{$i=1$ to $n$}
{

\If{$\Delta_{A_{i}}^{i}({e})>0$}
{

 Let $A_{i}=A_{i}\cup \{e\}$;

 break;

}
}
}

\Return ${\bf A}$;

\end{algorithm}

Given an instance $(T, N, (v_i)_{i\in N})$, Algorithm \ref{alg: indivisible_submodular+binary_greedy} can compute an NW allocation.
 %An NW allocation $\textbf{A}$ can be found by the following way: 
%When item $e_{k}$ ($k\in [t]$) arrives, if there exists an agent $j \in N$ such that $\Delta_{e_k}^{j}(A_j^{k-1}) = 1$, we allocate it to agent $j$; else, we discard it.
When item $e_{k}$ ($k\in [t]$) arrives, if there exists an agent $j \in N$ such that $\Delta_{A_j^{k-1}}^{j}({e_k}) = 1$, $e_k$ will be allocated.
We assume that $e_{k}$ 
 is allocated to agent $i$.
We will show that $\Delta_{A_i^{r} \setminus \{e_r\}}^{i}({e_k})=1$ holds for any round $r > k$ after some new items are added to agent $i$'s bundle.
Without loss of generality, we assume that $e_s$ is an arbitrarily newly added item in agent $i$'s bundle, where $s > k$, it holds that
$\Delta_{A_i^{s-1}}^{i}({e_s}) = v_i(A_i^{s-1} \cup \{e_s\}) - v_i(A_i^{s-1}) = 1$.
 Then, we have 
 \begin{align*}
   \Delta_{{A_i^{s}} \setminus\{e_r\}}^{i}({e_k}) &= v_i(A_i^{s-1} \cup \{e_s\})-v_i(A_i^{s-1} \cup \{e_s\} \setminus\{e_k\})\\ &= v_i(A_i^{s-1}) + 1-v_i(A_i^{s-1} \cup \{e_s\} \setminus\{e_k\})\\
   &\geq v_i(A_{i}^{s-1})+1 - v_i(A_i^{s-1}\setminus\{e_k\})-v_i(e_s)\\
   &=v_i(A_i^{s-1}) - v_i(A^{s-1} \setminus \{e_k\})\\ &= \Delta_{A_i^{s-1} \setminus\{e_k\}}^{i}({e_k}),  
 \end{align*} where the inequality holds for submodularity.
 Following this induction by repeatedly utilizing submodularity, we have that $\Delta_{A_i^{s} \setminus\{e_k\}}^{i}({e_k}) \geq \ldots \geq \Delta_{A_i^{k}\setminus\{e_k\}}^{i}({e_k}) = \Delta_{A_i^{k-1}}^{i}({e_k}) = 1$.
 Since $v_i$ is a submodular binary valuation function, it holds that $\Delta_{A_i^{s} \setminus\{e_k\}}^{i}({e_k})\leq 1$.
 Thus, we can conclude that $\Delta_{A_i^{s} \setminus\{e_k\}}^{i}({e_k})= 1$. 
\end{proof}

\begin{lemma}
\label{no_submodular_NW}
For submodular valuation functions,
there exists an instance that no deterministic online algorithm can compute an NW allocation, even for two agents with identical valuation functions.
\end{lemma}
\begin{proof}  
We construct an instance with two agents, $v_1=v_2=v$.
When first item $e_1$ arrives, $v(e_1)=3$. W.l.o.g., we allocate  $e_1$ to agent 1. 
When the second item $e_2$ arrives, $v(e_2)=3, v(\{e_1,e_2\})=5$.
Case (1): we allocate $e_2$ to agent 2, for the third item $e_3$,$v(e_3)=5, v(\{e_1,e_3\})=v(\{e_2,e_3\})= v(\{e_1,e_2,e_3\})=5$. 
We have that $v(\{e_1,e_3\})-v(e_3)=v(\{e_2,e_3\})-v(e_3)=0$.
So, no matter who is assigned $e_3$, it cannot satisfy NW.
Case (2): we allocate $e_2$ to agent 1, for the third item $e_3$ with $v(e_3)=5, v(\{e_1,e_3\})=v(\{e_2,e_3\})= v(\{e_1,e_2,e_3\})=5$.  
To ensure NW, we must allocate $e_3$ to agent 2.
For the fourth item $e_4$, $v(e_4)=v(\{e_1,e_4\})=v(\{e_2,e_4\})= v(\{e_3,e_4\})=v(\{e_1,e_2,e_4\})=v(\{e_1,e_3,e_4\})=v(\{e_2,e_3,e_4\})=v(\{e_1,e_2,e_3,e_4\})=6$. 
We have that $v(\{e_1,e_2,e_4\})-v(\{e_2,e_4\})=v(\{e_3,e_4\})-v(e_4)=0$.
So, no matter who is assigned $e_4$, it cannot be NW.
\end{proof}

\subsection{The Impossibilities of Approximate EF1 and MMS
%Upper Bounds of Approximate EF1 and MMS
}
\citep{aleksandrov2015online} show that in the online setting, an EF1 allocation exists for additive binary valuations, where for each agent $i \in N$, $v_i$ is additive and $v_i(e) \in \{0,1\}$ for any $e \in T$.
In this part, we present a strong negative result regarding the approximations of EF1 and MMS 
for general additive valuations. 
%for any number of agents with tri-valued valuation functions.
We say that the valuation function $v_i$ is
%additive 
additive personalized tri-valued if $v_i$ is additive and $v_i(e)\in \{a_i, b_i,z_i\}$ for any $e\in T$, where $0<a_i\leq b_i\leq z_i$.

\begin{theorem}
\label{the: trivalued_ef1_mms_goods}
No deterministic online algorithm can compute an NW allocation that guarantees $\alpha$-EF1 or $\alpha$-MMS for any $\alpha > 0$, even for two agents with additive personalized tri-valued valuation functions.
\end{theorem}
\begin{proof}
\begin{table}[tb]
\centering
    \begin{tabular}{c|c|c|c}
    \toprule
      % & $t = 1$ & $t=2$ & $t = 3$   \\
      % \midrule
      & $e_1$ & $e_2$ & $e_3$   \\
      \midrule
        agent $1$  & \textcolor{red}{1} & $\frac{1}{\epsilon}$ &  $\frac{1}{\epsilon}$ \\
        \midrule
        agent 2 & $1$ & \textcolor{red}{$\epsilon$} &  $\frac{1}{\epsilon}$\\
        \bottomrule
    \end{tabular} 
\caption{The impossibility result of EF1 and MMS.}
    \label{tab: two_agents_impossibility}
\end{table}
Given any $\alpha>0$, let $\epsilon$ be a constant arbitrarily close to 0 such that $0<\epsilon<\alpha$.
Consider an online instance with two agents. For any agent $i \in N$ and item $e \in T$, $v_i(e)\in\{\epsilon, 1, \frac{1}{\epsilon}\}$, where $0<\epsilon <1$.
The items arrive online, and the valuation of each item is shown in Table \ref{tab: two_agents_impossibility}.
% When $t =1$, the arriving item has the value of 1 for both agents.
%     %Without loss of generality, 
%     W.l.o.g., agent 1 picks it.
%     When $t = 2$, the item that has the value of $\frac{1}{\epsilon}$ for agent 1 and $\epsilon$ for agent 2 arrives.
%     If it is allocated to agent 1, agent 2's bundle is empty, resulting $0$-EF1 or $0$-MMS.
%     In order to avoid the above case occurring, this item must be allocated to agent 2.
%     At last, in round $t = 3$, the arriving item has the value of $\frac{1}{\epsilon}$ for both agents.
%     In this case, if this item is allocated to agent 1, for agent 2, the approximation ratio of EF1 is $\epsilon$, and the approximation ratio of MMS is $\frac{\epsilon}{1+\epsilon}$;
%     if this item is allocated to agent 2, for agent 1, the approximation ratio of EF1 is $\frac{1}{\frac{1}{\epsilon}}=\epsilon$ and the approximation ratio of MMS is $\frac{1}{1+\frac{1}{\epsilon}}= \frac{\epsilon}{1+\epsilon}$.
   % Therefore, for any $\alpha > 0$, let $\epsilon < \alpha$, and we can conclude that no algorithm can guarantee an NW allocation that is $\alpha$-EF1 or $\alpha$-MMS for any given $\alpha > 0$.
%     This result can be generalized to more than two agents, and we defer the proof to Proposition \ref{prop: no_ef1_mms_three_agents} in the appendix.
When $t =1$, the item $e_1$ arrives, $T^1=\{e_1\}$.
    Without loss of generality, 
    suppose it is allocated to agent 1, $A^1_1=\{e_1\}, A^1_2=\emptyset$.   
    When $t = 2$, the item $e_2$ arrives, $T^2=\{e_1,e_2\}$.
    To maintain the fairness of allocation, we must allocate it to agent 2, i.e., $A^2_1=\{e_1\}, A^2_2=\{e_2\}$.
    Otherwise, agent 1 gets $\{e_1,e_2\}$ and agent 2 gets an empty set, $v_2(\emptyset)=0<v_2(\{e_1,e_2\}\setminus\{e\})$ for any $e\in \{e_1,e_2\}$ and $MMS_2(T^2)=\epsilon$.
    It is $0$-EF1 and $0$-MMS.
   When $t = 3$, item $e_3$ arrives, $T^3=\{e_1,e_2,e_3\}$. $MMS_1(T^3)=\frac{1}{\epsilon}$, $MMS_2(T^3)=1+\epsilon$.
    If we allocate it to agent 1, i.e., $A^3_1=\{e_1,e_3\}, A^3_2=\{e_2\}$,
    we have $v_2(A^3_2)=\epsilon$, $v_2(A^3_1\setminus\{e_1\})=\frac{1}{\epsilon}$, and $v_2(A^3_1\setminus\{e_3\})=1$.
    Thus,
    ${\bf A^3}$ is $\epsilon$-EF1 and $\frac{\epsilon}{1+\epsilon}$-MMS.
    If $e_3$ is allocated to agent 2, i.e., $A^3_1=\{e_1\}, A^3_2=\{e_2,e_3\}$, we have
    $v_1(A^3_1)=1=\epsilon\cdot v_1(A^3_2\setminus\{e\})$ for any $e\in A^3_2$.
    Thus,
    ${\bf A^3}$ is $\epsilon$-EF1 and $\epsilon$-MMS.
    
    Therefore, in the above instance, no algorithm can guarantee an NW allocation that is $\alpha$-EF1 or $\alpha$-MMS.
    This result can be generalized to more than two agents, and we defer the proof to Proposition \ref{prop: no_ef1_mms_three_agents}. 
    \end{proof}

\begin{proposition}
\label{prop: no_ef1_mms_three_agents}
 For general additive functions, no deterministic online algorithm can compute an NW allocation that guarantees $\alpha$-EF1 or $\alpha$-MMS for any $\alpha > 0$ when there are more than two agents with additive personalized tri-valued valuation functions.   
\end{proposition}

\begin{proof}
\begin{table}[tb]
\centering
    \begin{tabular}{c|c|c|c|c|c|c|c}
    \toprule
      & $e_1$ & $e_2$ & $\ldots$&$e_{n-2}$ & $e_{n-1}$ & $e_n$ & $e_{n+1}$   \\
      \midrule
        agent $1$  & \textcolor{red}{$\frac{1}{\epsilon}$} & $\frac{1}{\epsilon}$ &  $\ldots$ &$\frac{1}{\epsilon}$ & $1$ & $\epsilon$ & $\frac{1}{\epsilon}$ \\
        \midrule
       agent $2$  & $\frac{1}{\epsilon}$ & \textcolor{red}{$\frac{1}{\epsilon}$} &  $\ldots$&$\frac{1}{\epsilon}$ & $1$ & $\epsilon$ & $\frac{1}{\epsilon}$\\
       \midrule
       $\ldots$& $\ldots$ & $\ldots$ & $\ldots$ & $\ldots$ & 
       $\ldots$ &$\ldots$ & $\ldots$\\
       \midrule
       agent $n-2$  & $\frac{1}{\epsilon}$ & $\frac{1}{\epsilon}$ &  $\ldots$&\textcolor{red}{$\frac{1}{\epsilon}$} & $1$ & $\epsilon$ & $\frac{1}{\epsilon}$\\
       \midrule
       agent $n-1$  & $\frac{1}{\epsilon}$ & $\frac{1}{\epsilon}$ &  $\ldots$&$\frac{1}{\epsilon}$ & \textcolor{red}{$1$} & $\frac{1}{\epsilon}$ & $\frac{1}{\epsilon}$\\
       \midrule
       agent $n$  & $\frac{1}{\epsilon}$ & $\frac{1}{\epsilon}$ &  $\ldots$&$\frac{1}{\epsilon}$ & $1$ & \textcolor{red}{$\epsilon$} & $\frac{1}{\epsilon}$\\
        \bottomrule
    \end{tabular} 
\caption{The impossibility results of EF1 and MMS for additive personalized tri-valued valuations when there are more than two agents.}
    \label{tab: impossiblility_ef1_three_values}
\end{table}
    We further extend the impossibility results of approximate EF1 and MMS for $n\geq 3$.
    Consider the following instance with $n \geq 3$ agents, where the value of each item is shown in Table \ref{tab: impossiblility_ef1_three_values}. 
    In the first $(n-2)$ rounds, to avoid 0-EF1 or 0-MMS, each item is allocated to the agent with the empty bundle.
    For the item that arrives before round $t = n-2$, it has the value of $\frac{1}{\epsilon}$ for any agent.
    Without loss of generality, we assume that agent $i$ picks the item in the round $t = i$, where $i \leq n$.   
    In round $t = n-1$, the arriving item has the value of 1 for any agent, and we assume that it is allocated to agent $n-1$. 
    In round $t = n$, the arriving item has the value of $\epsilon$ for agent $i \in N \setminus \{n-1\}$ and $\frac{1}{\epsilon}$ for agent $n-1$, and it should be allocated to agent $n$ to avoid 0-EF1 or 0-MMS.
     Finally, in the round $t = n+1$, the item with the value of $\frac{1}{\epsilon}$ for each agent arrives.
     Due to non-wastefulness, this item must be allocated to one agent.
     If this item is allocated to agent $i \in N \setminus \{n-1, n\}$, it is not hard to see that the allocation is $\epsilon^{2}$-EF1 and $\frac{\epsilon}{1+\epsilon}$-MMS.   
     If this item is allocated to agent $n-1$, this allocation is $\epsilon$-EF1 and $\frac{\epsilon}{1+\epsilon}$-MMS.
     If this item is allocated to agent $n$, we can derive that the competitive ratio of EF1 is $\epsilon$ and MMS is $\epsilon$.   
     Combining the above cases, we can conclude that the upper bound of EF1 is $\epsilon$ and MMS is $\epsilon$.
     For given $\alpha > 0 $, let $\epsilon < \alpha$, we can conclude that no online algorithm can guarantee an NW allocation that is $\alpha$-EF1 or $\alpha$-MMS for any given $\alpha > 0$. 
\end{proof}

By Theorem \ref{the: trivalued_ef1_mms_goods}, we cannot make further progress in the approximations of EF1 and MMS for additive personalized tri-valued valuation functions and beyond in the online setting.
In the following parts, we focus on two classes of valuation functions that are more general than additive binary functions, including submodular binary and additive personalized bi-valued valuation functions, and study whether there are positive results about the approximations of EF1 and MMS.

\subsection{Submodular Binary Valuations}
% (Additive) binary valuation function has been studied in \cite{aleksandrov2015online}, where they show that an EF1 allocation always exists.
% Thus, in this part, we extend the binary valuation function to submodular binary valuation functions.
% A valuation function is submodular binary if it is both submodular and binary. 
This part focuses on submodular binary valuation functions, where for each agent $i \in N$, $v_i$ is submodular and the marginal value of any item is 0 or 1.
%In the online allocation problem, 
For a subclass of submodular binary valuation functions based on maximum (unweighted) bipartite matching, also known as assignment \cite{munkres1957algorithms} or OXS \cite{lehmann2001combinatorial} valuations,
\cite{hosseini2024class} show that online allocations satisfying $\frac{1}{2}$-EF1 and $\frac{1}{2}$-MMS exist. 

\textit{Limitation of \cite{hosseini2024class}}:
Their $\frac{1}{2}$-EF1/MMS guarantees are restricted to OXS valuations, where items must be assigned within predefined groups under rigid matching constraints. Each member in the group can only receive at most one item, and once someone gets an item, it cannot be changed. 
Thus, their algorithm can only handle the OXS valuation function and its subclasses.

% In their model, there are some groups, each of which includes some members, and each member can only receive one item at most.
% %Each member in the group can only receive a maximum of one project, and once a project is received, it cannot be changed.
% %Every group values a bundle $S$ via a matching of the items to its members that maximizes the cardinality of matching edges.
% When an item arrives, they greedily match the item to some member who likes the item and has never received any items.
% %who have a weight of $0$ or $1$ for each item. Every group values a bundle $S$ via a matching of the items to its members that maximizes the cardinality of matching edges.
% Due to the variety, complexity, and flexibility of the matroid and submodularity, the algorithm for assignment valuations in \cite{hosseini2024class} cannot work under the submodular functions. 
Moving beyond the online allocation model based on matching, we consider the general online allocation problem.
Building upon the effective greedy technique in the submodular optimization and the sequential allocating technique to guarantee fairness in \cite{hosseini2024class},
we propose the Marginal-Greedy Algorithm (Algorithm \ref{alg: indivisible_submodular+binary}), greedily assigning each arriving item to the agent who receives the maximum marginal value, for submodular binary valuation functions, and then show that the Marginal-Greedy Algorithm computes an NW allocation that is $\frac{1}{2}$-EF1, $\frac{1}{2}$-MMS, and $\frac{1}{2}$-max-USW. 
% Additionally, we conjecture that the Algorithm \ref{alg: indivisible_submodular+binary} will also perform well when we input subadditive binary valuation functions into Algorithm \ref{alg: indivisible_submodular+binary}, and leave this conjecture as an open question.

\begin{algorithm}[tb]
\caption{Marginal-Greedy Algorithm}
\label{alg: indivisible_submodular+binary}
\KwIn{An instance $(T, N, (v_i)_{i\in N})$ with submodular binary valuation functions}
\KwOut{An NW approximate EF1 and approximate MMS allocation ${\bf A}$}
Initialize: let $\pi=(\pi_1,\ldots, \pi_n)$ be an arbitrary order of $n$ agents;

Let ${\bf A} = (\emptyset, \ldots, \emptyset)$;

\SetKwProg{Def}{when}{ do}{}
\Def{item $e \in T$ arrives}{

\For{$i=1$ to $n$}
{

\If{$\Delta_{A_{\pi_i}}^{\pi_i}(e)>0$}
{

 Let $A_{\pi_i}=A_{\pi_i}\cup \{e\}$;

$\pi=(\pi_1,\ldots,\pi_{i-1},\pi_{i+1},\ldots,\pi_n,\pi_i)$;

 break;

}
}
}

\Return ${\bf A}$;
\end{algorithm}

In the analysis of results about submodular binary valuation functions, we effectively leverage submodularity and the properties of the matroid.
There exists a relationship of submodularity and matroid: a function is submodular binary iff it is a rank function of a matroid \cite{schrijver2003combinatorial}.

\begin{definition}
    A set system  $\mathcal{M} = (E, \mathcal{F})$, where $E$ is a ground set and $\mathcal{F}$ is a family of subset of $E$, is a matroid if	it satisfies the following properties:
(1) $\emptyset \in \mathcal{F}$;
(2) if  $D\in \mathcal{F}$ and $C\subseteq D$, then $C \in \mathcal{F}$;
(3) if $C, D\in \mathcal{F}$ and $|D|>|C|$, then there exists $ x \in D\setminus C$ such that $C + x \in \mathcal{F}$. 
The rank function of matroid $\mathcal{M}$, denoted by $r(\cdot)$, is defined by:
$r:2^T\rightarrow R$, $r(X)=\max\{|Y|:Y\subseteq T, Y\in \mathcal{F}\}$.
\end{definition}

\begin{theorem}
\label{the: indivisible_submodular+binary}
    For the deterministic allocation of indivisible goods with submodular binary valuation functions, the Marginal-Greedy Algorithm (Algorithm \ref{alg: indivisible_submodular+binary}) computes an NW allocation that satisfies $\frac{1}{2}$-EF1, $\frac{1}{2}$-MMS and $\frac{1}{2}$-max-USW.
\end{theorem}

\begin{proof}
Given an online fair allocation instance $I=(T, N, (v_i)_{i\in N})$, where $v_i$ is submodular and binary for all $i\in N$.
Fix an arbitrary round $k$, let $T^k=\{e_1,\ldots,e_k\}$ be the set of items that have already arrived, denoted by ${\bf A}^k=(A^k_1,\ldots, A^k_n)$ and $A^k_0$ the allocation and the set of all unallocated items at the end of round $k$ respectively.

In Algorithm \ref{alg: indivisible_submodular+binary}, the arriving item is allocated to an agent who thinks that the marginal value of this item is $1$ if there exists such an agent. By Lemma \ref{submodular_binary_NW}, $\textbf{A}^k$ is NW.
For each $v_i$, there exists a matroid $\mathcal{M}_i=(T^k,\mathcal{F}_i)$ such that $v_i$ is the rank function of $\mathcal{M}_i$. 
%Let $\mathcal{M}_i=(T,\mathcal{F}_i)$ be the corresponding matroid of $v_i$.
We have $v_i(A^k_i)=|A^k_i|$, and $A^k_i\in \mathcal{F}_i$.

First, we show that ${\bf A}^k$ is $\frac{1}{2}$-EF1. 
Consider any two agents $i,j\in N$, $i\ne j$.
If $v_i(A^k_j)>v_i(A^k_i)$, by the property (3) of the matroid, there must exist an item $e\in A^k_j$ such that $A^k_i+e\in \mathcal{F}_i$.
Let $B_j=\{e:e\in A^k_j, A^k_i+e\in \mathcal{F}_i\}$ and $B_j^*\in\arg \max\{|S|: S\subseteq B_j, A_i^k\cup S\in \mathcal{F}_i\}$. We will show $|B_j|\le |A^k_i|+1$. Assume that $|B_j|\ge |A^k_i|+2$, $j$ has the priority higher than $i$ at least $|B_j|\ge |A^k_i|+2$ times, which contradicts the order in the algorithm.
Combining the property (3) of the matroid, we have $v_i(A^k_j)=|A^k_i|+|B_j^*|\le |A^k_i|+|B_j|\le 2|A^k_i|+1$. 
Thus, ${\bf A}^k$ is $\frac{1}{2}$-EF1.

Second, we show that ${\bf A}^k$ is $\frac{1}{2}$-MMS. Consider a MMS partition ${\bf X}=(X_1,\ldots, X_n)$ for agent $i$ of item set $T$, we assume that the share of agent $i$ in ${\bf A}^k$ is less than half of the maximin share of $i$,  i.e. $v_i(A^k_i)\le \frac{1}{2}v_i(X_j)$ for any $j\in N$. 
By the property (3) of the matroid, there exists an item $e\in X_j$ such that $A^k_i+e\in \mathcal{F}_i$.
Let $C_j=\{e: e\in X_j, A^k_i+e\in \mathcal{F}_i\}$ and
$C_j^*\in\arg \max\{|S|: S\subseteq C_j, A_i^k\cup S\in \mathcal{F}_i\}$.
Then, we have $v_i(X_j)=v_i(A^k_i\cup C_j^*)=v_i(A^k_i)+v_i(C_j^*)=|A^k_i|+|C_j^*|$.
Since $v_i(A^k_i)=|A^k_i|\le \frac{1}{2}v_i(X_j)$ holds, 
we get $|C_j|\ge |C_j^*|>|A^k_i|+1$.
Then, summing up all agents, we have $|\cup_{j\in N}C_j|>n(|A^k_i|+1)$.
At last, we get 
$|\cup_{j\in N}C_j|=|\{e: e\in \cup_{j\in N}X_j, A^k_i+e\in \mathcal{F}_i\}|=|\{e: e\in T, A^k_i+e\in \mathcal{F}_i\}|=|\cup_{j\in N}B_j|\le (n-1)(|A^k_i|+1)$, which contradicts the inequality $|\cup_{j\in N}C_j|>n(|A^k_i|+1)$. Thus, ${\bf A}^k$ is $\frac{1}{2}$-MMS.

Finally, we show that ${\bf A}^k$ is $\frac{1}{2}$-USW. 
We make a slight adjustment to the Algorithm \ref{alg: indivisible_submodular+binary}: when a new item arrives, if the marginal utility is 0 for all agents, then assign the item to agent $1$ without updating the agents' order.
Then, we will show that
${\bf A}^{k'}=(A^k_1\cup A^k_0,\ldots, A^k_n)$ is the output of the adjusted algorithm.
We only need to prove that in each round $s\in [k]$, if the marginal value of the arrived item $e_s$ to $A^{s-1}_1$ is $1$, we have $v_i(A^{s-1}_1\cup\{e_s\})=v_i(A^{s-1}_1)+1$. 
By the monotonicity of $v_i$,
we get $v_i(A^{s-1}_1\cup A^{s-1}_0\cup\{e_s\})\ge v_i(A^{s-1}_1\cup\{e_s\})=v_i(A^{s-1}_1)+1=v_i(A^{s-1}_1\cup A^{s-1}_0)+1$, implying that the item $e_s$'s marginal value to $A^{s-1}_1\cup A^{s-1}_0$ is $1$.
%Let $r_i$ be the rank function of the matroid $\mathcal{M}_i$.
%We have $r_i(A^{s-1}_1\cup\{e_s\})=r_i(A^{s-1}_1)+1$.
%If $v_i(A^{s-1}_1\cup A^{s-1}_0\cup\{e_s\})-v_i(A^{s-1}_1\cup A^{s-1}_0)=0$, then $v_i(A^{s-1}_1\cup A^{s-1}_0\cup\{e_s\})\ge v_i(A^{s-1}_1\cup\{e_s\})=v_i(A^{s-1}_1)+1$.By the property (3) of matroid, there must exist an item $e'$ in $A^k_0$ such that $A^k_1\cup\{e'\}\in \mathcal{F}_1$. We assume that $e'$ arrives at $h$-round, where $h\le s-1$, then $A^{h-1}\cup\{e'\}\in \mathcal{F}_1$, which means $\Delta_{A^{h-1}}^1(\{e'\})=1$. So, in the $h$-round, $e'$ must be allocated to agent $1$ instead of going to $A^h_0$.
Thus, ${\bf A}^{k'}=(A^k_1\cup A^t_0,\ldots, A^k_n)$ is the output of the adjusted algorithm.
Since $v_1(A^k_1\cup A^t_0)=v_1(A^k_1)$, we have $USW({\bf A}^k)=USW({\bf A}^{k'})$. 
In other words, we only need to show that ${\bf A}^{k'}$ is $\frac{1}{2}$-max-USW.
The remaining proof part is based on the idea from \cite{lehmann2001combinatorial}.
Consider the first item $e_1$, 
we assume that $e_1$ is allocated to agent $i^*$ during the execution of the algorithm and $v_{i^*}(e_1)=1$, $A_{i^*}^1=\{e_1\}$ and $A_i^1=\emptyset$ for $i\ne i^*$.
Let ${\bf Y}=(Y_1,\ldots,Y_n)$ be an allocation achieving the optimal USW of instance $I=(T^k, N, (v_i)_{i\in N})$, and we assume that $e_1$ is allocated to agent $i^o$ in $Y$, $e_1\in Y_{i^o}$.
Next, we construct a new instance ${I^{(1)}}=(T',N, (v_i^{(1)})_{i\in N})$, 
where $v_{i^*}^{(1)}={\Delta^{i^*}_{e_1}}$ and $v_i^{(1)}=v_i$ for all $i\ne i^*$, $T'=T^k\setminus\{e_1\}$, that is, the first item arriving is $e_2$,
and then input the instance $I^{(1)}$ into the adjusted algorithm.
Let initialized order of $N$ be $\pi=(\pi_1,\ldots,\pi_n,\pi_{i^*})$, the initialized allocation be ${\bf D}^0=(D^0_1,\ldots,D^0_n)$ , where $D_{i^*}=\{e_1\}$ and $D_i=\emptyset$ for all $i\ne i^*$.
Assume that ${\bf D}^k=(D^k_1,\ldots, D^k_n)$ is the allocation at the end of round $k$.
We have ${\bf D}^k={\bf A}^{k'}$, then $USW({\bf D}^k)=USW({\bf A}^{k'})=USW({\bf A}^k)$.
Now, we construct an allocation ${\bf S}=(S_1,\ldots, S_n)$ for instance $I'$ as follows: $S_{i^o}=Y_{i^o}\setminus\{e_1\}$, $S_i=Y_i$ for all $i\ne {i^o}$.
It holds that $v_{i^o}^{(1)}(S_{i^o})=v_{i^o}(S_{i^o})=v_{i^o}(Y_{i^o}\setminus \{e_1\})= v_{i^o}(Y_{i^o})-v_{i^o}(\{e_1\})\ge v_{i^o}(Y_{i^o})-v_{i^*}(\{e_1\})$.
By the monotonicity of $v_{i^*}$,  we have $v_{i^*}^{(1)}(S_{i^*})=v_{i^*}(Y_{i^*}\cup \{e_1\})-v_{i^*}(\{e_1\})\ge v_{i^*}(Y_{i^*})-v_{i^*}(\{e_1\})$. 
 $v_{i^o}^{(1)}(S_{i^o})=v_{i^o}(Y_{i^o}\setminus \{e_1\})\ge v_{i^o}(Y_{i^o})-v_{i^o}(\{e_1\})\ge v_{i^o}(Y_{i^o})-v_{i^*}(\{e_1\})$, where the first inequality follows from the submodularity and the second inequality follows from $v_{i^o}(\{e_1\})\le v_{i^*}(\{e_1\})$.
Let ${\bf Y}^{(1)}=(Y_1^{(1)},\ldots,Y_n^{(1)})$ be an allocation achieving the optimal USW of instance $I^{(1)}$.
It holds that
$USW({\bf Y}^{(1)})\ge USW({\bf S})=\cup_{i\in N}v_i^{(1)}(S_i)=\cup_{i\in N, i\ne i^*,i\ne i^o}v_i^{(1)}(S_i)+v_{i^o}^{(1)}(S_{i^o})+v_{i^*}^{(1)}(S_{i^*})\ge USW({\bf Y})-2v_{i^*}(\{e_1\})$.
By lemma 1 in \cite{lehmann2001combinatorial}, $v_{i^*}^{(1)}$ is also submodular.
So, by constructing the new instances $I^{(2)},\ldots, I^{(k-1)}$ and inducting each round based on the items that arrive, we have 
$USW({\bf A}^{k})=USW({\bf A}^{k'})=USW({\bf D}^{k})=USW({\bf Y}^{k})\ge \frac{1}{2}USW({\bf Y}).$
\end{proof}

\begin{theorem}
\label{the: martroid_rank_negative_results}
    For any $\epsilon>0$, with submodular binary valuation functions, no deterministic online algorithm can achieve 
    \begin{itemize}
        \item $(\frac{1}{2}+\epsilon)$-EF1 and non-wastefulness,
        \item or $(\frac{1}{2}+\epsilon)$-MMS and non-wastefulness,
        \item or $(\frac{1}{2}+\epsilon)$-max-USW and non-wastefulness.
    \end{itemize}
\end{theorem}

\begin{proof}
 \begin{figure}
 \label{1/2ef1+1/2-mms}
     \centering
     \includegraphics[width=0.6\textwidth]{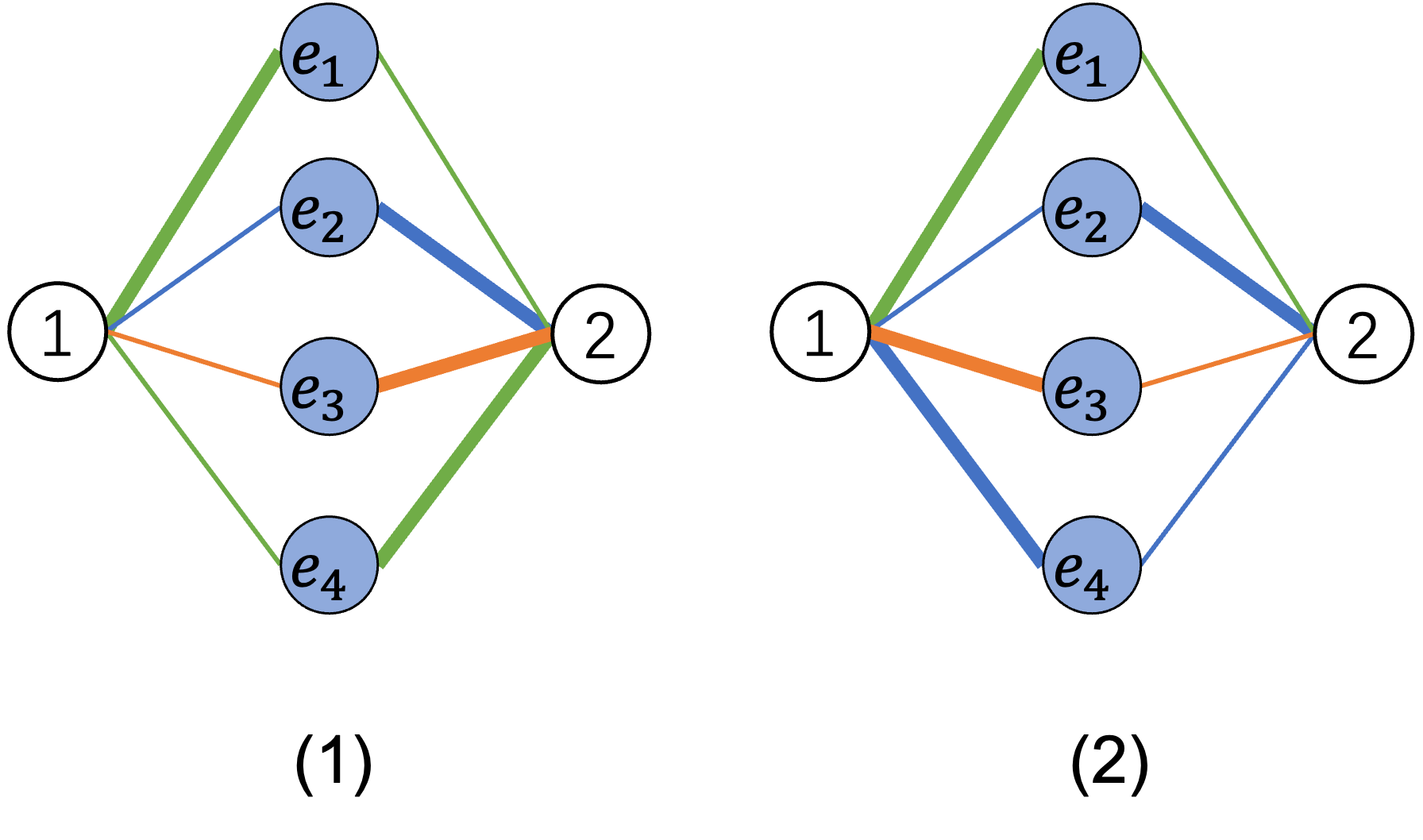}
     \caption{An instance with $\frac{1}{2}$-EF1 and $\frac{1}{2}$-MMS allocations}
     \label{fig1}
 \end{figure}
First, let us consider EF1 or MMS. Consider an online allocation instance, where there are two agents $\{1,2\}$ with identical valuations $v$, and the online items belong to several categories $C_1,\ldots, C_l$.
In Figure \ref{fig1}, %\ref{1/2ef1+1/2-mms}, 
categories are distinguished by colours, and edges of different colours represent items of different categories.
Denote a partition matroid $\mathcal{M}=(T,\mathcal{F})$, $\mathcal{F}=\{S: S\subseteq T, |S\cap C_k|\le 1, \text{for all}~ k\in [l]\}$. 
The valuation $v$ is the rank function of the matroid $\mathcal{M}$; for any item set $S$, the value of $S$ is the number of items of the maximal independent set of $S$. 
For the first arrived item $e_1$, the marginal values for agent $1$ and agent $2$ are $1$.
Without loss of generality, we allocate item $e_1$ to agent $1$. 
When the second item $e_2$ arrives, in order not to break fairness, we must allocate it to agent $2$.
For the third item $e_3$, the marginal values for agent $1$ and agent $2$ are $1$. 
We will discuss the following two cases.

\noindent\textbf{Case 1}: Allocate $e_3$ to agent $2$, shown in the left graph in Figure \ref{fig1} %\ref{1/2ef1+1/2-mms}. We must allocate the fourth item, $e_4$, to agent $2$ due to the requirement of non-wastefulness.
Then, $A^4_1=\{e_1\}$,$A^4_2=\{e_2,e_3,e_4\}$. 
We have that, 
$v(A^4_1)=\frac{1}{2}v(A^4_2\setminus\{e\})$ for any $e\in A^4_2$, and $v(A^4_1)=\frac{1}{2}\mathsf{MMS}_1$.

\noindent\textbf{Case 2}: Allocate $e_3$ to agent $1$, shown in the right graph in Figure \ref{fig1}; we must allocate the fourth item, $e_4$, to agent $1$ due to the requirement of non-wastefulness.
Then, $A^4_1=\{e_1,e_3,e_4\}$,$A^4_2=\{e_2\}$. 
We have that, 
$v(A^4_2)=\frac{1}{2}v(A^4_1\setminus\{e\})$ for any $e\in A^4_1$, and $v(A^4_2)=\frac{1}{2}\mathsf{MMS}_2$.

Next, let us consider USW. Consider an online allocation instance, where there are two agents $\{1,2\}$.
    For the first arrived item $e_1$, the marginal values for agents $1$ and $2$ are both $1$, $v_1(\{e_1\})=v_2(\{e_1\})=1$. 
    Without loss of generality, we allocate $e_1$ to agent $1$.
    When the second item $e_2$ arrives, $v_1(\{e_2\})=v_1(\{e_1,e_2\})=1$,$v_2(\{e_2\})=0$,
    the marginal values for agents $1$ and $2$ are both $0$.
    No matter which agent gets $e_2$, $USW(A^2)=1$. However, the max-USW is $2$ by allocating $e_1$ to agent $2$ and $e_1$ to agent $1$.
\end{proof}

 For an independent interest, we find that there is a side result of the Marginal-Greedy Algorithm, where it can also work for additive binary valuation functions to compute an NW allocation that is EF1, MMS, and max-USW, which demonstrates the robustness of our algorithm.
\subsection{Additive Personalized Bi-valued Valuations}

In this part, we consider the additive personalized bi-valued valuation function, which is often considered in the literature, such as \cite{ebadian2022fairly,aziz2023fair}. 
For any agent $i \in N$, her valuation function $v_i$ is additive personalized bi-valued if $v_i$ is additive and $v_i(e)\in \{a_i, b_i\}$ for any $e\in T$, where $0<a_i\leq b_i$.

It is worth investigating whether an NW allocation satisfying EF1 or MMS is achievable with additive personalized bi-valued valuation functions, as the binary case is a special type of bi-valued function.
However, we find that even for two agents with additive personalized bi-valued valuations, an NW allocation satisfying EF1 or MMS cannot be guaranteed.

\subsubsection{Two Agents}
% We study the setting of additive personalized bi-valued valuations for two agents and first establish the best possible non-trivial result for the approximation of EF1 and MMS simultaneously.
% To achieve this goal, we propose the Adapted Envy-Graph Procedure (Algorithm \ref{alg: two_agents_ef1_mms}), which involves a subtle case analysis since we do not know the value of the items in the future rounds.
We investigate the setting of additive personalized bi-valued valuations for two agents and establish the optimal non-trivial result for simultaneously approximating both EF1 and MMS.
To achieve this objective, we propose the Adapted Envy-Graph Procedure (Algorithm \ref{alg: two_agents_ef1_mms}), which requires sophisticated and subtle case analysis due to the uncertainty of future item values.

The Adapted Envy-Graph Procedure (AEGP) consists of three sub-algorithms: the Envy-Graph Procedure (EGP), the Preliminary-Breaking-Cycle (PBC), and the Deep-Breaking-Cycle (DBC), which are detailed in the Appendix.
The core structure of our algorithm employs the EGP as the primary procedure.
However, the EGP has a critical requirement: the input allocation is EF1 and must not contain any envy cycle.
% The Adapted Envy-Graph Procedure (AEGP) has three sub-algorithms, including the Envy-graph procedure (EGP), Preliminary-Breaking-Cycle (PBC), and Deep-Breaking-Cycle (DBC), which are deferred to the Appendix.
% The outline of our algorithm is that we utilize the EGP as the main procedure.
% However, it has a crucial requirement: the input allocation must be EF1 and cannot contain an envy cycle.
%it requires one precondition that the input allocation is EF1 without an ency cycle.

% When the envy cycle appears, in the following rounds, we execute the PBC and DBC algorithms as auxiliary algorithms to allocate the future items more effectively.
% As long as the output allocation of the PBC or DBC algorithm satisfies the precondition of the EGP algorithm, the EGP algorithm continues to be executed.
% The execution flow of this process is illustrated in Figure \ref{fig: cases_illustration}.
When an envy cycle is detected, we execute the PBC and DBC algorithms as auxiliary procedures in subsequent rounds to allocate future items more effectively.
Once the output allocation from either the PBC or DBC algorithm satisfies the EGP algorithm's precondition, the EGP algorithm resumes execution.
The execution flow of this integrated process is illustrated in Figure \ref{fig: cases_illustration}.

\begin{figure}[tb]
     \centering
\includegraphics[width=0.8\textwidth]{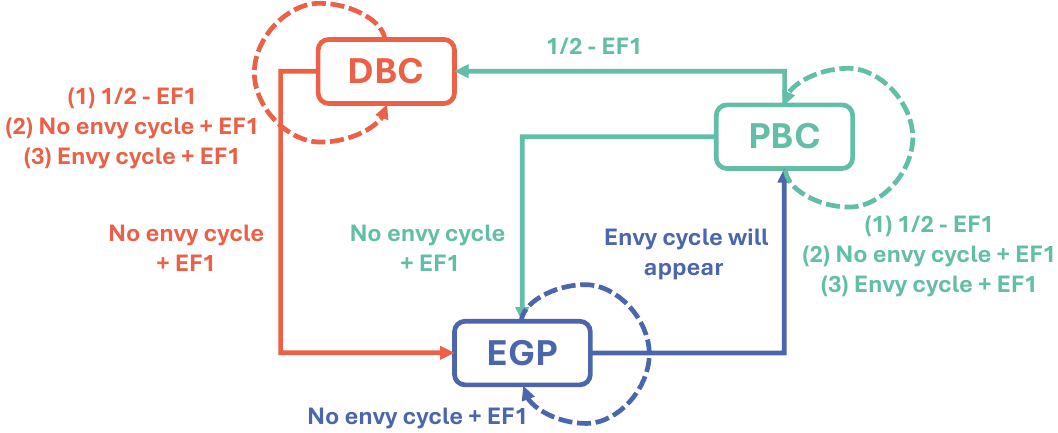}
     \caption{The execution flow of EGP, PBC, and DBC algorithms.  
     The dashed line represents the execution of the corresponding algorithm.
     The solid line means that when some specific cases occur, we execute another algorithm.}
     \label{fig: cases_illustration}
 \end{figure}

The main challenge lies in addressing envy cycles effectively.
This issue is straightforward to resolve in offline settings but becomes significantly more complex in online environments where reallocation is infeasible.
To overcome this obstacle, we first introduce the PBC algorithm, which is developed through careful analysis of item values over the next two rounds.
Specifically, the PBC algorithm can be executed consecutively for at most two rounds.
After these two rounds, the PBC algorithm successfully addresses most cases, achieving EF1 allocation without an envy cycle.
However, if the allocation remains be $\frac{1}{2}$-EF1, a more strategic approach is required in subsequent rounds to ensure fairer distribution.
To address this limitation, we introduce the DBC algorithm.
In the DBC algorithm, we establish allocation rules based on the parity of consecutive execution counts.
Unlike the PBC algorithm, the DBC algorithm can be executed consecutively across multiple rounds while maintaining the $\frac{1}{2}$-EF1 approximation guarantee, before its output allocation satisfies the precondition of the EGP.
In summary, both the PBC and DBC algorithms establish specific allocation rules for future items when an envy cycle emerges.
These rules not only ensure EF1 approximation throughout their execution but also continuously monitor whether the EGP precondition can be satisfied and actively work to achieve it when feasible.

% On the one hand, these rules guarantee the approximation of EF1 during their execution.
% On the other hand, these rules monitor whether the precondition of the EGP can be satisfied and achieve it if possible.

\begin{definition}[Envy Graph for Goods]
Given an allocation ${\bf A}^k$, the corresponding envy graph is defined as $G = (V,E)$, where the vertex set $V$ corresponds to agent set $N$, and a directed edge $(i,j)\in E$ iff agent $i$ envies agent $j$, i.e., $v_i(A_i^k) < v_i(A_j^k)$. Additionally, the directed cycle in the envy graph is called an envy cycle.
\end{definition}

 \begin{algorithm}[tb]
\caption{Adapted Envy-Graph Procedure}
\label{alg: two_agents_ef1_mms}
\KwIn{An instance $(T, N, (v_1,v_2))$ with additive personalized bi-valued valuation functions}
\KwOut{An NW approximate EF1 and approximate MMS allocation ${\bf A}$}

Let ${\bf A} = (\emptyset, \ldots, \emptyset)$; 

Initialize $\alpha = 1$; \tcp{Algorithm Indicator}

Initialize $(\lambda_1, \mu_1)$ and  $(\lambda_2, \mu_2) $, where $\lambda_1=\lambda_2=1 $ and $\mu_1= \mu_2 =0$; \tcp*{$\lambda$: the times of consecutive execution for the  PBC or DBC algorithm, and $\mu$: the case label in the PBC or DBC algorithm.}

Build an envy graph $G$ for two agents;

\SetKwProg{Def}{when}{ do}{}
\Def{item $e \in T$ arrives}{
\uIf{$\alpha = 1$}{

${\bf A}, \alpha \leftarrow \mathsf{EGP}(N, {\bf A}, \alpha, e)$;

\If{$\alpha=2$}{
${\bf A}, \alpha \leftarrow \mathsf{PBC}(N, {\bf A}, \alpha, e, (\lambda_1, \mu_1))$ \tcp*{Envy cycle will appear.}

}

}
\uElseIf{$\alpha = 2$}{
${\bf A}, \alpha, (\lambda_1, \mu_1) \leftarrow \mathsf{PBC}(N, {\bf A}, \alpha, e, (\lambda_1, \mu_1))$;

}
\ElseIf{$\alpha=3$}{${\bf A}, \alpha, (\lambda_2, \mu_2) \leftarrow  \mathsf{DBC}(N, {\bf A}, \alpha, e, (\lambda_2, \mu_2))$;
}

Update the envy graph $G$;

}
\Return ${\bf A}$;
\end{algorithm}

\begin{theorem}
\label{the: two_agents_bivalued_ef1_mms}
 For the deterministic allocation of indivisible goods for two agents with additive personalized bi-valued valuations, the AEGP  algorithm (Algorithm \ref{alg: two_agents_ef1_mms}) computes an NW allocation that satisfies $\frac{1}{2}$-EF1 and $\frac{1}{3}$-MMS.   
\end{theorem}
%Before we show the above Theorem, we list the following key lemmas about the approximation guarantees of three sub-algorithms in the AEGP.
The following two lemmas are crucial in the proof of Theorem \ref{the: two_agents_bivalued_ef1_mms}. 

\begin{algorithm}[t]
\caption{EGP (Envy-Graph Procedure)}
\label{alg: egp}
\KwIn{$N, A, \alpha, e$}
\KwOut{An NW allocation $\textbf{A}$ and $\alpha$}

\eIf{there is no envy between the two agents
}{
\eIf{$\exists i \in N$, $v_i(e) = b_i$ }{
Choosing the agent $i$ who values $e$ at $b_i$ (breaking ties arbitrarily); 
\label{line: choose the high value}
}
{
Arbitrarily choose one agent $i$;
}
}
{
Let $i \in N$ be an unenvied agent and $j = 2-i$;

\If{$v_i(A_i \cup \{e\})<v_i(A_j)$}
{
$\alpha=2$;

\Return $\textbf{A}$  and $\alpha$;
\label{line: terminated_jump}
}
}
$A_i = A_i \cup \{e\}$;

\Return $\textbf{A}$ and $\alpha$;
\end{algorithm}

\begin{lemma}
\label{lem: egp}
    The EGP algorithm (Algorithm \ref{alg: egp}) computes an NW allocation that is EF1 without an envy cycle.
\end{lemma}

\begin{proof}
    Note that there are three kinds of input allocation.
    The first one is the initial empty allocation, which is EF1 without an envy cycle.
    The second one is the output allocation of the PBC algorithm.
    The last one is the output allocation of the DBC algorithm. 
    It suffices to show that the last two kinds of input allocation are EF1 without an envy cycle.
    That is because if the input allocation of the EGP algorithm is EF1 without an envy cycle, by the description of the algorithm and induction, it is easy to check that the property of EF1 without an envy cycle is always maintained.
    By Lemma \ref{lem: pbc_dbc}, we can derive that when the PBC or DBC algorithm outputs $\alpha=1$, implying that the EGP algorithm is ready to be executed, the output allocation of the PBC or DBC algorithm is EF1 without an envy cycle.
    Therefore, our claim holds.   
\end{proof}

\begin{lemma}
\label{lem: pbc_dbc}
    The PBC or DBC algorithm (Algorithm \ref{alg: two_agents_cycle} or \ref{alg: two_agents_cycle_cases}) computes an NW allocation satisfying 
    \begin{itemize}
        \item $\frac{1}{2}$-EF1,
        \item or EF1 with an envy cycle,
        \item or EF1 without an envy cycle.
    \end{itemize}
    % (1) $\frac{1}{2}$-EF1, (2) EF1 with an envy cycle, or (3) EF1 without an envy cycle.
\end{lemma}
% Now we are ready to prove Theorem \ref{the: two_agents_bivalued_ef1_mms}.
We split the proof of Lemma \ref{lem: pbc_dbc} into two parts, where we show that the claim holds for both the PBC and DBC algorithms.
\begin{observation}
\label{obs: pre_breaking_cycle}
    When the PBC algorithm (Algorithm \ref{alg: two_agents_cycle})  is executed, for the first arriving item $e$, we have $v_i(e) =a_i$.
\end{observation}
\begin{proof}
    The reason why the PBC algorithm is executed is that for the arriving item $e$, if we add it to agent $i$'s bundle, there will be an envy cycle between two agents. i.e., $v_i(A_i \cup \{e\}) < v_i(A_j)$ and $v_j(A_j)<v_i(A_i \cup \{e\})$.
    Suppose for the contradiction that $v_i(e) = b_i$.
    Since the input allocation $A$ is EF1, for agent $i$, we have $v_i(A_j) -v_i(A_i) \leq b_i$.
    Thus, allocating item $e$ to agent $i$ eliminates the envy and cannot lead to an envy cycle, which is a contradiction.
    Therefore, for the arriving item $e$, it holds that $v_i(e)=a_i$.
\end{proof}

\begin{observation}
     When the PBC algorithm (Algorithm \ref{alg: two_agents_cycle})  is executed, we have $|A_i| \geq 1$ and $|A_j|\geq 1$.
\end{observation}

\begin{proof}
    Suppose that $|A_i|=0$.
    In that case, we have $|A_j|=1$.
    Otherwise, there would be an envy cycle.
    After allocating $e$, there will be an envy cycle, i.e., $v_i(e)=a_i<v_i(A_j)=b_i$ and $v_j(A_j)<v_j(e)$.
    If the second inequality holds, we have $v_j(A_j)=a_j$ but $v_i(A_j)=b_i$, which contradicts the line \ref{line: choose the high value} in the EGP algorithm (Algorithm \ref{alg: egp}).
    Suppose that $|A_j|=0$.
    In that case, we have $|A_i|=0$, which contradicts the fact that after allocating $e$, an envy cycle appears. \end{proof}

\begin{lemma}
\label{lem: preliminary_breaking_cycle}
   The PBC algorithm (Algorithm \ref{alg: two_agents_cycle}) computes an NW allocation satisfying:
   \begin{itemize}
       \item $\frac{1}{2}$-EF1,
       \item or EF1 with an envy cycle,
       \item or EF1 without an envy cycle.
   \end{itemize}
\end{lemma}
\begin{proof}
When the PBC algorithm is implemented ($\alpha=2$), the input allocation $\textbf{A}=(A_1, A_2)$ output by the EGP algorithm (Algorithm \ref{alg: egp}) is EF1 without an envy cycle.
By the description of the PBC algorithm, it is executed in at most two consecutive rounds.

\bm{$\lambda_1 = 1$}. 
For the first arriving item $e_1$, by Observation \ref{obs: pre_breaking_cycle}, we have $v_i(e_1) = a_i$.
If $\lvert A_j \rvert = 1$, we have $v_i(A_j)=b_i$ since agent $i$'s bundle is EF1.
Then, agent $j$ picks this item.
For agent $i$, her bundle is EF1 since we have $\lvert A_i \lvert \geq 1$ and $v_i(e_1)=a_i$.
For agent $j$, her new bundle is EF since she does not envy agent $i$.
Thus, this allocation is EF1 without an envy cycle, and the algorithm is terminated.
If $\lvert A_j \rvert > 1$, let us consider two cases about the valuation of item $e_1$.
\begin{itemize}
    \item If $v_j(e_1) = b_j$, item $e$ is allocated to agent $j$.
    Then, we set $\lambda_1=2$ and $\mu_1 = 1$.
In this case, for agent $j$, her new bundle is EF since $v_j(A_j \cup \{e_1\})> v_j(A_j)\geq v_j(A_i)$ and for agent $i$, her bundle is $\frac{1}{2}$-EF1, i.e, $v_i(A_i)\geq \frac{1}{2}v_i(A_j \cup \{e_1\} \setminus \{g\})$ holds for $g \in A_j \cup \{e_1\}$,  since we have $v_i(A_i) \geq v_i(A_j \setminus \{g\})$ for some $g \in A_j$ and $v_i(A_i) \geq a_i=v_i(e_1)$.
Thus, this allocation is $\frac{1}{2}$-EF1.
\item If $v_j(e_1)=a_j$, item $e$ is allocated to agent $i$.
Then, we set $\lambda_1=2$ and $\mu_1 = 2$.
In this case, for agent $i$, her new bundle is still EF1 since she is an unenvied agent in Algorithm \ref{alg: two_agents_ef1_mms} and the input allocation $A$ is EF1, and for agent $j$, her bundle is EF1 since we have $v_j(A_j)\geq v_j(A_i)$ but $v_j(A_j)< v_j(A_j \cup \{e_1\})$.
Thus, this allocation is EF1 with an envy cycle.
\end{itemize}

It is clear that any possible allocation from the above cases does not satisfy EF1 without an envy cycle.
Thus, we keep $\alpha=2$.

\bm{$\lambda_1 = 2$}.
The second item $e_2$ arrives.
Let us consider two cases about $\mu_1$.
\begin{itemize}
    \item If $\mu_1 = 1$, agent $i$ picks it.
In this case, for agent $j$, her bundle is still EF since we have $v_j(A_j)\geq v_j(A_i)$ and $v_j(e_1)=b_j\geq v_j(e_2)$.
For agent $i$, her new bundle is EF1 since we have $v_i(A_i)\geq v_i(A_j \setminus\{g\})$ for some $g \in A_j$ and $v_i(e_2)\geq v_i(e_1)=a_i$.
Thus, this allocation is EF1 without an envy cycle, and we set $\alpha=1$, $\lambda_1=1$, and $\mu_1=0$.
\item If $\mu_1 = 2$, let us consider two cases.
When $v_i(e_2)=a_i$, agent $j$ picks it.
For agent $i$, her bundle is still EF1 since we have $v_i(A_i) \geq v_i(A_j \setminus \{g\})$ for some $g\in A_j$ and $v_i(e_1)=v_i(e_2)=a_i$.
For agent $j$, her new bundle is EF because $v_j(A_j) \geq v_j(A_i)$ and $v_j(e_2)=v_j(e_1)=a_j$.
Hence, the allocation is EF1 without an envy cycle, and we set $\alpha=1$, $\lambda_1=1$, and $\mu_1=0$. 
When $v_i(e_2)=b_i$, it is allocated to agent $i$.
For agent $i$, her new bundle is EF since it holds that $v_i(A_i \cup \{e_1,e_2\}) = v_i(A_i)+a_i+b_i \geq v_i(A_j \setminus \{g\})+v_i(g)+a_i$ for some $g \in A_j$ where the inequality follows from the fact that $v_i(g)\in\{a_i, b_i\}$.
For agent $j$, her bundle is $\frac{1}{2}$-EF1, i.e., $v_j(A_j) \geq \frac{1}{2}v_j(A_i \cup\{e_1,e_2\}\setminus\{g\})$ holds for some $e \in A_i\cup \{e_1,e_2\}$, since $v_j(A_j)\geq v_j(A_i)$ and $v_j(A_j) \geq v_j(e_1) =\min \{v_j(e_1),v_j(e_2)\}$.
Thus, this allocation is $\frac{1}{2}$-EF1, and we set $\alpha=3$, $\lambda_1=1$, and $\mu_1=0$.
\end{itemize}
Combining the above cases, it is not hard to see that our claim is correct.
\end{proof}

\begin{observation}
\label{obs: deep_breaking_cycle}
    When the DBC algorithm (Algorithm \ref{alg: two_agents_cycle_cases}) is executed, there are two properties about input allocation $\bf A$:
    \begin{itemize}
        \item $v_i(A_i) \geq v_i(A_j)+a_i$;
        \item Let $e_1$ and $e_2$ denote two new allocated items in the PBC algorithm, we have $v_j(A_j) \geq \frac{1}{2} v_j(A_i \setminus \{e\})$, where $e \in \arg\max_{ g\in \{e_1,e_2\}}v_j(g)$.
    \end{itemize}
\end{observation}
\begin{proof}
When the DBC algorithm is ready to be executed, the PBC algorithm is terminated for $\lambda_1=\mu_1 = 2$. 
Then, by the proof of Lemma \ref{lem: preliminary_breaking_cycle}, it is easy to see that the statements are correct.
\end{proof}

\begin{lemma}
\label{lem: deep_breaking_cycle}
   The DBC algorithm (Algorithm \ref{alg: two_agents_cycle_cases}) computes an NW allocation satisfying
   \begin{itemize}
       \item $\frac{1}{2}$-EF1,
       \item or EF1 with an envy cycle,
       \item or EF1 without an envy cycle.
   \end{itemize}
\end{lemma}
\begin{proof}
By Observation \ref{obs: deep_breaking_cycle}, the input allocation $\textbf{A}=(A_1, A_2)$ is $\frac{1}{2}$-EF1, where $v_i(A_i)\geq v_i(A_j) +a_i$ and $v_j(A_j) \geq \frac{1}{2}v_j(A_i\setminus\{e\})$ for some $e \in A_i$.
Unlike the PBC algorithm, this algorithm may be executed in many rounds.
Let $g_1$ and $g_2$ denote the two online items that are allocated to agent $i$ in the PBC algorithm and we have $\min\{v_j(g_1),v_j(g_2)\}=a_j$.
Currently, we have $\alpha=3$.

\bm{$\lambda_2 =1$}. For the arriving item $e_1$ when $\lambda_2 = 1$, let us consider three cases of agents' valuations.
\begin{itemize}
    \item When $v_i(e_1) = a_i$, agent $j$ picks it.
    There are two subcases about $v_j(e_1)$. If $v_j(e_1)=b_j$, for agent $i$, her bundle is EF since we have $v_i(A_i) \geq v_i(A_j)+a_i=v_i(A_j \cup \{e_1\})$.
For agent $j$, her new bundle is EF1 since we have $v_j(A_j) \geq v_j(A_i \setminus \{g_1,g_2\})$ and $v_j(e_1)=b_j \geq \max\{v_j(g_1),v_j(g_2)\}$.
Thus, this allocation is EF1 without an envy cycle.
If $v_j(e_1)=a_j$, for agent $i$, her bundle is EF since we have $v_i(A_i) \geq v_i(A_j)+a_i=v_i(A_j \cup \{e_1\})$.
For agent $j$, her new bundle is EF1 since 
we have $v_j(A_j) \geq v_j(A_i \setminus \{g_1,g_2\})$ and $v_j(e_1)=a_i=\min\{v_j(g_1),v_j(g_2)\}$.
Thus, this allocation is EF1 without an envy cycle. 
Combining the above two cases, we set $\alpha=1$, $\lambda_2=1$, and $\mu_2=1$, implying that the EGP algorithm is ready to be executed.
\item When $v_i(e_1)=b_i$ and $v_j(e_1)=b_j$, it is allocated to agent $j$ and we set $\mu_2 = 1$.
For agent $i$, her bundle is EF1 since after removing $e_1$ from $A_j$, we have $v_i(A_i)\geq v_i(A_j)+a_i$.
For agent $j$, her new bundle is also EF1 since we have $v_j(A_j) \geq v_j(A_i \setminus \{g_1,g_2\})$ and $v_j(e_1)=b_j \geq \max\{v_j(g_1),v_j(g_2)\}$.
Thus, this allocation is EF1 with an envy cycle, and we set $\lambda_2=2$.
\item When $v_i(e_1)=b_i$ and $v_j(e_1)=a_j$, it is allocated to agent $i$ and we set $\mu_2 = 2$.
For agent $i$, her new bundle is clearly EF since $v_i(A_i \cup \{e_1\})>v_i(A_i) \geq v_i(A_j)+a_i$.
For agent $j$, her bundle is $\frac{1}{2}$-EF1 since we have$\frac{v_j(A_j)}{v_j(A_i \setminus \{g_1,g_2\})+v_j(e_1)+\min \{v_j(g_1), v_j(g_2)\}}\geq\frac{v_j(A_j)}{v_j(A_j)+2a_j}\geq \frac{1}{2}$, where the inequalities follow from the fact that $v_j(A_j) \geq v_j(A_i \setminus \{g_1,g_2\})$, $\lvert A_j\rvert \geq 2$ and $v_j(e_1)=\min \{v_j(g_1), v_j(g_2)\}=a_j$.
Thus, this allocation is $\frac{1}{2}$-EF1,  and we set $\lambda_2=2$.
\end{itemize}

\bm{$\lambda_2 = 2$}.
For the arriving item $e_2$ when $\lambda_2 = 2$, let us consider two cases about $\mu_2$.
\begin{itemize}
%     \item When $\mu_2 = 1$, item $e_2$ is directly allocated to agent $j$.
% For agent $i$, her bundle is EF1 since after removing $e_2$ from $A_j$, we have $v_i(A_i) \geq v_i(A_j \cup \{e_1\})=v_i(A_j)+a_i$.
% For agent $j$, her new bundle is EF since we have $v_j(A_j)\geq v_j(A_i \setminus \{g_1,g_2\})$, $v_j(e_1)=b_j\geq v_j(g_2)$ and $v_j(e_2)\geq v_j(g_1)=a_j$.
% Thus, this allocation is EF1 without an envy cycle, and we set $\alpha=1$, $\lambda_2=1$,  and $\mu_2=0$, implying the EGP algorithm is ready to be executed.
\item When $\mu_2 = 1$, we need to consider two subcases.
If $v_i(e_2)=b_i$, agent $i$ picks it.
For agent $i$, her new bundle is EF since we have $v_i(A_i) \geq v_i(A_j)+a_i$ and $v_i(e_2)=b_i=v_i(e_1)$.
For agent $j$, her bundle is $\frac{1}{2}$-EF1 since we have $v_j(A_j) \geq \frac{1}{2}v_j(A_i \setminus \{g\}) $ for some $g \in A_j$ and $v_j(e_1)=b_j\geq v_j(e_2)$.
Thus, this allocation is $\frac{1}{2}$-EF1 and satisfies the properties in Observation \ref{obs: deep_breaking_cycle}. 
Additionally, we set $\lambda_2=1$, $\mu_2=0$, and keep $\alpha=3$, implying that the DBC algorithm continues to be executed.
If $v_i(e_2)=a_i$, it is allocated to agent $j$.
For agent $i$, her bundle is EF1 since after removing $e_1$ from $A_j$,we have $v_i(A_i) \geq v_i(A_j)+a_i=v_i(A_j \cup \{e_2\})$.
For agent $j$, her new bundle is trivially EF since we have $v_j(A_j)\geq v_j(A_i \setminus \{g_1,g_2\})$, $v_j(e_1)=b_j \geq v_j(g_2)$ and $v_j(e_2)\geq a_i =v_j(g_1)$.
Thus, this allocation is EF1 without an envy cycle, and we set $\alpha=1$, $\lambda_2=1$,  and $\mu_2=0$, implying the EGP algorithm is ready to be executed.

% \item When $\mu_2 = 3$, let us also consider two subcases.
% If $v_j(e_2) = b_j$, agent $j$ picks it.
% For agent $i$, her bundle is EF1 since after removing $e_2$ from $A_j$, we have $v_i(A_i) \geq v_i(A_j \cup \{e_1\})=v_i(A_j)+a_i$.
% For agent $j$, her new bundle is EF since $v_j(A_j) \geq v_j(A_i \setminus\{g_1,g_2\})$, $v_j(e_2)=b_j\geq v_j(g_2)$ and $v_j(e_1)=a_j=v_j(e_2)$.
% Thus, this allocation is EF1 without an envy cycle, and we set $\alpha=1$, $\lambda_2=1$,  and $\mu_2=0$, implying the EGP algorithm is ready to be executed.
% If $v_j(e_2)=a_j$, agent $i$ picks it.
% For agent $i$, her new bundle is EF since we have $v_i(A_i) \geq v_i(A_j)+a_i$ and $v_i(e_2)=a_i=v_i(e_1)$.
% For agent $j$, her bundle is $\frac{1}{2}$-EF1 since we have $v_j(A_j)\geq \frac{1}{2}v_j(A_i \setminus \{g\})$ for some $g \in A_i$ and $v_j(e_1)=a_j=v_j(e_2)$.
% Thus, this allocation is $\frac{1}{2}$-EF1 and satisfies the properties in Observation \ref{obs: deep_breaking_cycle}. 
% Additionally, we set $\lambda_2=1$, $\mu_2=0$, and keep $\alpha=3$, implying that the DBC algorithm continues to be executed.
\item When $\mu = 2$, it is directly allocated to agent $j$.
For agent $i$, her bundle is EF since $v_i(A_i)\geq v_i(A_j)+a_i$ and $v_i(e_1)=b_i\geq v_i(e_2)$.
For agent $j$, her new bundle is $\frac{1}{2}$-EF1 since $v_j(A_j) \geq \frac{1}{2}v_j(A_i \setminus \{g\})$ for some $g \in A_i$ and $v_j(e_2)\geq v_j(e_1)=a_j$.
Thus, this allocation is $\frac{1}{2}$-EF1 and satisfies the properties in Observation \ref{obs: deep_breaking_cycle}. 
Additionally, we set $\lambda_2=1$, $\mu_2=0$, and keep $\alpha=3$, implying that the DBC algorithm continues to be executed.
\end{itemize}
Combining the above cases, our claim holds.
\end{proof}

\begin{algorithm}[t]
\caption{PBC (Preliminary-Breaking-Cycle)}
\label{alg: two_agents_cycle}
\KwIn{$N, A, \alpha, e, (\lambda_1, \mu_1)$}
\KwOut{An NW allocation $\textbf{A}$, $\alpha$, and $(\lambda_1, \mu_1)$}

\eIf {$\lambda_1 = 1$}{

\eIf{$\lvert A_j \rvert =1$}{

$A_j = A_j \cup \{e\}$;

$\alpha=1$, $\lambda_1=1$, and $\mu_1=0$;
}
{

\eIf{$v_j(e) = b_j$}{

$A_j = A_j \cup \{e\}$;

$\lambda_1 = 2$ and $\mu_1 = 1$;
}
{
$A_i = A_i\cup \{e\}$;

$\lambda_1 = 2$ and $\mu_1 = 2$;
}

}

}
{

\eIf{$\mu_1 = 1$ 
}{

$A_i = A_i \cup \{e\}$;

}
{

\eIf{$v_i(e) = a_i$}{

$A_j = A_j \cup \{e\}$;

}
{

$A_i = A_i \cup \{e\}$;

$\alpha=3$, $\lambda_1=1$, and $\mu_1=0$;

}
}

$\alpha=1$, $\lambda_1=1$, and $\mu_1=0$;

}

\Return $\textbf{A}$, $\alpha$ and $(\lambda_1, \mu_1)$;
\end{algorithm}

 \begin{algorithm}[tb]
\caption{DBC (Deep-Breaking-Cycle)}
\label{alg: two_agents_cycle_cases}
\KwIn{$N, A, \alpha, e, (\lambda_2, \mu_2)$}
\KwOut{An NW allocation $\textbf{A}$, $\alpha$, and $(\lambda_2, \mu_2)$}

\uIf{$\lambda_2 \mod  2 =1$}{

\uIf{$v_i(e) = a_i$}{
$A_j = A_j \cup \{e\}$;

$\alpha=1$, $\lambda_2=1$, and $\mu_2=0$;

}
\uElseIf{$v_i(e) = b_i$ and $v_j(e) = b_j$}{

$A_j = A_j \cup \{e\}$;

$\lambda_2 = \lambda_2+1$ and $\mu_2 = 1$;

}
\Else{
$A_i = A_i \cup \{e\}$;

$\lambda_2 = \lambda_2+1$ and $\mu_2 = 2$;

}
}
\ElseIf{$\lambda_2 \mod  2 =0$}{
\uIf{ 
$\mu_2 = 1$
}{
\eIf{$v_i(e) =b_i$}{

$A_i = A_i \cup \{e\}$;

$\lambda_2 = \lambda_2+1$ and $\mu_2=0$;
}
{
$A_j = A_j \cup \{e\}$;

$\alpha=1$, $\lambda_2=1$, and $\mu_2=0$;

}
}
\Else{

$A_j = A_j \cup \{e\}$;

$\lambda_2 = \lambda_2+1 $ and $\mu_2=0$;

}
}
\Return $\textbf{A}$, $\alpha$, $(\lambda_2, \mu_2)$;

\end{algorithm}

\noindent\textit{Proof of Theorem \ref{the: two_agents_bivalued_ef1_mms}.}
%\begin{proof}[Proof of Theorem \ref{the: two_agents_bivalued_ef1_mms}]
In the AEGP, it can be seen that for different values of the algorithm indicator $\alpha$, we need to choose the corresponding algorithm, i.e., $\alpha=1$: EGP, $\alpha=2$: PBC, or $\alpha=3$: DBC, to execute. 
By Lemmas \ref{lem: egp} and \ref{lem: pbc_dbc}, it can be seen that the output allocation is NW and $\frac{1}{2}$-EF1.
Regarding MMS, by the implication between EF1 and MMS \cite{amanatidis2018comparing}, i.e., $\alpha$-EF1 $\Rightarrow$ $\frac{\alpha}{(n-1)\alpha+1}$-MMS, where $n$ is the number of agents, it is clear that the output allocation is $\frac{\frac{1}{2}}{\frac{1}{2}+1}=\frac{1}{3}$-MMS.\qed
%\end{proof}

For an independent interest, in the Appendix, we show that if each online item has a deadline of one period, i.e., each arriving item can wait for at most one round to be allocated, EF1 allocation can be found in the above setting.
At last, we present the impossible results to show that our approximation guarantees for EF1 and MMS are tight, and the approximate EF1 (or MMS) is incompatible with USW.
\begin{theorem}
\label{the: ef1_bi_valued_impossiblity_results}
    For two agents with additive personalized bi-valued valuation functions, no deterministic online algorithm can compute an NW allocation that is $(\frac{1}{2}+\epsilon)$-EF1 or $(\frac{1}{3}+\epsilon)$-MMS for any $\epsilon>0$.
\end{theorem}

\begin{proof}
 \begin{table}[tb]
 \begin{minipage}{\linewidth}
    \centering
    \begin{tabular}{c|c|c|c|c}
    \toprule
      & $e_1$ & $e_2$ & $e_3$& $e_4$    \\
      \midrule
        agent $1$  & \textcolor{red}{$\epsilon$} & $1$ &  \textcolor{red}{$\epsilon$} & 1  \\
        \midrule
        agent 2 & $\epsilon$ & \textcolor{red}{$\epsilon$} & $\epsilon$ & 1\\
        \bottomrule
    \end{tabular}
    \caption*{Case 1}
\end{minipage}

\begin{minipage}{\linewidth}
\centering
    \begin{tabular}{c|c|c|c|c|c}
    \toprule
      & $e_1$ & $e_2$ & $e_3$& $e_4$ & $e_5$   \\
      \midrule
        agent $1$  & \textcolor{red}{$\epsilon$} & $1$ &  $\epsilon$ & \textcolor{orange}{$\epsilon$} & $1$ \\
        \midrule
        agent 2 & $\epsilon$ & \textcolor{red}{$\epsilon$} & \textcolor{blue}{$\epsilon$} & 1 & $1$\\
        \bottomrule
    \end{tabular} 
    \caption*{Case 2-1}
\end{minipage}

\begin{minipage}{\linewidth}
    \centering
    \begin{tabular}{c|c|c|c|c}
    \toprule
      & $e_1$ & $e_2$ & $e_3$& $e_4$    \\
      \midrule
        agent $1$  & \textcolor{red}{$\epsilon$} & $1$ &  $\epsilon$ & $\epsilon$ \\
        \midrule
        agent 2 & $\epsilon$ & \textcolor{red}{$\epsilon$} & \textcolor{blue}{$\epsilon$} & \textcolor{green}{1}\\
        \bottomrule
    \end{tabular}
    \caption*{Case 2-2}
\end{minipage}
\caption{The impossibility results of EF1 and MMS for additive personalized bi-valued valuations for two agents.}
    \label{tab: upper_bound_bi_valued_ef1}
\end{table}
Let $a_1=a_2=\epsilon$ and $b_1=b_2=1$, where $0<\epsilon<1$.
 Consider the following online instance for $n=2$, where the value of each item is shown in Table \ref{tab: upper_bound_bi_valued_ef1}.
 When $t = 1$, the item with the value of $\epsilon$ for both agents arrives.
 Without loss of generality, agent 1 picks it.
 When $t = 2$, the arriving item has the value of 1 for agent 1 and $\epsilon$ for agent 2.
 If it is allocated to agent 1, agent 2's bundle is empty, resulting in 0-EF1 and 0-MMS. 
 Thus, it should be allocated to agent 2.
 Then, when $t = 3$, the arriving item has the value of $\epsilon$ for both agents.
 Let us consider two different choices of this item.

\textbf{Case 1.} The arriving item in round $t=3$ is allocated to agent 1. Then, in round $t=4$, the item that has the value of 1 for both agents arrives.
Since we consider non-wastefulness, one agent must pick it.
If it is allocated to agent 1, for agent 1, her bundle is EF1 and MMS to her, and for agent 2, her bundle is $\frac{\epsilon}{1+\epsilon}$-EF1 and $\frac{1}{3}$-MMS.
If it is allocated to agent 2, for agent 1, her bundle is $2\epsilon$-EF1 and $\frac{2\epsilon}{1+\epsilon}$-MMS, and for agent 2, her bundle is EF1 and MMS.
Therefore, we can derive that the upper bound of EF1 is $2\epsilon$ and MMS is $\frac{1}{3}$ respectively.

\textbf{Case 2.} The arriving item in round $t=3$ is allocated to agent 2. Then, in round $t = 4$, the item has the value of $\epsilon$ for agent 1 and 1 for agent 2. 
Let us consider two choices for allocating this item.
\begin{itemize}
    \item If it is allocated to agent 1, this allocation is EF1 and $\frac{2}{3}$-MMS.
    At last, in round $t = 5$, the arriving item has the value of 1 for both agents.
    To guarantee non-wastefulness, it must be allocated to some agent.
    It is not hard to check that no matter which agent picks it, the final allocation is $\frac{2\epsilon}{1+\epsilon}$-EF1 and $\frac{2\epsilon}{1+\epsilon}$-MMS.
    \item If it is allocated to agent 2, there is no further arriving item.
    For agent 1, her bundle is $\frac{1}{2}$-EF1 and $\frac{1}{3}$-MMS, and for agent 2, her bundle is EF and MMS.
    Thus, this allocation is $\frac{1}{2}$-EF1 and $\frac{1}{3}$-MMS.
\end{itemize}
Therefore, we can derive that the upper bound of EF1 is $\frac{1}{2}$ and MMS is $\frac{1}{3}$ respectively.
Combining the above two cases, it is clear that no algorithm can guarantee an NW allocation that is $\alpha$-EF1 or $\beta$-MMS for any $\alpha>\frac{1}{2}$ and $\beta>\frac{1}{3}$.
\end{proof}

 \begin{proposition}
 \label{prop: bivalued_no_efficiency_goods}
Given any $\alpha, \beta >0$, for the deterministic allocation of indivisible goods, no online algorithm can compute an NW allocation satisfying 
\begin{itemize}
    \item $\alpha$-EF1 and $\beta$-max-USW,
    \item or $\alpha$-MMS and $\beta$-max-USW,
\end{itemize}
even for two agents with additive personalized bi-valued valuations.
\end{proposition}

\begin{proof}
    Without loss of generality, we assume that $v_i \in \{1, \frac{1}{\epsilon}\}$ and $v_2 \in \{\epsilon^2,\epsilon\}$, where $0<\epsilon< 1$.
    When $t=1$, the first item $e_1$ that has the value of 1 for agent 1 and $\epsilon$ for agent 2 arrives.
    If it is allocated to agent 2, no further items arrive.
    Hence, this allocation is clearly EF1, MMS, and $\epsilon$-max-USW.
    If not, agent $1$ picks it, and this allocation is EF1, MMS, and max-USW.
    Then, when $t=2$, the second arriving item $e_2$ has the value of 1 for agent 1 and $\epsilon^2$ for agent 2.
    If it is allocated to agent 1, this allocation is 0-EF1, 0-MMS, and max-USW.
    If it is allocated to agent 2, this allocation is EF1, MMS, and $\frac{\epsilon^2+1}{1+\frac{1}{\epsilon}}=\frac{\epsilon^3+\epsilon}{1+\epsilon}$-max-USW.
    For any $\beta>0$, let $\epsilon<\beta^2$, we can derive that no online algorithm can guarantee an NW allocation that is $\alpha$-EF1 (or $\alpha$-MMS) and $\beta$-max-USW for any $\alpha,\beta>0$    
\end{proof}

\subsubsection{The Combination with Additive Binary Valuations}
%In this part, we consider the combination of additive personalized bi-valued and additive binary valuation functions, which is inspired by the above results, and study the setting where only one agent (agent $n$) has the additive personalized bi-valued valuation function, and the remaining agents have identical additive binary valuation functions.
In this part, we focus on the setting where only one agent (agent $n$) has an additive personalized bi-valued valuation function, while all other agents share identical additive binary valuation functions, formed by $(T, N, (v_i)_{i\in N\setminus\{n\}}, v_n)$.
We show that an NW allocation that is EF1 and MMS always exists and can be computed by the Adapted-Picking Algorithm (Algorithm \ref{alg: binary_bivalued_ef1_mms}).
The high-level idea of our algorithm is to allocate the arriving items that all agents value at non-zero values ``evenly'', according to the requirement that agent $n$ has the top priority of picking the item with the value of $b_n$ and the lowest priority of picking the item with the value of $a_n$.
 In that case, agent $n$ takes advantage of the item options, and then we can guarantee that EF1 and MMS hold simultaneously. 
 In Table \ref{tab: binary_bivalued_ef1_mms}, we illustrate how our algorithm works for two agents.
 
 % The reason is that if agent $n$ misses the item with the value of $b_n$, the adversary can design the value of subsequent items to enable the final allocation not to be EF1 and MMS simultaneously.

 \begin{algorithm}[tb]
\caption{Adapted-Picking Algorithm}
\label{alg: binary_bivalued_ef1_mms}
\KwIn{An instance $(T, N, (v_i)_{i \in N\setminus\{n\}}, v_n)$ 
with identical $(n-1)$ additive binary valuation functions $(v_1, \ldots, v_n)$ and one additive personalized bi-valued valuation function $v_n$
}
\KwOut{An NW, EF1, and MMS allocation ${\bf A}$}

Let ${\bf A} = (\emptyset, \ldots, \emptyset)$;

\SetKwProg{Def}{when}{ do}{}
\Def{item $e \in T$ arrives}{

Let $i_{min} \in \arg \min_{i \in [n-1]}|A_i|$;

\eIf{$v_{i_{min}}(e) = 0$}{

$A_n = A_n \cup \{e\}$;
}
{

\eIf{$v_n(e) =b_n$}{

\eIf{$ v_{i_{min}}(A_n) \leq v_{i_{min}}(A_{i_{min}})  $}{
$A_n = A_n \cup \{e\}$;
}
{

$A_{i_{min}} = A_{i_{min}} \cup \{e\}$;

}

}
{
\eIf{$  v_{i_{min}}(A_n) > v_{i_{min}}(A_{i_{min}})  $}{

$A_{i_{min}} = A_{i_{min}} \cup \{e\}$;

}
{
$A_n = A_n \cup \{e\}$;

}
}

}
}
\Return ${\bf A}$;
\end{algorithm}

 \begin{table}[t]
\centering
    \begin{tabular}{c|c|c|c|c|c|c}
    \toprule
      & $e_1$ & $e_2$ & $e_3$& $e_4$ & $e_5$& $\ldots$   \\
      \midrule
        agent $1$  & \textcolor{red}{$1$} & $1$ &  $0$ & 1 & \textcolor{red}{1}&$\ldots$\\
        \midrule
        agent 2 & $a_2$ & \textcolor{red}{$a_2$} &  \textcolor{red}{$a_2$} & \textcolor{red}{$b_2$} & $a_2$&$\ldots$\\
        \bottomrule
    \end{tabular} 
    \caption{A simple illustration for two agents of Algorithm \ref{alg: binary_bivalued_ef1_mms}.}
    \label{tab: binary_bivalued_ef1_mms}
\end{table}

\begin{theorem}
\label{the: binary_bivalued_goods}
 For the deterministic allocation of indivisible goods for $(n-1)$ agents with identical additive binary valuation functions and one agent with the additive personalized bi-valued valuation function, the Adapted-Picking Algorithm (Algorithm \ref{alg: binary_bivalued_ef1_mms}) computes an NW, EF1, and MMS allocation.  
\end{theorem}
\begin{proof}
Given an online fair allocation instance $I=(T, N, (v_i)_{i\in N})$,
Fix an arbitrary round $k$, let $T=\{e_1,\ldots,e_k\}$ be the set of items that have already arrived, denoted by $\textbf{A}^k=(A^k_1,\ldots, A^k_n)$, the allocation at the end of round $t$ in Algorithm \ref{alg: binary_bivalued_ef1_mms}.  
 \paragraph{EF1.} We show that the allocation in any round is EF1 by induction.  
 For the base case, the empty allocation is trivially EF1.
 For the induction step, after allocating item $e_k$, the allocation $\textbf{A}^k$ is EF1.
 Next, we show that after allocating item $e_{k+1}$, the allocation $\textbf{A}^{k}$ is still EF1.
 When item $e_{k+1}$ arrives, let us consider several cases about the allocation of this item.
 \begin{itemize}
     \item  If $v_{i_{min}}(e_{k+1})=0$, it is allocated to agent $n$. 
     For any agent $i \in [n-1]$, she values it at zero, so her bundle is still EF1 to her.
     For agent $n$, the value of her bundle increases, so her bundle is also EF1 to her.
     \item If $v_n(e_{k+1})=b_n$, let us consider two subcases.
     When $  v_{i_{min}}(A_n^k) \leq v_{i_{min}}(A_{i_{min}}^k) $, agent $n$ picks it.
     For agent $n$, her new bundle is EF.
     That is because she has the priority of picking the item with the value of $b_n$ and before allocating item $e_{k+1}$, if she envies other agents $j \in [n-1]$, in her perspective, the gap between the value of agent $j$'s bundle and her bundle is at most $a_n$, i.e.,$v_n(A_j) - v_n(A_n) \leq a_n$.
     Now, fix an arbitrary agent $i \in [n-1]$.
     Since agent $i$ has the same additive binary valuation function and she only picks the item with the non-zero value, we have $v_i(A_i^{k+1}) = v_i(A_i^{k}) \geq v_i(A_n^k) = v_i(A_n^{k+1}\setminus\{e_{k+1}\})$, implying her bundle is also EF1.
     When $  v_{i_{min}}(A_n^k) > v_{i_{min}}(A_{i_{min}}^k)$. agent $i_{min}$ picks it.
     Clearly, for agent $i_{min}$, her new bundle is still EF1.
     For any other agent $i \in N \setminus \{i_{min},n\}$, we have $v_i(A_i^{k+1}) = v_i(A_i^{k}) \geq v_i(A_{i_{min}}^k) = v_i(A_{i_{min}}^{k+1}\setminus\{e_{k+1}\})$, implying that her bundle is still EF1.
     For agent $n$, based on the description of our algorithm, she must already pick one item that has the value of $b_n$, and in her perspective, the number of items in her bundle is at least one more than that of agent $i_{min}$'s bundle, which means that agent $n$ does not envy agent $i_{min}$ and in her perspective, the value of her bundle is at least $b_n$ greater than agent $i_{min}$'s bundle.
     Therefore, allocating $e_{k+1}$ to agent $i_{min}$ will not make agent $n$ envy agent $i_{min}$.
     \item If $v_n(e_{k+1}) =a_n$,
     let us consider two subcases.
     When $  v_{i_{min}}(A_n^k) \leq v_{i_{min}}(A_{i_{min}}^k) $, agent $i_{min}$ picks it.
     It is trivial that for agent $i_{min}$, her new bundle is EF1.
      For any other agent $i \in N \setminus \{i_{min},n\}$, we have $v_i(A_i^{k+1}) = v_i(A_i^{k}) \geq v_i(A_{i_{min}}^k) = v_i(A_{i_{min}}^{k+1}\setminus\{e_{k+1}\})$, implying that her bundle is still EF1.
      For agent $n$, before allocating $e_{k+1}$, she does not envy agent $i_{min}$ based on the fact that the number of items in her bundle is at least the same as that of agent $i_{min}$'s bundle and our algorithm description.
      When $ v_{i_{min}}(A_n^k) > v_{i_{min}}(A_{i_{min}}^k) $, agent $n$ picks it.
      For agent $n$, her new bundle is actually EF since she has the priority of picking the item with the value of $b_n$ and before allocating $e_{k+1}$, she envies other agent $j \in [n-1]$; in her perspective, the gap between the value of agent $j$'s bundle and her bundle is at most $a_n$, i.e.,$v_n(A_j) - v_n(A_n) \leq a_n$.
      Now, fix an arbitrary agent $i \in [n-1]$.
      Since agent $i$ has the same additive binary valuation function and she only picks the item with the non-zero value, we have $v_i(A_i^{k+1}) = v_i(A_i^{k}) \geq v_i(A_n^k) = v_i(A_n^{k+1}\setminus\{e_{k+1}\})$, implying EF1 also holds.  
 \end{itemize}
Combining the above cases, it can be concluded that in round $k+1$, the allocation is EF1. 
This completes the induction and establishes the correctness of our claim.

\paragraph{MMS.}
Fix an arbitrary round $k$.
Pick an arbitrary agent $i \in [n-1]$.
For agent $i$, assume that the number of items with the value of one that arrive between round $t=1$ and $k$ is $s_i^{k}$.
By the description of our algorithm, items are almost evenly allocated to agents with identical additive binary valuations.
Thus, it is trivial that $v_i(A_i) = |v_i(A_i)| \geq \mathsf{MMS}_i = \lfloor \frac{s_i^k}{n}\rfloor$.
For agent $n$, if she picks item $e_k$, by the above proof of EF1, her bundle is EF, implying that her bundle is also MMS; 
If she does not pick item $e_k$, by the above proof of EF1, the value of her bundle is at most $a_n$ less than that of the envied agent's bundle. 
Considering that she has the priority of picking the item with the value of $b_n$ and the number of items in her bundle is at most one less than that of the other agent's bundle, it is not hard to check that her bundle is still MMS. 
\end{proof}

At last, we complement the above result by showing the incompatibility of EF1 (or MMS) and USW.
%in the setting of the combination of additive binary and additive personalized bi-valued valuation functions.
\begin{proposition}
\label{prop: binary_bivalued_no_efficiency_goods}
Given any $\alpha>0$, for the deterministic allocation of indivisible goods, no online algorithm can compute an NW allocation satisfying
\begin{itemize}
    \item EF1 and $\alpha$-max-USW,
    \item or MMS and $\alpha$-max-USW,
\end{itemize}
even for two agents, where one agent has the additive binary valuation function, and the other agent has the additive personalized bi-valued valuation function.
\end{proposition}

\begin{proof}
    Without loss of generality, we assume that $v_1\in \{0,1\}$ and $v_2 \in \{\epsilon^2,\epsilon\}$, where $0<\epsilon< 1$.
    When $t=1$, the first item $e_1$ that has the value of 1 for agent 1 and $\epsilon$ for agent 2 arrives.
    If it is allocated to agent 2, no further items arrive.
    Hence, this allocation is clearly EF1, MMS, and $\epsilon$-max-USW.
    If not, agent $1$ picks it, and this allocation is EF1, MMS, and max-USW.
    Then, when $t=2$, the second arriving item $e_2$ has the value of 1 for agent 1 and $\epsilon^2$ for agent 2.
    If it is allocated to agent 1, this allocation is 0-EF1, 0-MMS, and max-USW.
    If it is allocated to agent 2, this allocation is EF1, MMS, and $\frac{\epsilon^2+1}{2}\approx \frac{1}{2}$-max-USW.
    When $t =3$, the item $e_3$ that has the value of 1 for agent 1 and $\epsilon$ for agent 2 arrives.
    If it is allocated to agent 1, this allocation is $\epsilon$-EF1, $\epsilon$-MMS, and $\frac{2+\epsilon^2}{3}\approx \frac{2}{3}$-max-USW.
    If it is allocated to agent 2, this allocation is $\frac{1}{2}$-EF1, $\frac{1}{2}$-MMS, and $\frac{1+\epsilon^2+\epsilon}{3}\approx\frac{1}{3}$-max-USW.
    In other words, no allocation is EF1 or MMS.
    Combining the above cases, it can be seen that if an online algorithm outputs an NW allocation that is EF1 or MMS in the setting of additive binary and additive personalized bi-valued valuations, it may not guarantee $\alpha$-max-USW for any given $\alpha>0$.      \end{proof}

\section{Allocation of Chores}
\label{sec: chores_setting}
In this section, we study the setting of indivisible chores.
Because ``online'' chores can be regarded as daily compulsory tasks or duties, it is unreasonable to throw them away.
In other words, for each arriving chore, there must be some agent taking it. 
Therefore, we first introduce the following constraint.
\begin{definition}[Completeness]
\label{def: completeness}
An allocation ${\bf A}^k$ is complete if, for any chore $e$ that has arrived by round $k$, it is not discarded.
\end{definition}
Here are the fairness and efficiency definitions used in the chores setting.
\begin{definition}[Envy-freeness up to One Chore]
\label{def: ef1_chores}
For any $\alpha \in [1,+\infty)$, an allocation ${\bf A}^k$ is $\alpha$-approximate envy-free up to one chore ($\alpha$-EF1) if, for every pair of agents $i, i^{\prime} \in N$ with $A_{i}^k \neq \emptyset$, it holds that $c_i(A_i^k \setminus \{e\}) \leq \alpha \cdot c_i(A_{i^{\prime}}^k)$ for some $e \in A_{i}^k$.  
\end{definition}
\begin{definition}
[Minimax Share Fairness]
\label{def: mms_chores}
For any agent $i$, her minimax share $\mathsf{MMS}_{i}^k$ by round $k$ is defined as:
\[
\mathsf{MMS}_{i}^k = \min_{{\bf B}^k \in \Pi_{n}(T^k)} \max_{i^{\prime} \in N}
c_i(B_{i^{\prime}}^k),
\]
where $\Pi_{n}(T^k)$ is the set of allocations of $T^k$.
For any $\alpha \in [ 1, +\infty)$, an allocation ${\bf A}^k$ is $\alpha$-approximate minimiax share fair ($\alpha$-MMS) if for any agent $i \in N$, it holds that $c_i(A_i^k) \leq \alpha \cdot \mathsf{MMS}_i^k$.
\end{definition}
\begin{definition}[Utilitarian Social Cost]
\label{def: usc}
The utilitarian social cost $USC$ of  allocation ${\bf A}^k$ is given by $USC({\bf A}^k) = \sum_{i \in N}c_i(A_i^k)$.
For any $\alpha \in [1, +\infty)$, an allocation ${\bf A}^k$ is $\alpha$-min-USC
if it holds that $USC({\bf A}^k) \leq \alpha \cdot \min_{{\bf B}^k \in \Pi_{n}(T^k)}USC({\bf B}^k)$, where $\Pi_{n}(T^k)$ is the set of allocations of $T^k$.
\end{definition}

\subsection{The Impossibility of Approximate EF1}

In this part, we show the following negative result about the approximation of EF1 for additive personalized tri-valued cost functions.
For any agent $i \in N$,  her cost function $c_i$ is additive personalized tri-valued if $c_i$ is additive and $c_i(e)\in \{a_i, b_i,z_i\}$ for any $e\in T$, where $0<a_i\leq b_i\leq z_i$.
\begin{theorem}
\label{the: trivalued_ef1_chores}
 No deterministic online algorithm can compute a complete allocation that guarantees $\alpha$-EF1 for any $\alpha \geq 1$, even for two agents with
additive personalized tri-valued cost functions.
\end{theorem}
\begin{proof}
\begin{table}[b]
\centering
    \begin{tabular}{c|c|c|c}
    \toprule
      & $e_1$ & $e_2$ & $e_3$   \\
      \midrule
        agent $1$  & \textcolor{red}{$1$} & $\epsilon$ &  $\frac{1}{\epsilon}$ \\
        \midrule
        agent 2 & $1$ & \textcolor{red}{$\frac{1}{\epsilon}$} &  $\frac{1}{\epsilon}$\\
        \bottomrule
    \end{tabular} 
    \caption{The impossibility result of EF1 for additive personalized tri-valued cost functions.}
    \label{tab: impossiblility_ef1_three_values_chores}
\end{table}
% First, let us consider the case when $n=2$.
    Let $a_1=a_2 = \epsilon$, $b_1=b_2=1$, and $z_1=z_2=\frac{1}{\epsilon}$, where $0<\epsilon <1 $.
    Consider the following online instance for $n=2$, where the cost of each item is shown in Table \ref{tab: impossiblility_ef1_three_values_chores}.
    When $t =1$, the arriving item has the cost of 1 for both agents.
    Without loss of generality, agent 1 picks it.
    When $t = 2$, the item that has the cost of $\epsilon$ for agent 1 and $\frac{1}{\epsilon}$ for agent 2 arrives.
    In order to avoid the approximation ratio of EF1 being infinite, this item must be allocated to agent 2.
    At last, in round $t = 3$, the arriving item has the cost of $\frac{1}{\epsilon}$ for both agents.
    In this case, if this item is allocated to agent 1, for agent 2, the approximation ratio of EF1 is $\frac{1}{\epsilon}$;
    if this item is allocated to agent 2, for agent 1, the approximation ratio of EF1 is $\frac{1}{\epsilon}$.
    Therefore, for any $\alpha > 0$, let $\epsilon < \frac{1}{\alpha}$, and we can conclude that no algorithm can guarantee a complete allocation that is $\alpha$-EF1 for any given $\alpha > 0$. 
    This impossibility result can be generalized to more than two agents, and we defer the proof to Proposition \ref{prop: no_ef1_chores_three_agents}.
\end{proof}

\begin{proposition}
\label{prop: no_ef1_chores_three_agents}
    For general additive cost functions, no deterministic online algorithm can compute a complete allocation that guarantees $\alpha$-EF1 for any $\alpha \geq 1$ for more than two agents with
additive personalized tri-valued cost functions.
\end{proposition}
\begin{proof}

\begin{table}[tb]
\centering
    \begin{tabular}{c|c|c|c|c|c|c|c}
    \toprule
      & $e_1$ & $e_2$ & $\ldots$&$e_{n-2}$ & $e_{n-1}$ & $e_n$ & $e_{n+1}$   \\
      \midrule
        agent $1$  & \textcolor{red}{$\frac{1}{\epsilon}$} & $\frac{1}{\epsilon}$ &  $\ldots$ &$\frac{1}{\epsilon}$ & $1$ & $\epsilon$ & $\frac{1}{\epsilon}$ \\
        \midrule
       agent $2$  & $\frac{1}{\epsilon}$ & \textcolor{red}{$\frac{1}{\epsilon}$} &  $\ldots$&$\frac{1}{\epsilon}$ & $1$ & $\epsilon$ & $\frac{1}{\epsilon}$\\
       \midrule
       $\ldots$& $\ldots$ & $\ldots$ & $\ldots$ & $\ldots$ & 
       $\ldots$ &$\ldots$ & $\ldots$\\
       \midrule
       agent $n-2$  & $\frac{1}{\epsilon}$ & $\frac{1}{\epsilon}$ &  $\ldots$&\textcolor{red}{$\frac{1}{\epsilon}$} & $1$ & $\epsilon$ & $\frac{1}{\epsilon}$\\
       \midrule
       agent $n-1$  & $\frac{1}{\epsilon}$ & $\frac{1}{\epsilon}$ &  $\ldots$&$\frac{1}{\epsilon}$ & \textcolor{red}{$1$} & $\epsilon$ & $\frac{1}{\epsilon}$\\
       \midrule
       agent $n$  & $\frac{1}{\epsilon}$ & $\frac{1}{\epsilon}$ &  $\ldots$&$\frac{1}{\epsilon}$ & $1$ & \textcolor{red}{$\frac{1}{\epsilon}$} & $\frac{1}{\epsilon}$\\
        \bottomrule
    \end{tabular} 

    \caption{The impossibility result of EF1 for additive personalized tri-valued cost functions when there are more than two agents.}
    \label{tab: impossiblility_ef1_three_values_chores_general}
\end{table}

We extend the impossibility result of approximate EF1 for chores to $n\geq 3$ agents.
    Consider the following instance with $n\geq 3$ agents, where the cost of each item is shown in Table \ref{tab: impossiblility_ef1_three_values_chores_general}.
    Before each agent's bundle has one agent, we should allocate the arriving item to the agent with an empty bundle to avoid the approximation ratio of EF1 being infinite.
    For $t \leq n-2$, each arriving item has the cost of $\frac{1}{\epsilon}$ for each agent.
    Without loss of generality, assume that agent $i \in [n-2]$ picks the item that arrives at round $i$.
    When $t = n-1$, the arriving item has the cost of 1 for all agents.
    Assume that it is allocated to agent $n-1$.
    When $t = n$, the arriving item has the cost of $\frac{1}{\epsilon}$ for agent $n$ and $\epsilon$ for the remaining agents.
    It is clear that agent $n$ picks it.
    In the last round $t = n+1$, the item that has the cost of $\frac{1}{\epsilon}$ for all agents arrives.
    If it is allocated to agent $i \in [n-2]$, we can derive that the approximation ratio of EF1 is $\frac{1}{\epsilon}$.
    If it is allocated to agent $n-1$ or $n$, the allocation is also $\frac{1}{\epsilon}$-EF1.
    Thus, the allocation in this round is $\frac{1}{\epsilon}$-EF1.
    Given any $\alpha>1$, let $\epsilon>\frac{1}{\alpha}$, we can conclude that no online algorithm can compute a complete allocation that is $\alpha$-EF1 for any $\alpha>1$.  
\end{proof}

By the above theorem, no approximation of EF1 can be guaranteed for additive personalized tri-valued cost functions. Like the negative result of goods setting, we cannot make further progress in the approximation of EF1 beyond additive personalized tri-valued cost functions if there is no prior information about future items. Therefore, in the following parts, we focus on additive binary, supermodular binary, and additive personalized bi-valued cost functions, and study whether there are positive results about the approximation of EF1.

\subsection{Additive/Supermodular Binary Costs}
In this part, we consider (additive) binary and supermodular binary cost functions in the chores setting.
\subsubsection{Additive Binary Costs}
For additive binary cost functions, where for any agent $i \in N$, $c_i$ is additive and $c_i(e)\in \{0,1\}$ for any $e \in T$, we propose the Compelled-Greedy Algorithm (Algorithm \ref{alg: indivisible_suppermodular+binary}) to compute a complete allocation that satisfies EF1, MMS, and min-USC.

% Unexpectedly, for supermodular binary cost functions, there is a strong negative result regarding the approximation of fairness and efficiency, which is in sharp contrast to the results of supermodular binary cost functions in the offline setting, where an EF1, MMS, or min-USC allocation can be found efficiently \cite{barman2023suppermodular}.

\begin{algorithm}[tb]
\caption{Compelled-Greedy Algorithm}
\label{alg: indivisible_suppermodular+binary}
\KwIn{An instance $(T, N, (c)_{i \in N})$ with additive binary cost functions}
\KwOut{An complete EF1 allocation ${\bf A}$}
Initialize: let $\pi=(\pi_1,\ldots, \pi_n)$ be an arbitrary order of $n$ agents;

Let ${\bf A} = (\emptyset, \ldots, \emptyset)$;

\SetKwProg{Def}{when}{ do}{}
\Def{item $e \in T$ arrives}{

\For{$i=1$ to $n$}
{

\If{$v_{\pi_i}(\{e\})=0$}
{
 Let $A_{\pi_i}=A_{\pi_i}\cup \{e\}$;

 break;
}
}
$A_{\pi_1}=A_{\pi_1}\cup \{e\}$;

$\pi=(\pi_2,\ldots,\pi_{i-1},\pi_{i+1},\ldots,\pi_n,\pi_1)$;
}

\Return ${\bf A}$;

\end{algorithm}

\begin{theorem}
\label{the: indivisible_binary_chore}
 For additive binary cost functions, the Compelled-Greedy Algorithm (Algorithm \ref{alg: indivisible_suppermodular+binary}) computes a complete allocation that satisfies EF1, MMS, and min-USC.  
\end{theorem}
\begin{proof}
Given an online fair allocation instance $I=(T, N, (c_i)_{i\in N})$, where $c_i$ is additive and binary for all $i\in N$.
% For any submodular binary function $c_i$, there exists a matroid $\mathcal{M}_i$ such that function $g_i(\cdot)=|\cdot|-c_i(\cdot)$ is the rank function of $\mathcal{M}_i$.
 Fix an arbitrary round $k$, let $T=\{e_1,\ldots,e_k\}$ be the set of items that have already arrived, denoted by ${\bf A^k}=(A^k_1,\ldots, A^k_n)$ the allocation at the end of round $k$.
 For every item $e\in T$, if there exists one agent such that the cost about $e$ of some agent is $0$, the algorithm must allocate $e$ to some agent $a$ with $v_a(e)=0$.
 Otherwise, the costs of all agents are $1$.
 Regardless of how it is allocated, the increment brought to utilitarian social welfare is $1$, so $A^k$ is a min-USC allocation.
 
 Next, we show that ${\bf A^k}$ is EF1.
 Assume that there exists $i,j\in N$ such that, for any item $e\in A^k_i$, 
 \begin{equation}
 \label{non-ef1}
   c_i(A^k_j)\le c_i(A^k_i\setminus\{e\}). 
 \end{equation} 
  There must exist an item $e^i\in A^k_i$ such that $c_i(A^k_i)-c_i(A^k_i\setminus\{e\})=1$. 
 % Otherwise, for any item $e\in A^t_i$ and any $S\in A^t_i$, $c_i(S)-c_i(S_i\setminus\{e\})\le c_i(A^t_i)-c_i(A^t_i\setminus\{e\})=0$, where the inequality is based on suppermodularity.
%  It indicates $c_i(A^t_i)=0$, contradicting to inequality (\ref{non-ef1}).
 Without loss of generality, we assume $e^i$ is the last item with $1$ cost, and it arrives at $\ell$-th round.
  For any agent, the cost of $e^i$ is $1$, and agent $i$'s
 priority is higher than $j$.
 Due to the algorithm's order in each round, the number of items with a $1$ cost of agent $j$ must be at least as large as the number of items with a $1$ cost of agent $j$.
 It contradicts inequality (\ref{non-ef1}).

 Finally, we show that ${\bf A^k}$ is MMS.
 For any agent $j
 \in N$ and any item $e\in A^k_j$, if the cost of $c_j(\{e\})$ is $1$, $c_i(\{e\})=1$ for all $i\in N$. 
 So, $c_i(A^{r}_j)\ge c_j(A^{r}_j)$ for any $r$-th round.
 $e^i$ is the last item with $1$ cost and it arrives at $\ell$-th round, 
 let $T_{\ell}$ be the set arrived after 
$\ell$-th round.
 In the $\ell-1$-th round, 
$c_i(A^{\ell}_i)=c_i(A^{\ell-1}_i)+1\le \frac{\sum_{j\in N}c_i(A^{\ell-1}_j)}{n}+1=\frac{c_i(T_{\ell-1})}{n}+1$.
Since $c_i(A^{\ell}_i)$ is an integer, 
we have $c_i(A^k_i)=c_i(A^{\ell}_i)\le \lceil\frac{c_i(T_{\ell})}{n}\rceil\le\lceil\frac{c_i(T)}{n}\rceil=\mathsf{MMS}_i$.
\end{proof}

\subsubsection{Supermodular Binary Costs}
For supermodular binary cost functions, where for any agent $i \in N$, $c_i$ is supermodular and the marginal cost of any item is 0 or 1, we first obtain a significant negative result regarding the approximation of fairness and efficiency. 
This contrasts starkly with the findings in offline settings, where it is possible to efficiently obtain an allocation satisfying EF1, MMS, or min-USC \cite{barman2023suppermodular}.

\begin{lemma}(Lemma 2 in \cite{barman2023suppermodular})
\label{lem: supmodular_binary_chore}
A cost function $c$ is supermodular and binary if and only if the function $g$, $g(S)=|S|-c(S)$ for any set $S$, is a matroid rank function.
\end{lemma}

\begin{theorem}
\label{the: supermodular_binary_negative_results}
    For any $\alpha \geq  1$, with supermodular binary cost functions, no deterministic online algorithms can achieve
    \begin{itemize}
        \item $\alpha$-EF1 and completeness,
        \item or $\alpha$-MMS and completeness,
        \item or $\alpha$-min-USC and completeness. 
    \end{itemize} 
\end{theorem}

\begin{proof}
\begin{figure}
     \centering
     \includegraphics[width=0.7\textwidth]{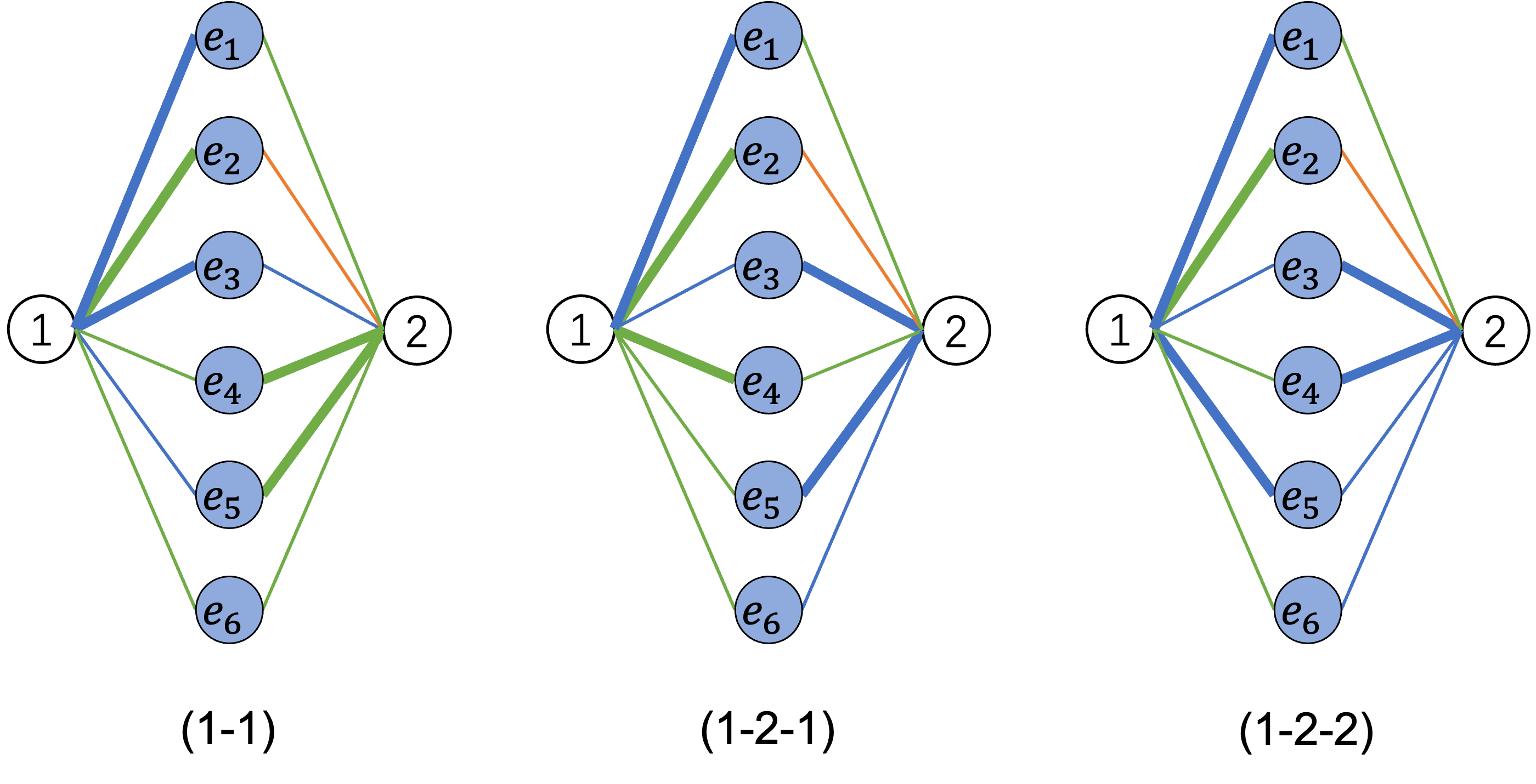}
     
     \includegraphics[width=0.5\textwidth]{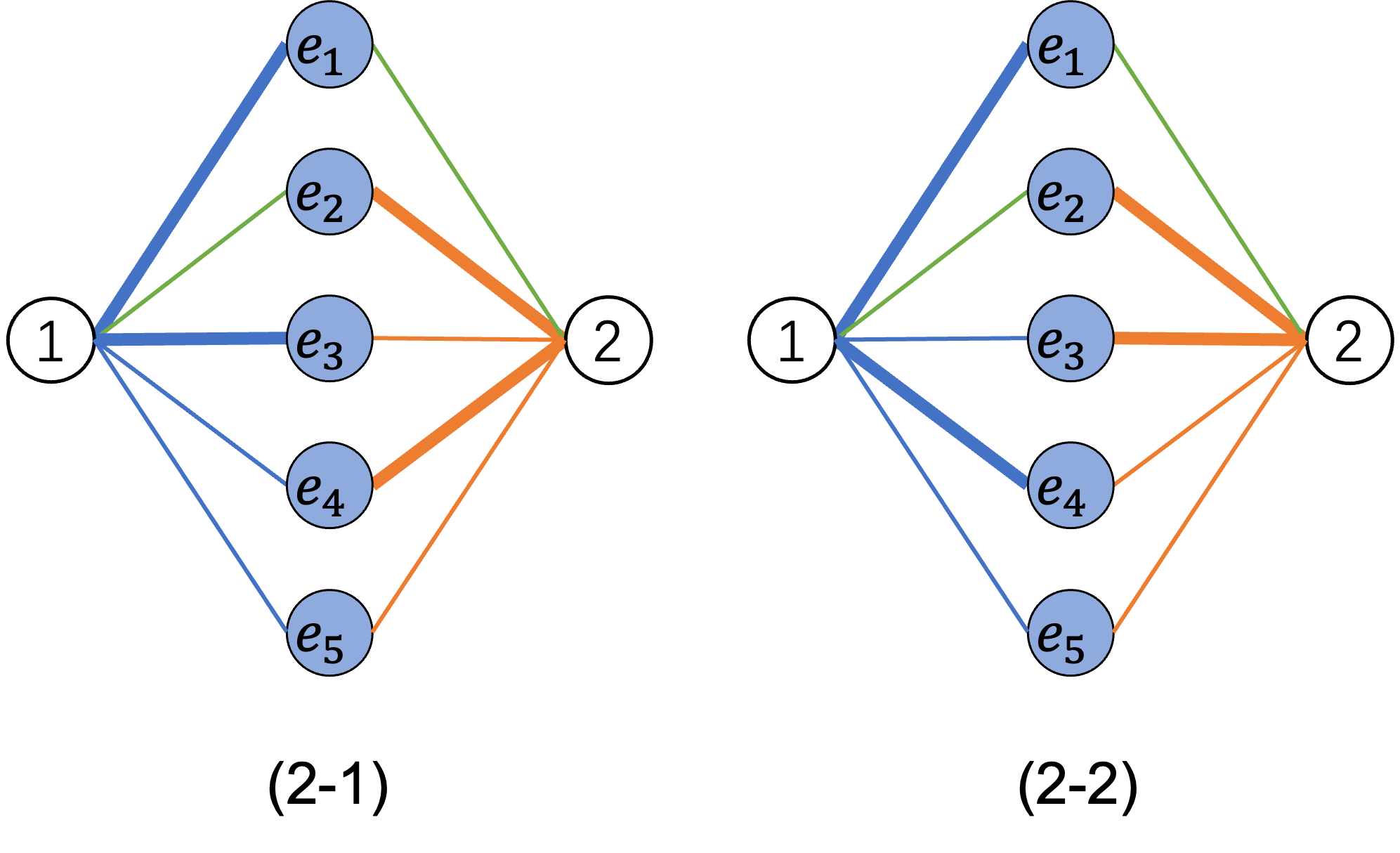}
     \caption{An instance with $\frac{1}{0}$-EF1 allocations}
     \label{fig3}
 \end{figure}
First, let us consider EF1. Consider an online allocation instance; there are two agents $\{1,2\}$ with supermodular binary cost functions $c_1$ and $c_2$.

Based on the relationship between the supermodular and binary functions and the mathematical rank function in Lemma \ref{lem: supmodular_binary_chore}, we construct cost functions for the instance by constructing a matroid rank function.
For $i\in \{1,2\}$, the online items in $T$ belongs to several categories $\{C^i_1,\ldots,C^i_l\}$.
In Figure \ref{fig3}, categories are distinguished by colours, and edges of different colours represent items of different categories for adjacent agents.
For agent $i\in \{1,2\}$, consider the corresponding partition matroid $\mathcal{M}_i=\{T,\mathcal{F}_i\}$, $\mathcal{F}_i=\{S: S\subseteq T, |S\cap C_k|\le 1, \text{for all}~ k\in [l]\}$. 
Let the function $r_i$ be the rank function of matroid $\mathcal{M}_i$, the cost function $c_i(\cdot)$ is defined as $c_i(S)=|S|-r_i(S)$ for any item set $S\in T$.
For the first arrived item $e_1$, the marginal utilities for agent $1$ and agent $2$ are $0$.
W.l.o.g.,
we allocate item $e_1$ to agent $1$. 
When the second item $e_2$ arrives, the marginal utilities for agent $1$ and agent $2$ are $0$.
We will discuss the following two cases. 

\noindent\textbf{Case 1}: we allocate $e_2$ to agent $1$.
For the third item $e_3$, the marginal utility for agent $1$ is $1$, and the marginal utility for agent $2$ is $0$. 
Regardless of how it is allocated, it satisfies EF1, so we continue to discuss it in two sub-cases.
\textbf{Case 1-1}: we allocate $e_3$ to agent $1$.
When the fourth item $e_4$ arrives, the marginal utility for agent $1$ is $1$, and the marginal utility for agent $2$ is $0$. 
To not break the requirement of EF1, we must allocate item $e_4$ to agent $2$.
When the fifth item $e_5$ arrives, the marginal utility for agent $1$ is $1$, and the marginal utility for agent $2$ is $0$. 
To avoid breaking the requirement of EF1, we must allocate item $e_5$ to agent $2$.
When the sixth item $e_6$ arrives, the marginal utility for agent $1$ and agent $2$ is $1$. 
If we allocate $e_6$ to agent $1$, $A^6_1=\{e_1, e_2, e_3,e_6\}$, $A^6_2=\{e_4, e_5\}$, then we have $c_1(A^6_2)=0<1=c_1(A^6_1\setminus \{e\})$.
 Else, we allocate $e_6$ to agent $2$, $A^6_1=\{e_1, e_2, e_3\}$, $A^6_2=\{e_4, e_5,e_6\}$, then we have $c_2(A^6_1)=0<1=c_2(A^6_2\setminus \{e\})$.
\textbf{Case 1-2}: we allocate $e_3$ to agent $2$.
When the fourth item $e_4$ arrives, the marginal utility for agent $1$ is $1$, and the marginal utility for agent $2$ is $0$. 
Regardless of how it is allocated, it satisfies EF1, so we continue to discuss it in two sub-cases.
\textbf{Case 1-2-1}: we allocate $e_4$ to agent $1$.
When the fifth item $e_5$ arrives, the marginal utilities for agent $1$ and agent $2$ are $1$. 
To avoid breaking the requirement of EF1, we must allocate item $e_5$ to agent $2$.
When the sixth item $e_6$ arrives, the marginal utility for agent $1$ and agent $2$ is $1$. 
If we allocate $e_6$ to agent $1$, $A^6_1=\{e_1, e_2, e_4,e_6\}$, $A^6_2=\{e_3,e_5\}$, then we have $c_1(A^6_2)=0<1=c_1(A^5_1\setminus \{e\})$ for any $e\in A^5_1$.
 Else, we allocate $e_6$ to agent $2$, $A^6_1=\{e_1, e_2, e_4\}$, $A^6_2=\{e_3, e_5,e_6\}$, then we have $c_2(A^6_1)=0<1=c_2(A^6_2\setminus \{e\})$ for any $e\in A^6_2$.
\textbf{Case 1-2-2}: we allocate $e_4$ to agent $2$.
When the fifth item $e_5$ arrives, the marginal utilities for agent $1$ and agent $2$ are $1$. 
To avoid breaking the requirement of EF1, we must allocate item $e_5$ to agent $1$.
When the sixth item $e_6$ arrives, the marginal utility for agent $1$ and agent $2$ is $1$. 
If we allocate $e_6$ to agent $1$, $A^6_1=\{e_1, e_2, e_5,e_6\}$, $A^6_2=\{e_3,e_4\}$, then we have $c_1(A^6_2)=0<1=c_1(A^5_1\setminus \{e\})$ for any $e\in A^5_1$.
 Else, we allocate $e_6$ to agent $2$, $A^6_1=\{e_1, e_2, e_5\}$, $A^6_2=\{e_3, e_4,e_6\}$, then we have $c_2(A^6_1)=0<1=c_2(A^6_2\setminus \{e\})$ for any $e\in A^6_2$.

\noindent\textbf{Case 2}: we allocate $e_2$ to agent $2$.
For the third item $e_3$, the marginal utilities for agent $1$ and agent $2$ are $0$. Regardless of how it is allocated, it satisfies EF1, so we continue to discuss it in two sub-cases.
\textbf{Case 2-1}: we allocate $e_3$ to agent $1$.
For the fourth item $e_4$,  
the marginal utilities for agent $1$ and agent $2$ are $1$. 
To avoid breaking the requirement of EF1, we must allocate item $e_4$ to agent $2$.
When the fifth item $e_5$ arrives, 
the marginal utilities for agent $1$ and agent $2$ are $1$. 
If we allocate $e_5$ to agent $1$, $A^5_1=\{e_1, e_3, e_5\}$, $A^5_2=\{e_2, e_4\}$, then we have $c_1(A^5_2)=0<1=c_1(A^5_1\setminus \{e\})$.
 Else, we allocate $e_5$ to agent $2$, $A^5_1=\{e_1, e_3\}$, $A^5_2=\{e_2, e_4,e_5\}$, then we have $c_2(A^5_1)=0<1=c_2(A^5_2\setminus \{e\})$.
\textbf{Case 2-2}: we allocate $e_3$ to agent $2$.
For the fourth item $e_4$,  
the marginal utilities for agent $1$ and agent $2$ are $1$. 
To avoid breaking the requirement of EF1, we must allocate item $e_4$ to agent $1$.
When the fifth item $e_5$ arrives, 
the marginal utilities for agent $1$ and agent $2$ are $1$. 
If we allocate $e_5$ to agent $1$, $A^5_1=\{e_1, e_4, e_5\}$, $A^5_2=\{e_2, e_3\}$, then we have $c_1(A^5_2)=0<1=c_1(A^5_1\setminus \{e\})$.
 Else, we allocate $e_5$ to agent $2$, $A^5_1=\{e_1, e_4\}$, $A^5_2=\{e_2, e_3,e_5\}$, then we have $c_2(A^5_1)=0<1=c_2(A^5_2\setminus \{e\})$.

\begin{figure}
     \centering
     \includegraphics[width=0.8\textwidth]{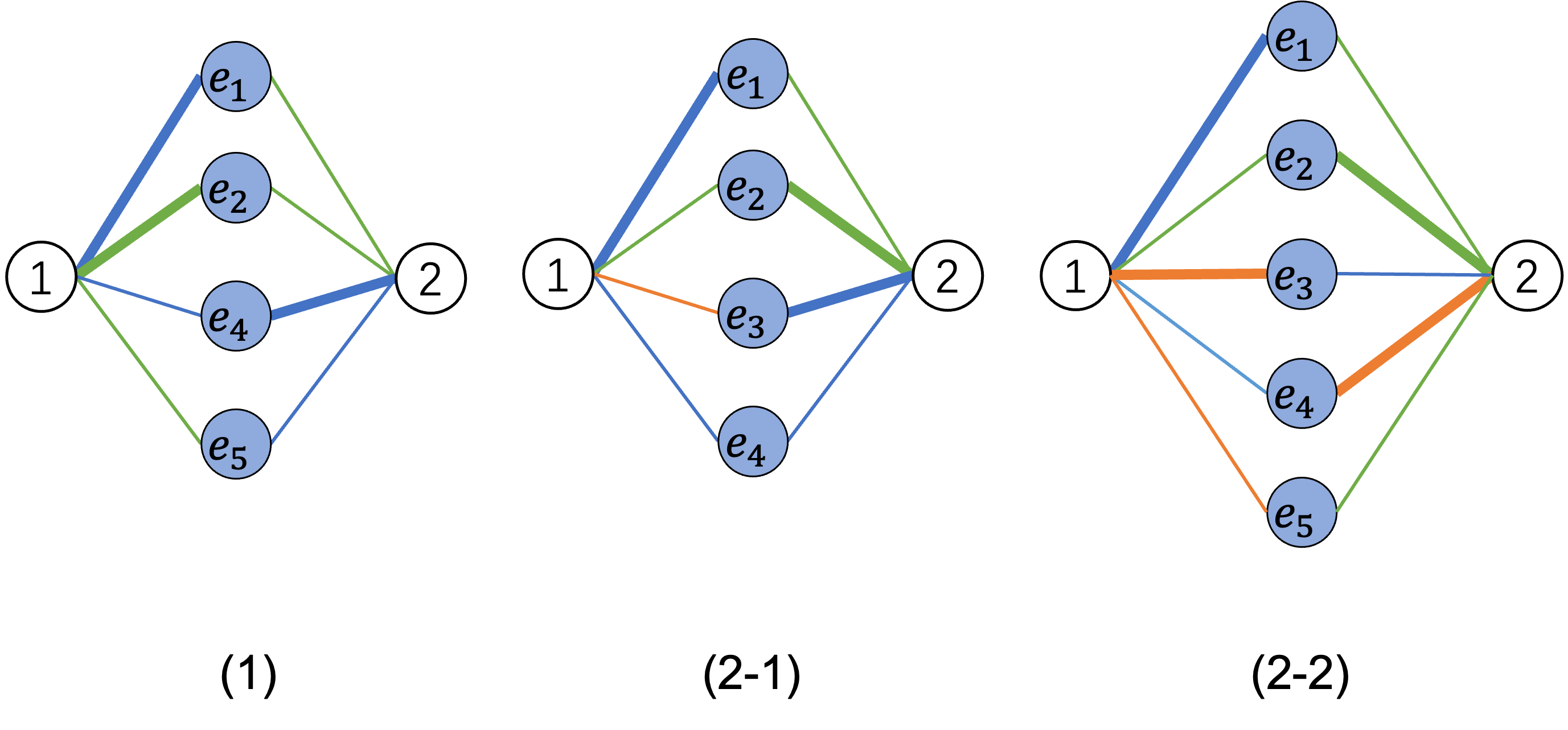}
     \caption{An instance with $\frac{1}{0}$-MMS allocations}
     \label{fig2}
 \end{figure}
Next, let us consider MMS or USC. Consider an online allocation instance; there are two agents $\{1,2\}$ with supermodular binary cost functions $v_1 v_2$.
For $i\in \{1,2\}$, the online items in $T$ belongs to several categories $C^i_1,\ldots,C^i_l$.
In Figure \ref{fig2}, categories are distinguished by colours, and edges of different colours represent items of different categories for adjacent agents.
For agent $i\in \{1,2\}$, consider the corresponding partition matroid $\mathcal{M}_i=\{T,\mathcal{F}_i\}$, $\mathcal{F}_i=\{S: S\subseteq T, |S\cap C_k|\le 1, \text{for all}~ k\in [l]\}$. 
Let the function $r_i$ be the rank function of matroid $\mathcal{M}_i$, the cost function $c_i(\cdot)$ is defined as $c_i(S)=|S|-r_i(S)$ for any item set $S\in T$.
For the first arrived item $e_1$, the marginal utilities for agent $1$ and agent $2$ are $0$.
We allocate item $e_1$ to agent $1$. 
When the second item $e_2$ arrives, the marginal utilities for agent $1$ and agent $2$ are $0$.
We will discuss the following two cases. 

\noindent\textbf{Case 1}: we allocate $e_2$ to agent $1$.
For the third item $e_3$, the marginal utility for agent $1$ is $1$, and the marginal utility for agent $2$ is $0$. In order not to break the requirement of $\mathsf{MMS}_{1}(S)$, we must allocate item $e_3$ to agent $2$.
When the fourth item $e_4$ arrives, the marginal utilities for agent $1$ and agent $2$ are $1$. 
No matter which agent is assigned, the cost of its allocation is $1$, but $\mathsf{MMS}_{1}(T)=\mathsf{MMS}_{2}(T)=0$, $\min USC(T)=0$.

\noindent\textbf{Case 2}: we allocate $e_2$ to agent $2$.
For the third item $e_3$, the marginal utilities for agent $1$ and agent $2$ are $0$. Continue to discuss in two sub-cases
\textbf{Case 2-1}: we allocate $e_3$ to agent $2$.
For the fourth item $e_4$, the marginal utilities for agent $1$ and agent $2$ are $1$. 
No matter which agent is assigned, the cost of its allocation is $1$, but $\mathsf{MMS}_{1}(T)=\mathsf{MMS}_{2}(T)=0$, $\min USC(T)=0$.
\textbf{Case 2-2}: we allocate $e_3$ to agent $1$.
For the fourth item $e_4$, the marginal utilities for agent $1$ and agent $2$ are $0$. 
In order not to break the requirement of $\mathsf{MMS}_{1}(S)$, we must allocate item $e_4$ to agent $2$.
When the fifth item $e_5$ arrives, 
the marginal utilities for agent $1$ and agent $2$ are $0$. 
In order not to break the requirement of $\mathsf{MMS}_{1}(S)$, we must allocate item $e_5$ to agent $2$.
When the sixth item $e_6$ arrives, the marginal utilities for agent $1$ and agent $2$ are $1$. 
No matter which agent is assigned, the cost of its allocation is $1$, but $\mathsf{MMS}_{1}(T)=\mathsf{MMS}_{2}(T)=0$, $\min USC(T)=0$.
\end{proof}

\subsection{Additive Personalized Bi-valued Costs}
In this part, we consider additive personalized bi-valued cost functions. For any agent $i \in N$, her cost function $c_i$ is additive personalized bi-valued if $c_i$ is additive and $c_i(e)\in \{a_i, b_i\}$ for any $e\in T$, where $0<a_i\leq b_i$. Similar to the results of goods setting, even for two agents, an EF1 or MMS allocation cannot be guaranteed.

\subsubsection{Two Agents}
First, we study the additive personalized bi-valued cost functions for two agents and propose the Adapted Chores Envy-Graph Procedure (Algorithm \ref{alg: two_agents_ef1_mms_chores}) to return an allocation that achieves the tight approximation guarantee for EF1 and has a good approximation for MMS.
Using techniques similar to the case of goods, we take the Chores Envy-Graphs procedure (Algorithm \ref{alg: chores_egp}) as the base algorithm.
When it may not be further executed (an envy cycle will appear), we skip to the Chores Preliminary-Breaking-Cycle algorithm (Algorithm \ref{alg: two_agents_cycle_chores}), where we carefully allocate the future items and guarantee the desired approximation.
During the execution of the Chores Preliminary-Breaking-Cycle (CPBC) algorithm, in one complicated case ($(\lambda_1, \mu_1)=(2,2)$), we introduce the Chores Deep-Breaking-Cycle algorithm (Algorithm \ref{alg: two_agents_cycle_cases_chores}) to address it.
The Chores Deep-Breaking-Cycle (CDBC) algorithm may be consecutively executed in many iterations, but the CPBC algorithm is at most consecutively executed in two iterations.
Note that the CPBC and CDBC algorithms are used to remove the possible envy cycle and guarantee that in some further iteration, we can return to the CEGP algorithm.

\begin{algorithm}[tb]
\caption{(ACEGP) Adapted Chores Envy-Graph Procedure}
\label{alg: two_agents_ef1_mms_chores}
\KwIn{An instance $(T, N, (c_1,c_2))$ with additive personalized bi-valued cost functions}
\KwOut{A complete approximate EF1 and approximate MMS allocation ${\bf A}$}

Let ${\bf A} = (\emptyset, \ldots, \emptyset)$; 

Initialize $\alpha = 1$; \tcp{Algorithm Indicator}

Initialize $(\lambda_1, \mu_1)$ and  $(\lambda_2, \mu_2) $, where $\lambda_1=\lambda_2=1 $ and $\mu_1= \mu_2 =0$; \tcp*{$\lambda$: the times of consecutive execution for the  CPBC or CDBC algorithm, and $\mu$: the case label in the CPBC or CDBC algorithm.}

Build an envy graph $G$ for two agents;

\SetKwProg{Def}{when}{ do}{}
\Def{item $e \in T$ arrives}{
\uIf{$\alpha = 1$}{

$A, \alpha \leftarrow \mathsf{CEGP}(N, A, \alpha, e)$;

\If{$\alpha=2$}{
$A, \alpha \leftarrow \mathsf{CPBC}(N, A, \alpha, e, (\lambda_1, \mu_1))$;
\tcp*{Envy cycle will appear.}

}

}
\uElseIf{$\alpha = 2$}{
$A, \alpha, (\lambda_1, \mu_1) \leftarrow \mathsf{CPBC}(N, A, \alpha, e, (\lambda_1, \mu_1))$;

}
\Else{
$A, \alpha, (\lambda_2, \mu_2) \leftarrow  \mathsf{CDBC}(N, A, \alpha, e, (\lambda_2, \mu_2))$;
}

Update the envy graph $G$;

}
\Return ${\bf A}$;
\end{algorithm}

\begin{algorithm}[tb]
\caption{(CEGP) Chores Envy-Graph Procedure}
\label{alg: chores_egp}
\KwIn{$N, {\bf A}, \alpha, e $ }
\KwOut{A complete allocation ${\bf A}$ and $\alpha$}

\eIf{
there is no envy between the two agents
}
{
\eIf{$\exists i \in N$, $c_i(e) = a_k$ }
{
Choosing the agent $i$ who values it at $a_i$ (breaking ties arbitrarily);

}
{
Arbitrarily choose one agent $i$;
}
}
{
Let $i \in N$ be an agent who does not envy the other agent and $j = 2-i$;

\If{$c_j(A_j)>c_j(A_i \cup \{e\})$}{
$\alpha=2$;

\Return ${\bf A}$ and $\alpha$;
}

}
$A_i = A_i \cup \{e\}$;

\Return ${\bf A}$ and $\alpha$;
\end{algorithm}

\begin{algorithm}[tb]
\caption{(CPBC) Chores Preliminary-Breaking-Cycle}
\label{alg: two_agents_cycle_chores}
\KwIn{$N, {\bf A}, \alpha, e, (\lambda_1, \mu_1)$}
\KwOut{A complete allocation ${\bf A}$, $\alpha$ and $(\lambda_1, \mu_1)$}

\eIf {$\lambda_1 = 1$}{

\eIf{$|A_j|=1$}{

$A_j = A_j \cup \{e\}$;

$\alpha=1$, $\lambda_1=1$, and $\mu_1=0$;

}
{

\eIf{$c_i(e) = b_i$}{

$A_j = A_j \cup \{e\}$;

$\lambda_1=2$ and $\mu_2 = 1$;
}
{
$A_i = A_i\cup \{e\}$;

$\lambda_1=2$ and $\mu_2 = 2$;
}
}

}
{

\eIf{$\mu_2 = 1$ 
}{

$A_i = A_i \cup \{e\}$;

}
{

\eIf{$c_j(e) = a_j$}{

$A_j = A_j \cup \{e\}$;

}
{

$A_i = A_i \cup \{e\}$;

$\alpha=3$, $\lambda_1=1$, and $\mu_2=0$;
}
}

$\alpha=1$, $\lambda_1=1$, and $\mu_2=0$;
}

\Return ${\bf A}$, $\alpha$ and $(\lambda_1, \mu_1)$;
\end{algorithm}

\begin{algorithm}[tb]
\caption{(CDBC) Chores Deep-Breaking-Cycle}
\label{alg: two_agents_cycle_cases_chores}
\KwIn{$N, {\bf A}, \alpha, e, (\lambda_2, \mu_2)$}
\KwOut{A complete allocation ${\bf A}$, $\alpha$, and $(\lambda_2, \mu_2)$}

\If{$\lambda_2 \mod 2 =1$}{

\uIf{$c_j(e) = a_j$}{
$A_j = A_j \cup \{e\}$;

$\alpha=1$, $\lambda_2=1$, and $\mu_2 = 0$;
}
\uElseIf{$c_j(e) = b_j$ and $c_i(e) = b_i$}{

$A_j = A_j \cup \{e\}$;

$\lambda_2= \lambda_2+1$ and $\mu_2 = 1$;

}
\Else{
$A_i = A_i \cup \{e\}$;

$\lambda_2= \lambda_2+1$ and $\mu_2 = 2$;
}

}
\uElseIf{$\lambda_2 \mod 2 =0$}{
\uIf { 
$\mu_2 = 1$
}{
\eIf{$c_j(e) =b_j$}{

$A_i = A_i \cup \{e\}$;

$\lambda_2 = \lambda_2+ 1$ and $\mu_2=0$;
}
{
$A_j = A_j \cup \{e\}$;

$\alpha=1$, $\lambda_2 = 1$ and $\mu_2=0$;

}

}
\Else{

$A_j = A_j \cup \{e\}$;

$\lambda_2 = \lambda_2+1$ and $\mu_2=0$;

}
}

\Return ${\bf A}$, $\alpha$, and $(\lambda_2, \mu_2)$;

\end{algorithm}

\begin{definition}[Envy Graph for Chores]
Given an allocation ${\bf A}^k$, the corresponding envy graph is defined as $G = (V,E)$, where the vertex set $V$ corresponds to agent set $N$, and a directed edge $(i,j)\in E$ iff agent $i$ envies agent $j$, i.e., $c_i(A_i^k) > c_i(A_j^k)$. Additionally, the directed cycle in the envy graph is called an envy cycle.
\end{definition}

\begin{theorem}
\label{the: two_agents_bivalued_ef1_mms_chores}
 For the deterministic allocation of indivisible chores for two agents with additive personalized bi-valued cost functions, the Adapted Chores Envy-Graph Procedures (Algorithm \ref{alg: two_agents_ef1_mms_chores}) computes a complete allocation that satisfies $2$-EF1 and $\frac{5}{3}$-MMS.  
 \end{theorem}

Before we show the above Theorem, we first list some lemmas used in the proof.
 \begin{lemma}
 \label{lem: cegp}
    The CEGP algorithm (Algorithm \ref{alg: chores_egp}) computes an EF1 allocation without an envy cycle.
\end{lemma}

\begin{proof}
     Note that there are three kinds of input allocation.
    The first one is the initial empty allocation, which is EF1 without an envy cycle.
    The second one is the output allocation of the CPBC algorithm.
    The last one is the output allocation of the CDBC algorithm. 
    It suffices to show that the last two kinds of input allocation are EF1 without an envy cycle.
    That is because if the input allocation of the CEGP algorithm is EF1 without an envy cycle, by the description of the algorithm and induction, it is easy to check that the property of EF1 without an envy cycle is always maintained.
    By Lemmas \ref{lem: prelinary_cycle_chores} and \ref{lem: deep_cycle_chores}, we can derive that when the CPBC or CDBC algorithm outputs $\alpha=1$, implying that the CEGP algorithm is ready to be executed, the output allocation of the CPBC or CDBC algorithm is EF1 without an envy cycle.
    Therefore, our claim holds.   
\end{proof}

 \begin{observation}
 \label{obs: chores_preliminary_cycle}
     When the CPBC algorithm (Algorithm \ref{alg: two_agents_cycle_chores}) is executed, for the first arriving item $e$, we have $c_j(e)=a_j$.
 \end{observation}
 \begin{proof}
When the CPBC algorithm is ready to be executed, it means that if the arriving item $e$ is allocated to agent $i$, there will be an envy cycle, i.e, $c_j(A_j)>c_j(A_i \cup \{e\})$ and $c_i(A_i \cup \{e\})>c_i(A_j)$
For agent $j$, her bundle is EF1.
Suppose that $c_j(e)=b_j$.
Then, we have $c_j(A_j)\leq c_j(A_i)+c_j(f)\leq c_j(A_i)+b_j =c_j(A_i \cup \{e\})$ for some $f \in A_j$, which is a contradiction.
Therefore, it must hold that $v_j(e)=a_j$.
\end{proof}

\begin{observation}
     When the CPBC algorithm (Algorithm \ref{alg: two_agents_cycle_chores}) is executed, we have $|A_i|\geq 1$ and $|A_j|\geq 1$.
 \end{observation}
 \begin{proof}
     Suppose that $|A_i|=0$.
     Then, we have $|A_j|\geq 1$ since agent $j$ envies agent $i$.
     After allocating item $e$ to agent $i$, they envies each other, i.e, $c_i(e)>c_i(A_j)$ and $c_j(A_j)>c_j(e)$.
     In that case, we derive that $c_j(A_j)=b_j$ and $c_i(A_j)=a_i$, which contradicts the description of our algorithm.
     Suppose that $|A_j|=0$.
     Then, we have $|A_i|=0$.
     Otherwise, agent $j$ will not envy agent $i$, after allocating item $e$ to agent $i$.
     In that case, the envy cycle will not appear.
     Thus, it hold that $|A_i|\geq 1$ and $|A_j|\geq 1$.
 \end{proof}

  \begin{lemma}
  \label{lem: prelinary_cycle_chores}
     The CPBC algorithm (Algorithm \ref{alg: two_agents_cycle_chores}) computes a complete allocation satisfying:
     \begin{itemize}
       \item 2-EF1,
       \item or EF1 with an envy cycle,
       \item or EF1 without an envy cycle.
   \end{itemize}
 \end{lemma}
 \begin{proof}
     The input allocation $\textbf{A}=(A_1, A_2)$ of the CPBC algorithm is EF1.
     According to the description of our algorithm, it is executed in at most two rounds.
     Then, let us consider different rounds.

     \bm{$\lambda_1=1$}.
     When the first item $e_1$ arrives, we have $c_j(e_1)=a_j$ by Observation \ref{obs: chores_preliminary_cycle}.
     When $|A_j|=1$, we have $c_j(A_j)=b_j$ since she envies agent $i$.
     Then, this item is allocated to agent $j$.
     For agent $i$, her bundle is EF since she does not envy agent $j$.
     For agent $j$, her new bundle is EF1 since $|A_j|=1$ and $c_j(e_1)=a_j$.
     Thus, this allocation is EF1 without an envy cycle, and we set $\alpha=1$, $\lambda_1=1$, and $\mu_1=0$, implying that the CEGP algorithm is ready to be implemented.
     When $|A_j|>1$, let us consider two subcases about the cost of $e_1$.
     \begin{itemize}
         \item If $c_i(e_1)=b_i$, agent $j$ picks it and we set $\lambda_2=2$ and $\mu_1=1$.
         For agent $j$, her new bundle is $2$-EF1 since we have $c_j(A_j \setminus \{f\})\leq c_j(A_i)$ for some $f \in A_j$ and $c_j(A_i) \geq c_j(e_1)=a_j$.
         For agent $i$, her bundle is EF since she does not envy agent $j$.
         Thus, this allocation is $2$-EF1.
         \item If $c_j(e_1)=a_i$, agent $i$ picks it and we set $\lambda_2=2$ and $\mu_2=2$.
         For agent $i$, her new bundle is EF1 since after removing $e_1$ from her bundle, we have $c_j(A_j)\leq c_j(A_i)$.
         For agent $j$, her bundle is EF1.
         Thus, this allocation is EF1 with an envy cycle. 
     \end{itemize}
     In the above cases, the CPBC algorithm continues to be executed.

     \bm{$\lambda_1=2$}. 
     When the second item $e_2$ arrives, let us consider two cases about $\mu_1$.
     \begin{itemize}
         \item If $\mu_1 = 1$, agent $i$ picks it.
         For agent $i$, her new bundle is EF since we have $c_i(A_i \cup \{e_2\})\leq c_i(A_i)+b_i \leq c_i(A_j \cup \{e_1\})$.
         For agent $j$, her bundle is EF1 since we have $v_j(A_j \setminus \{f\})\leq v_j(A_i)$ for some $f \in A_j$ and $c_j(e_1) =a_j \leq c_j(e_2)$.
         Thus, this allocation is EF1 without an envy cycle, and we set $\alpha=1$, $\lambda_1=1$, and $\mu_1=0$, implying that the CEGP algorithm is ready to be executed.
         \item If $\mu_1 = 2$, let us consider two subcases about the cost of $e_2$.
         When $c_j(e_2)=a_j$, agent $j$ picks it.
         For agent $i$, her bundle is EF since we have $c_i(A_i) \leq c_i(A_j)$ and $c_i(e_1)=a_i \leq c_i(e_2)$.
         For agent $j$, her new bundle is EF1 since we have $c_j(A_j \setminus \{f\}) \leq c_j(A_i)$ for some $f \in A_j$ and $c_j(e_2)=a_j \leq c_j(e_1)$. 
         Thus, this allocation is EF1 without an envy cycle, and we set $\alpha=1$, $\lambda_1=1$, and $\mu_1=0$, implying that the CEGP algorithm is ready to be executed.
         When $c_j(e_2)=b_j$, agent $i$ picks it.
         For agent $i$, her new bundle is 2-EF1 since we have $c_i(A_i) \leq c_i(A_j)$ and $c_i(e_1)=a_i\leq c_i(A_j)$.
         For agent $j$, her bundle is EF1 since she does not pick any item in these two rounds.
         Thus, this allocation is 2-EF1, and we set $\alpha=3$, $\lambda_1=1$, and $\mu_1=0$, implying that we need to implement the CDBC algorithm.
     \end{itemize}
     Combining the above cases together, it is not hard to see that the claim holds.
 \end{proof}

 \begin{observation}
 \label{obs: chores_deep_cycle}
     When the CDBC algorithm (Algorithm \ref{alg: two_agents_cycle_cases_chores}) is executed, there are two properties about allocation $\textbf{A}$:
     \begin{itemize}
         \item $c_j(A_j)+a_j\leq c_j(A_i)$;
         \item Let $e_1$ and $e_2$ denote the two new allocated items in the CPBC algorithm, we have $c_i(A_i \setminus \{e\} \leq 2c_i(A_j)$, where $e \in \arg\max_{f \in \{e_1,e_2\}}c_i(f)$.
     \end{itemize}
 \end{observation}
 \begin{proof}
     When the CDBC algorithm is executed, the CPBC algorithm is terminated for $\lambda_1=\mu_1=2$.
     Then, by the proof of Lemma \ref{lem: prelinary_cycle_chores}, it is easy to check that statements are correct.
 \end{proof}

\begin{lemma}
\label{lem: deep_cycle_chores}
     The CDBC algorithm (Algorithm \ref{alg: two_agents_cycle_cases_chores}) computes a complete allocation satisfying:
     \begin{itemize}
       \item 2-EF1,
       \item or EF1 with an envy cycle,
       \item or EF1 without an envy cycle.
   \end{itemize}
 \end{lemma}
 
 \begin{proof}
     When the CDBC algorithm is executed, by Observation \ref{obs: chores_deep_cycle}, the input allocation $A = (A_1, A_2)$ is 2-EF1, where $c_j(A_j)+a_j \leq c_j(A_i)$ and $c_i(A_i \setminus \{e\})\leq 2c_i(A_j)$ for some $e \in A_i$.
     Let $f_1$ and $f_2$ denote the two items that are allocated to agent $j$ in the CPBC algorithm and we have $\min \{c_j(e_1),c_j(e_2)\}=a_j$.

     \bm{$\lambda_2=1$}.
     When the item $e_1$ arrives at $\lambda_2=1$, let us consider three cases of cost functions.
     \begin{itemize}
         \item When $c_j(e_1)=a_j$, agent $j$ picks it and we set $\alpha=1$, $\lambda_1=1$, and $\mu_1=0$.
         For agent $i$, her bundle is EF1 since we have $c_i(A_i \setminus \{f_1,f_2\})\leq c_i(A_j)$ and $c_i(f_1)=a_i\leq c_i(e_1)$.
         For agent $j$, her new bundle is EF since we have $c_j(A_j)+a_j \leq c_j(A_i)$.
         Thus, this allocation is EF1 without an envy cycle, implying that the CEGP algorithm is ready to be executed.
         \item When $c_j(e_1)=b_j$ and $c_i(e_1)=b_i$, agent $j$ picks it.
         Then, we set $\lambda_2=1$ and $\mu_2 = 1$.
         For agent $i$, her bundle is EF1 since we have $c_i(A_i \setminus \{f_1,f_2\})\leq c_i(A_j)$ and $c_i(f_1)=a_i\leq c_i(e_1)$.
         For agent $j$, her new bundle is EF1 since we have $c_j(A_j \setminus \{f\})\leq c_j(A_i \setminus \{f_1,f_2\})$ for some $f \in A_j$ and $c_j(e_1)=b_j \leq c_j(f_2)$.
         Thus, this allocation is EF1 with an envy cycle.
          \item When $c_j(e_1)=b_j$ and $c_i(e_1)=a_j$, agent $i$ picks it.
          Then, we set $\lambda_2=2$ and $\mu_2 = 2$.
          For agent $i$, her new bundle is still 2-EF1 since we have $\frac{c_i(A_i \setminus \{f\} \cup \{e_1\})}{c_i(A_j)}\leq\frac{c_i(A_j)+c_i(e_1)}{c_i(A_j)}\leq 2$ for some $f \in \{f_1,f_2\}$ where the inequality follows from the fact that $c_i(A_i \setminus \{f_1,f_2\}) \leq c_i(A_j)$, $c_i(f_1)=a_i$ and $|A_j| \geq 2$.
          Thus, this allocation is 2-EF1.
     \end{itemize}
     It can be seen that in the last two cases, the CDBC algorithm continues to be executed.

     \bm{$\lambda_2=2$}.
     When the item $e_2$ arrives at $\lambda_2=2$, let us consider two cases about $\mu_2$.
     \begin{itemize}
         \item When $\mu_2 = 1$, we need to consider two subcases.
         If $c_j(e_2)=b_j$, agent $i$ picks it.
         For agent $i$, her new bundle is 2-EF1 since we have $c_i(A_i \setminus \{f_1,f_2\}) \leq c_i(A_j)$, $c_i(f_1)=a_i \leq c_i(A_j) $ and $c_j(e_2) \leq c_i(e_1)=b_i$.
         For agent $j$, her bundle is EF since we have $c_j(A_j)\leq c_j(A_i)$ and $c_j(e_1)=b_j=c_j(e_2)$.
         Thus, this allocation is 2-EF1 and satisfies the properties in Observation \ref{obs: chores_deep_cycle}.
         Additionally, we set $\lambda_2=1$ and $\mu_2=0$, implying that the CDBC algorithm continues to be executed.
         If $c_j(e_2)=a_j$, agent $j$ picks it.
         For agent $i$, her bundle is EF since we have $c_i(A_i \setminus \{f_1,f_2\})$, $c_i(f_1)=a_i \leq c_i(e_2)$ and $c_i(f_2) \leq c_i(e_1)$.
         For agent $j$, her new bundle is EF1 since we have $c_j(A_j \setminus \{f\}) \leq c_j(A_i \setminus \{f_1,f_2\})$, $c_j(e_1)=b_j=c_j(f_2)$ and $c_j(e_2)=a_j=c_j(f_1)$.
         Thus, this allocation is EF1 without an envy cycle, and we set $\alpha=1$, $\lambda_2=1$, and $\mu_2=0$, implying that the CEGP algorithm is ready to be executed.
         \item When $\mu_2 = 1$, agent $j$ picks it,
         For agent $i$, her bundle is 2-EF1 since we have $c_i(A_i \setminus \{f_1,f_2\}) \leq c_i(A_j)$, $c_i(e_1)=a_i \leq c_i(e_2) $ and $c_i(f_1)=a_i \leq c_i(A_j)$.
         For agent $j$, her new bundle is EF since we have $c_j(A_j) \leq c_j(A_i) $ and $c_j(e_2)\leq c_j(e_1)=b_j$.
        Thus, this allocation is 2-EF1 and satisfies the properties in Observation \ref{obs: chores_deep_cycle}.
         Additionally, we set $\lambda_2=1$ and $\mu_2=0$, implying that the CDBC algorithm continues to be executed.
     \end{itemize}
Combining the above cases, it can be seen that our statement holds.
 \end{proof}

 \begin{proof}[Proof of Theorem \ref{the: two_agents_bivalued_ef1_mms_chores}]
 In the main body of the CAGEP algorithm \ref{alg: two_agents_ef1_mms_chores}, it can be seen that for different values of the algorithm indicator $\alpha$, we need to choose the corresponding algorithm to execute. 
By Lemmas \ref{lem: cegp}, \ref{lem: prelinary_cycle_chores}, and \ref{lem: deep_cycle_chores}, it can be concluded that the output allocation is complete and $2$-EF1.
Regarding MMS, by the implication between EF1 and MMS \cite{sun2021connections}, i.e., $\alpha$-EF1 $\Rightarrow$ $\frac{n\alpha+n-1}{n-1+\alpha}$-MMS, where $n$ is the number of agents, it is easy to see that the output allocation is $\frac{5}{3}$-MMS.
 \end{proof}
Next, we complement our approximation results by providing the lower bounds below.
\begin{theorem}
\label{the: ef1_bi_valued_impossiblity_results_chores}
    For two agents with additive personalized bi-valued cost functions, no deterministic online algorithm can compute a complete allocation that is $(2-\epsilon)$-EF1 or $(\frac{3}{2}-\epsilon)$-MMS for any $\epsilon >0$.
\end{theorem}

\begin{proof}
\begin{table}[tb]
% \begin{minipage}{\linewidth}
%     \centering
%     \begin{tabular}{c|c|c}
%     \toprule
%       & $t = 1$ & $t=2$     \\
%       \midrule
%         agent $1$  & \textcolor{red}{$1$} & \textcolor{red}{$1$} \\
%         \midrule
%         agent 2 & $1$ & $\frac{1}{\epsilon}$  \\
%         \bottomrule
%     \end{tabular}
%     \caption*{Case 1}
% \end{minipage}
\begin{minipage}{\linewidth}
    \centering
    \begin{tabular}{c|c|c|c}
    \toprule
      & $e_1$ & $e_2$ & $e_3$     \\
      \midrule
        agent $1$  & \textcolor{red}{$1$} & $1$ &  $\frac{1}{\epsilon}$ \\
        \midrule
        agent 2 & $1$ & \textcolor{red}{$\frac{1}{\epsilon}$} & \textcolor{red}{$\frac{1}{\epsilon}$} \\
        \bottomrule
    \end{tabular}
    \caption*{Case 1}
\end{minipage}
 \begin{minipage}{\linewidth}
    \centering
    \begin{tabular}{c|c|c|c|c}
    \toprule
      & $e_1$ & $e_2$ & $e_3$& $e_4$    \\
      \midrule
        agent $1$  & \textcolor{red}{$1$} & $1$ &  \textcolor{blue}{$\frac{1}{\epsilon}$} & \textcolor{blue}{$1$}  \\
        \midrule
        agent 2 & $1$ & \textcolor{red}{$\frac{1}{\epsilon}$} & $\frac{1}{\epsilon}$ & $\frac{1}{\epsilon}$\\
        \bottomrule
    \end{tabular}
    \caption*{Case 2.2}
\end{minipage}

\begin{minipage}{\linewidth}
\centering
    \begin{tabular}{c|c|c|c|c|c}
    \toprule
      & $e_1$ & $e_2$ & $e_3$& $e_4$ & $e_5$   \\
      \midrule
        agent $1$  & \textcolor{red}{$1$} & $1$ &  \textcolor{blue}{$\frac{1}{\epsilon}$} & $1$ & $\frac{1}{\epsilon}$  \\
        \midrule
        agent 2 & $1$ & \textcolor{red}{$\frac{1}{\epsilon}$} & $\frac{1}{\epsilon}$ & \textcolor{green}{$\frac{1}{\epsilon}$} & $\frac{1}{\epsilon}$\\
        \bottomrule
    \end{tabular} 
    \caption*{Case 2.3}
\end{minipage}
\caption{The impossibility results of EF1 for additive personalized bi-valued cost functions for two agents.}
    \label{tab: upper_bound_bi_valued_ef1_chores}
\end{table} 

\begin{table}[tb]
\begin{minipage}{\linewidth}
    \centering
    \begin{tabular}{c|c|c}
    \toprule
      & $e_1$ & $e_2$     \\
      \midrule
        agent $1$  & \textcolor{red}{$1$} & \textcolor{red}{$1$} \\
        \midrule
        agent 2 & $1$ & $1$  \\
        \bottomrule
    \end{tabular}
    \caption*{Case 1}
\end{minipage}
\begin{minipage}{\linewidth}
    \centering
    \begin{tabular}{c|c|c|c}
    \toprule
      & $e_1$ & $e_2$ & $e_3$     \\
      \midrule
        agent $1$  & \textcolor{red}{$1$} & $1$ &  2 \\
        \midrule
        agent 2 & $1$ & \textcolor{blue}{$1$} & $2$ \\
        \bottomrule
    \end{tabular}
    \caption*{Case 2}
\end{minipage}
% \begin{minipage}{\linewidth}
%     \centering
%     \begin{tabular}{c|c|c|c|c}
%     \toprule
%       & $t = 1$ & $t=2$ & $t = 3$ & $t=4$     \\
%       \midrule
%         agent $1$  & \textcolor{red}{$1$} & $1$ &  \textcolor{green}{1} & 3 \\
%         \midrule
%         agent 2 & $1$ & \textcolor{blue}{$\frac{3}{2}$} & 1 & $\frac{3}{2}$ \\
%         \bottomrule
%     \end{tabular}
%     \caption*{Case 2.3}
% \end{minipage}
%  \begin{minipage}{\linewidth}
%     \centering
%     \begin{tabular}{c|c|c|c|c}
%     \toprule
%       & $t = 1$ & $t=2$ & $t = 3$& $t=4$    \\
%       \midrule
%         agent $1$  & \textcolor{red}{$1$} & $1$ &  \textcolor{blue}{$\frac{1}{\epsilon}$} & \textcolor{blue}{$1$}  \\
%         \midrule
%         agent 2 & $1$ & \textcolor{red}{$\frac{1}{\epsilon}$} & $\frac{1}{\epsilon}$ & $\frac{1}{\epsilon}$\\
%         \bottomrule
%     \end{tabular}
%     \caption*{Case 2.2}
% \end{minipage}

% \begin{minipage}{\linewidth}
% \centering
%     \begin{tabular}{c|c|c|c|c|c}
%     \toprule
%       & $t = 1$ & $t=2$ & $t = 3$& $t=4$ & $t = 5$   \\
%       \midrule
%         agent $1$  & \textcolor{red}{$1$} & $1$ &  \textcolor{blue}{$\frac{1}{\epsilon}$} & $1$ & $\frac{1}{\epsilon}$  \\
%         \midrule
%         agent 2 & $1$ & \textcolor{red}{$\frac{1}{\epsilon}$} & $\frac{1}{\epsilon}$ & \textcolor{green}{$\frac{1}{\epsilon}$} & $\frac{1}{\epsilon}$\\
%         \bottomrule
%     \end{tabular} 
%     \caption*{Case 2.3}
% \end{minipage}
\caption{The impossibility results of MMS for additive personalized bi-valued cost functions for two agents.}
    \label{tab: upper_bound_bi_valued_mms_chores}
\end{table}  
    Let $a_1=a_2=1$ and $b_1=b_2=\frac{1}{\epsilon}$ where $0<\epsilon<1$.
    Consider the following online instance for $n=2$, where the cost of each item is shown in Table \ref{tab: binary_bivalued_ef1_mms_chores}.
     To guarantee the completeness of the chores allocation, we should guarantee that each arriving item is allocated to one agent.
     \paragraph{EF1.}
    When $t =1$, the arriving item has the cost of 1 for both agents.
    Without loss of generality, assume that agent 1 picks it.
    When $t =2$, the item that has the cost of 1 for agent 1 and $\frac{1}{\epsilon}$ for agent 2 arrives.
    To avoid the approximation ratio of EF1 being infinity, it is allocated to agent 2.
    Then, it is easy to check that the allocation is EF1 and MMS.
    In round $t =3$, the item that has the cost of $\frac{1}{\epsilon}$ for both agents arrives.
    Then, let us consider two cases.

    \textbf{Case 1.} If it is allocated to agent 2, no further item will appear.
    Thus, the allocation is $\frac{1}{\epsilon}$-EF1.

    \textbf{Case 2.} If it is allocated to agent 1, it is clear that this allocation is EF1.
    Next, when $t = 4$, the arriving item has the cost of $1$ for agent 1 and $\frac{1}{\epsilon}$ for agent 2.
    Then, let us consider two subcases.
    If agent 1 picks it, no further item will appear.
    So this allocation is trvially $(1+\frac{1}{\epsilon})$-EF1.
    If agent $2$ picks it, the allocation is EF1.
    At last, in round $t =5$, the arriving item has the cost of $\frac{1}{\epsilon}$ for both agents.
    If it is allocated to agent 1, the allocation is $(\frac{1}{2\epsilon}+\frac{1}{2})$-EF1.
    If it is allocated to agent 2, the allocation is $\frac{2}{1+\epsilon}$-EF1.

    Combining the above cases together, we can derive that the lower bound of EF1 is 2.

    \paragraph{MMS.} The constructed instance is below.
    Let $a_1=a_2=1$, $b_1=b_2 = 2$.
    When $t = 1$, the arriving item has the cost of 1.
    Assume that it is allocated to agent 1.
    Then, in round $t =2$, the item that has the cost of 1 for both agents arrives.
    If it is allocated to agent 1, no further allocation will appear.
    Thus, this allocation is 2-MMS.
    If it is allocated to agent 2, this allocation is MMS.
   In round $t =3$, the arriving item has the cost of 2 for both agents.
   No matter which agent picks it, the final allocation is $\frac{3}{2}$-MMS.
    % If it is allocated to agent 2, no further items appear.
    % Thus, this allocation is $\frac{5}{3}$-MMS.
    % If not, agent 1 picks it, and this allocation is MMS.
    % Finally, in round $t = 4$, the item that has the cost of 3 for agent 1 and $\frac{3}{2}$ for agent 2 arrives.
    % If it is allocated to agent 1, this allocation is $\frac{5}{3}$-MMS.
    % If not, agent 2 picks it, and this allocation is $\frac{6}{5}$.
    Combining the above cases together, we can derive that the lower bound of MMS is $\frac{3}{2}$.  
\end{proof}

At last, we examine the compatibility between approximate EF1 or MMS and USC, and provide the following negative results.
\begin{proposition}
\label{prop: bivalued_no_efficiency_chores}
Given any $\alpha,\gamma \geq 1$ and $\beta<2$, for the deterministic allocation of indivisible chores, no online algorithm can compute a complete allocation satisfying
\begin{itemize}
    \item $\alpha$-EF1 and $\gamma$-min-USC,
    \item or $\beta$-MMS and $\gamma$-min-USC,
\end{itemize}
even for two agents with additive personalized bi-valued cost functions. \end{proposition}

\begin{proof}
    Without loss of generality, we assume that $c_1 \in \{1,2\}$ and $c_2 \in \{\frac{1}{\epsilon},\frac{1}{\epsilon^2}\}$, where $0<\epsilon< 1$.
    When $t=1$, the first item $e_1$ that has the cost of 1 for agent 1 and $\frac{1}{\epsilon}$ for agent 2 arrives.
    If it is allocated to agent 2, no further items arrive.
    Hence, this allocation is clearly EF1, MMS, and  $\frac{1}{\epsilon}\rightarrow \infty$-min-USC.
    If not, agent $1$ picks it, and this allocation is EF1, MMS, and min-USC.
    Then, when $t=2$, the second arriving item $e_2$ has the cost of 1 for agent 1 and $\frac{1}{\epsilon^2}$ for agent 2.
    If it is allocated to agent 1, this allocation is $\infty$-EF1, 2-MMS, and min-USC.
    If it is allocated to agent 2, this allocation is EF1, MMS, and $\frac{\frac{1}{\epsilon^2}+\epsilon}{\epsilon+1} \rightarrow \infty$-min-USC.
    Combining the above cases, it can be seen that no online algorithm can compute a complete allocation that is $\alpha$-EF1 for any $\alpha \geq 1$ or $\beta$-MMS for any $\beta<2$, and $\gamma$-min-USC for any given $\gamma \geq 1$ when two agents have the additive personalized bi-valued cost functions.      \end{proof}

\subsubsection{The Combination with Additive Binary Costs}

Next, we consider the setting where one agent (agent $n$) has the additive personalized bi-valued cost function, and the remaining agents have additive binary cost functions and propose the Adapted-Chores-Picking Algorithm (Algorithm \ref{alg: binary_bivalued_ef1_mms_chores}) to compute a complete allocation that satisfies EF1 and MMS. 
The intuition of this algorithm is similar to the goods setting.
In Table \ref{tab: binary_bivalued_ef1_mms_chores}, we illustrate how our algorithm works for three agents.
Regarding efficiency, we show that EF1 or MMS is incompatible with USC in this setting.

\begin{algorithm}[t]
\caption{Adapted-Chores-Picking Algorithm}
\label{alg: binary_bivalued_ef1_mms_chores}
\KwIn{An instance $(T, N, (c_i)_{i \in N})$ with $(n-1)$ additive binary cost functions $(c_1.\ldots,c_{n-1})$ and one additive personalized bi-valued cost function $(c_n)$.}
\KwOut{A complete EF1 and MMS allocation ${\bf A}$}

Let ${\bf A} = (\emptyset, \ldots, \emptyset)$;

\SetKwProg{Def}{when}{ do}{}
\Def{item $e \in T$ arrives}{

Let $i_{min} \in \arg \min_{i \in [n-1]}|\{f~|f \in A_i ~and ~ c_i(f)>0 \}|$;

\eIf{$\exists i^* \in N$ such that $c_{i^*}(e)=0$}{

$A_{i^*} = A_{i^*} \cup \{e\}$;
}
{

\eIf{$c_n(e) =a_n$}{

\eIf{$ |\{f~|f \in A_{i_{min}} ~and ~ c_{i_{min}}(f)>0 \}| \geq |A_n| $}{
$A_n = A_n \cup \{e\}$;
}
{

$A_{i_{min}} = A_{i_{min}} \cup \{e\}$;

}

}
{
\eIf{$ |\{f~|f \in A_{i_{min}} ~and ~ c_{i_{min}}(f)>0 \}| \leq |A_n| $}{

$A_{i_{min}} = A_{i_{min}} \cup \{e\}$;

}
{
$A_n = A_n \cup \{e\}$;

}
}

}
}
\Return ${\bf A}$;
\end{algorithm}

\begin{table}[t]
\centering
    \begin{tabular}{c|c|c|c|c|c|c}
    \toprule
      & $e_1$ & $e_2$ & $e_3$& $e_4$ & $e_5$& $\ldots$   \\
      \midrule
        agent $1$  & $1$ & \textcolor{red}{$0$} &  \textcolor{red}{$1$} & 1 & \textcolor{red}{$1$}&$\ldots$\\
        \midrule
         agent $2$  & $1$ & 1 & 1 & \textcolor{red}{$1$} & 1&$\ldots$\\
        \midrule
        agent 3 & \textcolor{red}{$a_3$} & $a_3$ &  $a_3$ & $b_3$ & $b_3$&$\ldots$\\
        \bottomrule
    \end{tabular} 
    \caption{A simple illustration for three agents of Algorithm \ref{alg: binary_bivalued_ef1_mms_chores}.}
    \label{tab: binary_bivalued_ef1_mms_chores}
\end{table}

\begin{theorem}
\label{the: binary_bivalued_chores}
 For the deterministic allocation of indivisible chores for $(n-1)$ agents with additive binary cost functions and one agent with the additive personalized bi-valued cost function, the Adapted-Chores-Picking Algorithm (Algorithm \ref{alg: binary_bivalued_ef1_mms_chores}) computes a complete allocation that satisfies EF1 and MMS.  
\end{theorem}
\begin{proof}
    Given an online fair allocation instance $I=(T, N, (v_i)_{i\in N})$,
Fix an arbitrary round $k$, let $T=\{e_1,\cdots,e_k\}$ is the set of items that have already arrived, denoted by $\textbf{A}^k=(A^k_1,\cdots, A^k_n)$ the allocation at the end of round $t$ in Algorithm \ref{alg: binary_bivalued_ef1_mms_chores}.

    \paragraph{EF1.} We show that the computed allocation is EF1 by induction.
    For the base case, the empty allocation is trivially EF1.
    For the induction step, after the round $k$, the allocation $A^k$ is EF1.
    We show that after allocating item $e_{k+1}$, the allocation $A^{k+1}$ is also EF1.
    \begin{itemize}
        \item If there exists an agent $i \in N$ such that $c_i(e_{k+1})=0$, agent $i$ picks it.
        It is clear that this allocation is still EF1.
        \item If $c_n(e_{k+1})=a_n$, then let us consider two subcases.
        When $|\{f~|f \in A_{i_{min}} ~and ~ c_{i_{min}}(f)>0 \}| \geq |A_n|$, agent $n$ picks it.
        For agent $i \in [n-1]$, her bundle is still EF1.
        For agent $n$, her bundle is EF1 since she has the priority of picking the item with the cost of $a_n$, and the number of items in her bundle is no more than that of items with non-zero costs in agent $i_{min}$'s bundle.
        When $|\{f~|f \in A_{i_{min}} ~and ~ c_{i_{min}}(f)>0 \}| < |A_n|$, agent $i_{min}$ picks it.
        For any agent $i \in N \setminus\{i_{min}\}$, her bundle is still EF1.
        For agent $i_{min}$, her bundle is EF1 since she has the additive binary cost valuation and the lowest number of items with non-zero costs.  
        \item If $v_n(e_{k+1})=b_n$, the let us consider two subcases.
         When $|\{f~|f \in A_{i_{min}} ~and ~ c_{i_{min}}(f)>0 \}| \leq |A_n|$, agent $i_{min}$ picks it.
         For any agent $i \in N \setminus\{i_{min}\}$, her bundle is still EF1.
        For agent $i_{min}$, her bundle is EF1 since she has the additive binary cost valuation and the lowest number of items with non-zero costs. 
        For any agent $i \in N \setminus \{i_{min}\}$, her bundle remains untouched and is still EF1.
        When $|\{f~|f \in A_{i_{min}} ~and ~ c_{i_{min}}(f)>0 \}| > |A_n|$, agent $n$ picks it.
        For agent $i \in [n-1]$, her bundle is trivially EF1.
        For agent $n$, her bundle is EF since she has the priority of picking the item with the cost of $a_n$, and in this case, the extra item in another agent's bundle must have the cost of $b_n$ for agent $n$.        
    \end{itemize}
Combining the above cases, it can be seen that after allocating item $e_{k+1}$, the allocation $\textbf{A}^{k+1}$ is still EF1.

\paragraph{MMS.}
Fix an arbitrary agent $i \in N$ and an arbitrary round $k$.
If her bundle is EF, we are done.
If not, then let us consider two cases.
When $i \in [n-1]$, she has the additive binary valuation function, and by our algorithm description, the number of items with non-zero costs in her bundle is at most one more than that of the other agent's bundle.
Thus, it is not hard to see that her bundle is MMS.
When $i = n$, currently, the number of items in her bundle is one more than that of the other agent's bundle, and the extra item has the cost of $a_n$.
After removing this item, the cost of her bundle is no more than that of the other agent's bundle from her perspective, so we can conclude that her bundle is MMS.
\end{proof}

\begin{proposition}
\label{prop: binary_bivalued_no_efficiency_chores}
Given any $\alpha, \beta\geq1$, for the deterministic allocation of indivisible chores, no online algorithm can compute a complete allocation satisfying
\begin{itemize}
    \item $\alpha$-EF1 and $\beta$-min-USC,
    \item or MMS and $\beta$-min-USC,
\end{itemize}
even for two agents where one agent has the additive binary cost function and the other agent has the additive personalized bi-valued cost function. 
\end{proposition}

\begin{proof}
    Without loss of generality, we assume that $c_1\in \{0,1\}$ and $c_2 \in \{\frac{1}{\epsilon},\frac{1}{\epsilon^2}\}$, where $0<\epsilon< 1$.
    When $t=1$, the first item $e_1$ that has the cost of 1 for agent 1 and $\frac{1}{\epsilon}$ for agent 2 arrives.
    If it is allocated to agent 2, no further items arrive.
    Hence, this allocation is clearly EF1, MMS, and $\frac{1}{\epsilon}\rightarrow \infty$-min-USC.
    If not, agent $1$ picks it, and this allocation is EF1, MMS, and min-USC.
    Then, when $t=2$, the second arriving item $e_2$ has the cost of 1 for agent 1 and $\frac{1}{\epsilon^2}$ for agent 2.
    If it is allocated to agent 1, this allocation is $\infty$-EF1, 2-MMS, and min-USC.
    If it is allocated to agent 2, this allocation is EF1, MMS, and $\frac{\frac{1}{\epsilon^2}+1}{2}= \frac{1}{2\epsilon^2}+\frac{1}{2}\rightarrow\infty$-min-USC.
    Combining the above cases, it can be seen that no online algorithm outputs a complete allocation that is $\alpha$-EF1 or MMS in the setting of additive binary and additive personalized bi-valued cost functions, and guarantees $\beta$-min-USC for any given $\alpha,\beta \geq 1$.      \end{proof}

\section{Conclusion and Future Work}
In our paper, we study the fair online allocation of indivisible items and present a series of positive and negative results regarding the existence and approximation results in fairness and efficiency with various valuation or cost functions, which provide a comprehensive view of what can be achieved.
For future work, the most interesting direction is to fill the gap in the approximation of EF1 and MMS for bi-valued valuation or cost functions when there are more than two agents, because it involves the complicated envy cycle, and the current technique cannot be directly generalized.
Additionally, based on our results in other online settings, it would also be interesting to study the items with deadline setting in a more complicated valuation or cost function setting. 
%since we have seen the improvement for two agents with bi-valued valuation or cost functions.

\newpage
\bibliographystyle{plainnat}
\bibliography{ref}

\newpage
\appendix

\section{Other Online Settings}
\label{sec: other_setting}

Online algorithms face two roadblocks: unknown future information and the inability to revoke previous decisions. A natural question is how to avoid the worst-case scenario in online algorithms.
In this section, we study two special settings.
One is that each arriving item can wait for one period and does not need to be allocated immediately, which can be found in the online matching literature \cite{ashlagi2018maximum}.
In this setting, we prove that an EF1 allocation for two agents with additive personalized bi-valued valuation/cost functions exists.
The other one is knowing some information in advance, which originates from the online MMS allocation with normalized valuations \cite{zhou2023multi}.
In this setting, we show that the Marginal-Greedy Algorithm computes a non-wasteful EF1 allocation when items are goods, and the Round-Robin algorithm computes a complete EF1 allocation when items are chores.

% In the last section, we discuss two special online item settings. 
% One is that each arriving item can wait for some periods without being allocated immediately, which can be found in the online matching literature \cite{ashlagi2018maximum}.
% In this setting, we prove that an EF1 allocation for two agents with bi-valued valuation/cost functions exists.
% The other one is the monotone online instance, which originates from the online MMS allocation with normalized valuations \cite{zhou2023multi}.
% In this setting, we show that the Marginal-Greedy Algorithm computes a non-wasteful EF1 allocation when items are goods, and the Round-Robin algorithm computes a complete EF1 allocation when items are chores.
%The other one is that the agents' preferences over items are known.

\subsection{Items with Deadline}
In this part, each item has a deadline $d$ after its arrival.
After $d$ units of time, we must decide how to allocate it.
For two agents with additive personalized bi-valued valuation functions in the goods setting, we design the Adapted Matching Algorithm (Algorithm \ref{alg: two_agents_ef1_mms_store_one_item})
to compute an NW allocation that satisfies EF1.
For two agents with additive personalized bi-valued cost functions in the chores setting, we design the Adapted Chores Matching Algorithm (Algorithm \ref{alg: two_agents_ef1_mms_store_one_item_chores})
to compute a complete allocation that satisfies EF1.
The common intuition of the above two algorithms is that we find a perfect matching between two agents and two items after two consecutive items arrive so that we can guarantee that there is no envy cycle between two agents and EF1 always holds.

\begin{theorem}
\label{the: two_agents_bivalued_save_item}
 For the deterministic allocation of indivisible goods for two agents with additive personalized bi-valued valuation functions, if we allow each item to stay $d=1$ period after its arrival, the Adapted Matching Algorithm (Algorithm \ref{alg: two_agents_ef1_mms_store_one_item}) computes an NW allocation that satisfies EF1.     
\end{theorem}

\begin{lemma}
\label{lem: two_agents_bivalued_save_one_item}
    In Algorithm \ref{alg: two_agents_ef1_mms_store_one_item}, for the allocation $A$ after allocating two consecutive online items $e_1$ and $e_2$, we have the following two properties for any agent $i \in N$:
    \begin{itemize}
        \item If $v_i(A_i)<v_i(A_j)$, then $v_i(A_i)-v_i(A_j)=a_i-b_i$;
        \item If $v_i(A_i)>v_i(A_j)$, then $v_i(A_i)-v_i(A_j)\geq b_i-a_i$.
    \end{itemize}
\end{lemma}
\begin{proof}
We show the claim by induction.
For the base case, the statement trivially holds.
For the induction step, after allocating items $e_{k-1}$ and $e_k$, the two properties hold.
Without loss of generality, let $A$ denote the allocation after allocating items $e_{k-1}$ and $e_k$. 
Then, we show that allocating items $e_{k+1}$ and $e_{k+2}$, our claim also holds.
Fix an arbitrary agent $i \in N$.
Now, consider two cases about her.
\begin{itemize}
    \item $v_i(A_i) \geq v_i(A_j)$.
    Because one agent picks exactly one item, it is clear that after allocating items $e_{k+1}$ and $e_{k+2}$, the claim holds since $v_i(e_{k+1}),v_i(e_{k+2})\in \{a_i, b_i\}$.
    \item $v_i(A_i)< v_i(A_j)$.
    In this case, according to our algorithm description, she picks one item with a value that is greater than or equal to that of the other one.
    Suppose that she picks $e_{k+1}$.
    If $v_i(e_{k+1})=v_i(e_{k+2})$, then we have $v_i(A_i \cup\{e_{k+1}\})-v_i(A_j \cup \{e_{k+2}\})=a_i-b_i+v_i(e_{k+1})-v_i(e_{k+2})=a_i-b_i$.
    If $v_i(e_{k+1})>v_i(e_{k+2})$, then we have $v_i(A_i \cup\{e_{k+1}\})-v_i(A_j \cup \{e_{k+2}\})=a_i-b_i+v_i(e_{k+1})-v_i(e_{k+2})=a_i-b_i+b_i-a_i=0$.
\end{itemize}
This completes the induction and establishes the correctness of our claim.
\end{proof}

\begin{proof}[Proof of Theorem \ref{the: two_agents_bivalued_save_item}]
Given an online fair allocation instance $I=(T, N, (v_i)_{i\in N})$,
Fix an arbitrary round $k$, let $T=\{e_1,\ldots,e_k\}$is the set of items that have already arrived, denoted by $\textbf{A}^k=(A^k_1,\ldots, A^k_n)$ the allocation at the end of round $t$ in Algorithm \ref{alg: two_agents_ef1_mms_store_one_item}.

\paragraph{EF1.} We will show that the returned allocation $\textbf{A}$ is always EF1 by induction.
For the base case, the empty bundle is trivially EF1 without an envy cycle.
For the induction step, after allocating item $e_t$ for any $t \in [k]$, the allocation $A^t$ is EF1 without an envy cycle. 
Next, we show that after allocating item $e_{k+1}$, the allocation $A^{k+1}$ is EF1 without an envy cycle.
Let us consider three cases about item $e_{k+1}$.
\begin{itemize}
    \item $\lambda=1$ and item $e_{k+1}$ is the last arriving item.
When envy exists between two agents, assume that agent $i$ is the agent who envies agent $j = 2-i$.
If she picks $e_{k+1}$, we are done.
If not, agent $j$ picks it and we get $v_i(e_{k+1})=a_i$.
For agent $j$, she still does not envy agent $i$, and her new bundle is EF.
For agent $i$, by Lemma \ref{lem: two_agents_bivalued_save_one_item}, we have $v_i(A_i^k)-v_i(A_j^k)=a_i-b_i$.
Then, we have $v_i(A_i^k) - v_i(A_j^k \cup \{e_{k+1}\} \setminus \{g\}) = a_i-b_i-a_i+b_i=0$ where $A_{i}^k$ is EF1 and there must exists an item $g \in A_j$ such that $v_i(g)=b_i$, implying the allocation $A^{k+1}$ is EF1 without an envy cycle.
When no envy exists between two agents, allocating $e_{k+1}$ to any agent will not violate EF1.
\item $\lambda = 2$ and item $e_{k+1}$ needs to be allocated firstly.
When envy exists between two agents, assume that agent $i$ is the agent who envies agent $j = 2-i$.
If agent $i$ picks $e_{k+1}$, we are done.
If not, for agent $i$, her bundle is still EF1 since we have $v_i(A_i)-v_i(A_j \cup \{e_{k+1}\} \setminus \{g\})=a_i-b_i-a_i+b_i=0$ where there must exist $g \in A_j$ such that $v_i(g)=b_i$.
For agent $j$, she does not envy agent $i$ and by Lemma \ref{lem: two_agents_bivalued_save_one_item}, we have $v_j(A_j)-v_j(A_i)\geq b_j-a_j$.
Thus, her bundle is EF.
When no envy exists between two agents, no matter which agent picks $e_{k+1}$, the new allocation is EF1 without an envy cycle.
\item $\lambda = 2$ and item $e_{k+1}$ needs to be allocated secondly.
When envy exists between two agents, assume that agent $i$ is the agent who envies agent $j = 2-i$.
If she picks it, we are done.
If not, it means that item $e_{k}$ is allocated to agent $i$ and $v_i(e_{k})\geq v_i(e_{k+1})$.
For agent $i$, considering that $A_{i}^{k-1}$ is EF1, we can derive that her bundle is still EF1.
For agent $j$, it is trivial that her new bundle is EF.
When no envy exists between two agents, no matter which agent picks $e_{k+1}$, the new allocation is EF1 without an envy cycle.
\end{itemize}
Combining the above case, we can conclude that the allocation $\textbf{A}^{k+1}$ is EF1 without an envy cycle. 
This completes the induction and establishes the correctness of our claim.
\end{proof}

\begin{theorem}
\label{the: two_agents_bivalued_save_item_chores}
 For the deterministic allocation of indivisible chores for two agents with additive personalized bi-valued cost functions, if we allow each item to stay $d=1$ period after its arrival, the Adapted Chores Matching algorithm (Algorithm \ref{alg: two_agents_ef1_mms_store_one_item_chores}) computes a complete allocation that satisfies EF1.     
\end{theorem}

\begin{lemma}
\label{lem: two_agents_bivalued_save_one_item_chores}
    In Algorithm \ref{alg: two_agents_ef1_mms_store_one_item_chores}, for the allocation $\textbf{A}$ after allocating two consecutive online items $e_1$ and $e_2$, we have the following two properties for any agent $i \in N$:
    \begin{itemize}
        \item If $c_i(A_i)<c_i(A_j)$, then $c_i(A_i)-c_i(A_j)\leq a_i-b_i$;
        \item If $c_i(A_i)>c_i(A_j)$, then $c_i(A_i)-c_i(A_j)= b_i-a_i$.
    \end{itemize}
\end{lemma}

\begin{proof}
    We show the claim by induction.
    For the base case, the statement trivially holds.
    For the induction step, after allocating items $e_{k-1}$ and $e_{k}$, the two properties hold.
    Let $\textbf{A}$ denote such an allocation.
    Then, we show that after allocating items $e_{k+1}$ and $e_{k+2}$, our claim also holds.
    Fix an arbitrary agent $i \in N$.
    Now, consider two cases about her.
    \begin{itemize}
        \item $c_i(A_i)< c_i(A_j)$.
        After allocating two items $e_{k+1}$ and $e_{k+2}$, the two properties hold since agent $i$ only picks one item and has a additive personalized bi-valued cost function i.e., $c_i(e_{k+1}),c_i(e_{k+2})\in \{a_i,b_i\}$.
        \item $c_i(A_i) > c_i(A_j)$.
        By our algorithm description, she picks one item with a cost that is lower than or equal to the other one.
        Suppose that she picks $e_{k+1}$.
        If $c_i(e_1)=c_i(e_{k+2})$, we have $c_i(A_i \cup \{e_{k+1}\}) - c_i(A_j \cup \{e_{k+2}\})=b_i-a_i$.
        If $c_i(e_1)<c_i(e_2)$, we have $c_i(A_i \cup \{e_{k+1}\}) - c_i(A_j \cup \{e_{k+2}\})=b_i-a_i+a_i-b_i=0$.
    \end{itemize}
    This completes the induction and establishes the correctness of our claim.
\end{proof}

\begin{proof}[Proof of Theorem \ref{the: two_agents_bivalued_save_item_chores}]
Given an online fair allocation instance $I=(T, N, (c_i)_{i\in N})$,
Fix an arbitrary round $k$, let $T=\{e_1,\ldots,e_k\}$ be the set of items that have already arrived, denoted by $\textbf{A}^k=(A^k_1,\ldots, A^k_n)$ the allocation at the end of round $t$ in Algorithm  \ref{alg: two_agents_ef1_mms_store_one_item_chores}.

\paragraph{EF1.}
We will show that the output allocation $\textbf{A}$ is always EF1 by induction.
For the base cases, the empty bundle is EF1 without an envy cycle.
For the induction step, after allocating item $e_{t}$ for any $t \in [k]$, the allocation $\textbf{A}^{t}$ is EF1 without an envy cycle.
Next, we show that after allocating item $e_{k+1}$, the allocation $\textbf{A}^{k+1}$ is EF1 without an envy cycle.
Let us consider three cases about item $e_{k+1}$.
\begin{itemize}
    \item $\lambda=1$ and item $e_{k+1}$ is the last arriving item.
    When envy exists between two agents, assume that agent $i$ is the agent who does not envy agent $j = 2-i$.
    If she picks it, we are done.
    If agent $j$ picks it, it means that $c_j(e_{k+1})=a_{j}$.
    For agent $i$, she does not envy agent $j$, so her bundle is still EF.
    For agent $j$, by Lemma \ref{lem: two_agents_bivalued_save_one_item_chores}, we have $c_j(A_j^k)-c_j(A_i^k)=b_j-a_j$.
    Then, we have $c_j(A_j^k \setminus \{f\} \cup \{e_{k+1}\})-c_j(A_i)=b_j-a_j-b_j+a_j=0$ where $A_{j}^k$ is EF1 and there must exist an item $f$ such that $c_j(f)=b_j$, implying that allocation $\textbf{A}^{k+1}$ is EF1 without an envy cycle.
    When there is no envy between two agents, allocating $e_{k+1}$ to any agent will not violate EF1.

    \item $\lambda = 2$ and item $e_{k+1}$ needs to allocated firstly.
    When envy exists between two agents, assume that agent $i$ is the agent who does not envy agent $j = 2-i$.
    If agent $i$ picks it, we are done.
    If not, agent $j$ picks it.
    For agent $i$, she does not envy agent $j$, so her bundle is EF.
    For agent $j$, her new bundle is EF1 since we have $c_j(A_j^k)-c_j(A_i^k)=b_j-a_j$ and $c_j(e_{k+1})=a_j$ by the condition of our algorithm.
    Thus, this allocation is EF1 without an envy cycle.
    When there is no envy between two agents, no matter which agent picks $e_{k+1}$, the new allocation is EF1 without an envy cycle.
    \item $\lambda=2$ and item $e_{k+2}$ needs to be allocated secondly.
    When envy exists between two agents, assume that agent $i$ is the agent who does not envy agent $j = 2-i$.
    If agent $i$ picks it, we are done.
    If not, agent $j$ picks it.
    In that case, item $e_{k}$ is allocated to agent $i$.
    For agent $i$, she does not envy agent $j$, so her bundle is EF.
     For agent $j$, we can still derive that her new bundle is EF1 since by our assumption, we have $A^{k-1}_{j}$ is EF1, and $c_j(e_{k+2})\leq c_j(e_{k+1})$ by the condition of our algorithm.
    Thus, this allocation is EF1 without an envy cycle.
    When there is no envy between two agents, no matter which agent picks $e_{k+1}$, the new allocation is EF1 without an envy cycle.
\end{itemize}
Combining the above cases, we can conclude that the allocation $\textbf{A}^{k+1}$ is EF1 without an envy cycle. 
This completes the induction and establishes the correctness of our claim.
\end{proof}

\begin{algorithm}[H]
\caption{Adapted Matching Algorithm}
\label{alg: two_agents_ef1_mms_store_one_item}
\KwIn{An instance $(T, N, (v_1,v_2))$ with additive personalized bi-valued valuation functions and deadline $d=1$}
\KwOut{An NW EF1 allocation $\textbf{A}$}

Let $A = (\emptyset, \ldots, \emptyset)$;

Initialize $\lambda = 1$;

\SetKwProg{Def}{when}{ do}{}
\Def{item $e \in T$ arrives}{

\eIf{$\lambda = 1$ }{

Let this item wait for $d$ period and use $e_1$ to denote it.

$\lambda = 2$;
}
{

Let $e_2$ denote the new arriving item;

\eIf{
there exists an envy between the two agents
}{

Let $i \in N$ be the agent who envies the other one and $j = 2-i$;

\uIf{
$v_i(e_1)= b_i$
}{
$A_i =A_i \cup \{e_1\}$ and $A_j = A_j \cup \{e_2\}$;
}
\uElseIf{$v_i(e_2)=b_i$}{
$A_j =A_j \cup \{e_1\}$ and $A_i = A_i \cup \{e_2\}$;
}
\uElseIf{$v_j(e_1)=b_j$}{
$A_j =A_j \cup \{e_1\}$ and $A_i = A_i \cup \{e_2\}$;
}
\Else{
$A_i =A_i \cup \{e_1\}$ and $A_j = A_j \cup \{e_2\}$;
}
}
{
\uIf{
$v_i(e_1)>v_i(e_2)$
}{
$A_i =A_i \cup \{e_1\}$ and $A_j = A_j \cup \{e_2\}$;
}
\uElseIf{
$v_i(e_1)<v_i(e_2)$
}{
$A_j =A_j \cup \{e_1\}$ and $A_i = A_i \cup \{e_2\}$;
}
\Else{
\eIf{
$v_j(e_1)>v_j(e_2)$
}{
$A_j =A_j \cup \{e_1\}$ and $A_i = A_i \cup \{e_2\}$;
}
{
$A_i =A_i \cup \{e_1\}$ and $A_j = A_j \cup \{e_2\}$;
}
}
}

$\lambda =1$;
}
}

\uIf{
$\lambda=1$ and no further items appear
}{
\eIf{there is no envy between two agents}{

Allocate item $e_1$ to an arbitrarily agent;

}
{

Let $i \in N$ be the agent who envies the other one and $j = 2-i$;

\eIf{$v_i(e_1)=b_i$}{
$A_i = A_i \cup\{e_1\}$;
}
{
$A_j = A_j \cup\{e_1\}$;
}
}

}

\Return $\textbf{A}$;
\end{algorithm}

\begin{algorithm}[H]
\caption{Adapted Chores Matching Algorithm}
\label{alg: two_agents_ef1_mms_store_one_item_chores}
\KwIn{An instance $(T, N, (c_1,c_2))$ with additive personalized bi-valued cost functions and deadline $d=1$}
\KwOut{A complete EF1 allocation $\textbf{A}$}

Let $A = (\emptyset, \ldots, \emptyset)$;

Initialize $\lambda = 1$;

\SetKwProg{Def}{when}{ do}{}
\Def{item $e \in T$ arrives}{

\eIf{$\lambda = 1$ }{

Let this item wait for $d$ period and use $e_1$ to denote it.

$\lambda = 2$;
}
{

Let $e_2$ denote the new arriving item;

\eIf{
there exists an envy between the two agents
}{

Let $i \in N$ be the agent who does not envy the other one, and $j = 2-i$;

\uIf{
$c_j(e_1)= b_j$
}{
$A_i =A_i \cup \{e_1\}$ and $A_j = A_j \cup \{e_2\}$;
}
\uElseIf{$c_j(e_2)=b_j$}{
$A_j =A_j \cup \{e_1\}$ and $A_i = A_i \cup \{e_2\}$;
}
\uElseIf{$c_i(e_1)=b_i$}{
$A_j =A_j \cup \{e_1\}$ and $A_i = A_i \cup \{e_2\}$;
}
\Else{
$A_i =A_i \cup \{e_1\}$ and $A_j = A_j \cup \{e_2\}$;
}
}
{
\uIf{
$c_j(e_1)>c_j(e_2)$
}{
$A_i =A_i \cup \{e_1\}$ and $A_j = A_j \cup \{e_2\}$;
}
\uElseIf{
$c_j(e_1)<c_j(e_2)$
}{
$A_j =A_j \cup \{e_1\}$ and $A_i = A_i \cup \{e_2\}$;
}
\Else{
\eIf{
$c_i(e_1)>c_i(e_2)$
}{
$A_j =A_j \cup \{e_1\}$ and $A_i = A_i \cup \{e_2\}$;
}
{
$A_i =A_i \cup \{e_1\}$ and $A_j = A_j \cup \{e_2\}$;
}
}
}

$\lambda =1$;
}
}

\uIf{
$\lambda=1$ and no further items appear
}{
\uIf{there is no envy between two agents}{

Allocate item $e_1$ to an arbitrarily agent;

}
\Else{

Let $i \in N$ be the agent who does not envy the other one and $j = 2-i$;

\eIf{$c_j(e_1)=b_j$}{
$A_i = A_i \cup\{e_1\}$;
}
{
$A_j = A_j \cup\{e_1\}$;
}
}

}

\Return $\textbf{A}$;
\end{algorithm}

\subsection{Partial Information in Advance}
In this part, agents' preferences over items are the same, and the arriving order of items is known.
With this information, we can obtain a non-wasteful EF1 allocation by the Marginal-Greedy (Algorithm \ref{alg: indivisible_submodular+binary}) in the goods setting and a complete EF1 allocation by the Round-Robin algorithm in the chores setting. 
\begin{definition}
    An online goods/chores instance is monotone if for any agent $i \in N$, we have
    $
    f_i(e_1) \geq \cdots \geq f_i(e_T)
    $ or for any agent $i \in N$, we have
     $
    f_i(e_1) \leq \ldots \leq f_i(e_T)
    $, where $f_i \in \{v_i,c_i\}$ .
\end{definition}

\begin{theorem}
\label{the: goods_monotone}
    For the deterministic allocation of indivisible goods from a monotone instance, the Marginal-Greedy Algorithm (Algorithm \ref{alg: indivisible_submodular+binary}) computes an NW allocation that is EF1.
\end{theorem}
\begin{proof}
 Without loss of generality, we assume that the picking sequence is $1, \ldots, n$.
 First, let us consider the case of $v_i(e_1) \leq \ldots\leq v_i(e_T)$ for any $i \in N$.
 Note that, for any agent $i \in N$, if she values some items at zero, these items will arrive before the items with non-zero values.
 In this phase, she does not pick anything until the item with non-zero values arrives.
 Additionally, when these zero items are allocated to other agents, from the perspective of agent $ i$, the value of the other agent's bundle remains unchanged.
 Thus, without loss of generality, we assume that $v_i(e_1) \neq 0$ for any agent $i \in N$.
 Fix an arbitrary round $k$ and an arbitrary agent $i \in N$.
    Let $\textbf{A}^{k}$ denote the allocation after the allocation of $e_{k}$.
    Suppose that she picks item $e_k$.
    For agent $i$, we have $v_i(A^{k}_i) \geq v_i(A^{k}_j \setminus \{e\})$, where $e$ is the last item added to agent $j$'s bundle, since for any $j <i$, agent $i$ is the latter one to pick the item and for any $j > i$, $v_i(e_s) \geq v_i(e_t)$ holds, where agent $i$ picks $e_s$, agent $j$ picks $e_t$ and $s-t=n-i+j$.
    For agent $j \neq i$, we have $v_j(A_j^{k})\geq v_j(A_{i}^{k} \setminus \{e\})$, where $e$ is the last item added to agent $i$'s bundle, since for any $j>i$, $v_j(e_s) \geq v_j(e_t)$, where agent $j$ picks $e_s$, agent $i$ picks $e_t$ and $s-t=j-i$, and $j<i$, $v_j(e_s) \geq v_j(e_t)$ holds, where agent $j$ picks $e_s$, agent $i$ picks $e_t$ and $s-t=n-j+i$. 
    Thus, the allocation $A^{k}$ is EF1.
    As for the case of $v_i(e_1) \geq  \ldots \geq v_i(e_T)$ for any $i \in N$, we can use a similar analysis to show that the returned allocation is EF1 and we omit it.
 \end{proof}

\begin{theorem}
\label{the: chores_monotone}
 For the deterministic allocation of indivisible chores from a monotone instance, the Round-Robin Algorithm computes a complete allocation that is EF1.
\end{theorem}
\begin{proof}
    Without loss of generality, we assume that the picking sequence is $1, \ldots, n$.
    First, let us consider the case of $c_i(e_1) \leq \ldots\leq c_i(e_T)$ for any $i \in N$.
    Fix an arbitrary round $k$ and an arbitrary agent $i \in N$.
    Let $\textbf{A}^{k}$ denote the allocation after the allocation of $e_{k}$.
    Suppose that she picks item $e_k$.
    For agent $i$, we have $c_i(A^{k}_i \setminus \{e_{k}\}) = c_i(A^{k-1}_i) \leq c_i(A^{k-1}_j)=c_i(A^{k}_j)$ since for any $j >i$, agent $i$ has the priority of picking the item and for any $j < i$, $c_i(e_s) \leq c_i(e_t)$ holds, where agent $i$ picks $e_s$, agent $j$ picks $e_t$ and $t-s=n-i+j$.
    For agent $j \neq i$, we have $c_j(A_j^{k})\leq c_j(A_{i}^{k})$ for any $j<i$ agent $j$ has the priority of picking items and $j>i$, $c_j(e_s) \leq c_j(e_t)$ holds, where agent $j$ picks $e_s$, agent $i$ picks $e_t$ and $t-s=n-j+i$. 
    Therefore, the allocation $\textbf{A}^{k}$ is EF1.
As for the case of $c_i(e_1) \geq  \ldots \geq c_i(e_T)$ for any $i \in N$, we can use a similar analysis to show that the returned allocation is EF1 and we omit it.\end{proof}

In the monotone instance, each agent has the same preference over all items, and the arriving order of items follows the nonincreasing/nondecreasing order in terms of value or cost.
Then, if we allow the arriving order to be arbitrary and keep the same preference, we have the following proposition.

\begin{proposition}
\label{prop: no_ef1_ido}
    For an online goods/chores instance where agents only have identical preferences, there exists an instance for which no EF1 allocation can be found, even for two agents with additive personalized bi-valued valuation or cost functions.
\end{proposition}
\begin{proof}
 \begin{table}[!tb]
 \begin{minipage}{\linewidth}
    \centering
    \begin{tabular}{c|c|c|c|c}
    \toprule
      & $e_1$ & $e_2$ & $e_3$& $e_4$    \\
      \midrule
        agent $1$  & \textcolor{red}{$1$} & $1$ &  $\epsilon$ & $1$  \\
        \midrule
        agent 2 & $1$ & \textcolor{red}{$\epsilon$} & \textcolor{red}{$\epsilon$} & $1$\\
        \bottomrule
    \end{tabular}
    \caption*{Case 1}
\end{minipage}

\begin{minipage}{\linewidth}
\centering
    \begin{tabular}{c|c|c|c|c|c}
    \toprule
      & $e_1$ & $e_2$ & $e_3$& $e_4$ & $e_5$   \\
      \midrule
        agent $1$  & \textcolor{red}{$1$} & $1$ &  \textcolor{blue}{$\epsilon$} & $1$ & $1$ \\
        \midrule
        agent 2 & $1$ & \textcolor{red}{$\epsilon$} & 1 & \textcolor{blue}{$\epsilon$}& $1$\\
        \bottomrule
    \end{tabular} 
    \caption*{Case 2}
\end{minipage}

% \begin{minipage}{\linewidth}
%     \centering
%     \begin{tabular}{c|c|c|c|c}
%     \toprule
%       & $t = 1$ & $t=2$ & $t = 3$& $t=4$    \\
%       \midrule
%         agent $1$  & \textcolor{red}{$\epsilon$} & $1$ &  $\epsilon$ & $\epsilon$ \\
%         \midrule
%         agent 2 & $\epsilon$ & \textcolor{red}{$\epsilon$} & \textcolor{blue}{$\epsilon$} & \textcolor{green}{1}\\
%         \bottomrule
%     \end{tabular}
%     \caption*{Case 2.2}
% \end{minipage}
\caption{The nonexistence of EF1 for two agents that have the same preference over goods.}
    \label{tab: no_ef1_ido_cases}
\end{table}

\begin{table}[!tb]
 \begin{minipage}{\linewidth}
    \centering
    \begin{tabular}{c|c|c|c}
    \toprule
      & $e_1$ & $e_2$ & $e_3$    \\
      \midrule
        agent $1$  & \textcolor{red}{$1$} & $1$ &  $\frac{1}{\epsilon}$   \\
        \midrule
        agent 2 & $1$ & \textcolor{red}{$\frac{1}{\epsilon}$} & \textcolor{red}{$\frac{1}{\epsilon}$} \\
        \bottomrule
    \end{tabular}
    \caption*{Case 1}
\end{minipage}

\begin{minipage}{\linewidth}
\centering
    \begin{tabular}{c|c|c|c|c|c}
    \toprule
      & $e_1$ & $e_2$ & $e_3$& $e_4$ & $e_5$   \\
      \midrule
        agent $1$  & \textcolor{red}{$1$} & $1$ &  \textcolor{blue}{$\frac{1}{\epsilon}$} & 1 & $\frac{1}{\epsilon}$ \\
        \midrule
        agent 2 & $1$ & \textcolor{red}{$\frac{1}{\epsilon}$} & $\frac{1}{\epsilon}$ & \textcolor{blue}{$\frac{1}{\epsilon}$} & $\frac{1}{\epsilon}$\\
        \bottomrule
    \end{tabular} 
    \caption*{Case 2}
\end{minipage}
% \begin{minipage}{\linewidth}
%     \centering
%     \begin{tabular}{c|c|c|c|c}
%     \toprule
%       & $t = 1$ & $t=2$ & $t = 3$& $t=4$    \\
%       \midrule
%         agent $1$  & \textcolor{red}{$\epsilon$} & $1$ &  $\epsilon$ & $\epsilon$ \\
%         \midrule
%         agent 2 & $\epsilon$ & \textcolor{red}{$\epsilon$} & \textcolor{blue}{$\epsilon$} & \textcolor{green}{1}\\
%         \bottomrule
%     \end{tabular}
%     \caption*{Case 2.2}
% \end{minipage}
\caption{The nonexistence of EF1 for two agents that have the same preference over chores.}
    \label{tab: no_ef1_ido_cases_chores}
\end{table} 
 First, we consider the online goods instance.
 Without loss of generality, we assume that $v_1, v_2 \in \{\epsilon, 1\}$, where $0<\epsilon<1$
 The value of each item is shown in Table \ref{tab: no_ef1_ido_cases}.
 When $t=1$, the arriving item has the value of $1$ for both agents.
 Without loss of generality, it is allocated to agent 1.
 For the second round $t=2$, the item that has the value of $1$ for agent 1 and $\epsilon$ for agent 2.
 To guarantee EF1, agent 2 picks it.
 In round $t=3$, the arriving item has the value of $\epsilon$ for both agents.
 Then, let us consider two cases related to allocating it.
 \begin{itemize}
     \item If agent 2 picks it, this allocation is still EF1.
     In the last round $t=4$, the arriving item has the value of $1$ for both agents.
     In this case, it is not hard to check that no matter which agent picks it, the resulting allocation is not EF1.
     \item If agent 1 picks it, this allocation is still EF1.
     Then, in round $t=4$, the arriving item has the value of $1$ for agent 1 and $\epsilon$ for agent 2.
     To guarantee that EF1 is not violated, it should be allocated to agent 2.
     In this last round $t =5$, the item has the value of $1$ for both agents.
     In this case, it is trivial that as long as one agent picks it, the final allocation is not EF1.
 \end{itemize}
 
Next, let us consider the online chores instance.
We assume that $c_1,c_2 \in \{1, \frac{1}{\epsilon}\}$ where $0<\epsilon<1$.
The cost of each item is shown in Table \ref{tab: no_ef1_ido_cases_chores}.
When $t=1$, the arriving item has the cost of 1 for both agents, and we assume that agent 1 picks it.
In round $t = 2$, the item that has the cost of 1 for agent $1$ and $\frac{1}{\epsilon}$ for agent 2.
To maintain EF1, agent 2 picks it.
For the next round $t = 3$, the arriving item has the cost of $\frac{1}{\epsilon}$ for both agents.
Then, let us consider two cases related to allocating it.
\begin{itemize}
    \item If it is allocated to agent 2, no further items appear.
    In this case, it is clear that this allocation is not EF1.
    \item If it is allocated to agent 1, in the following round $t= 4$, the arriving item has the cost of 1 for agent $1$ and $\frac{1}{\epsilon}$ for agent 2.
    It should be allocated to agent 2 to ensure that the property of EF1 still holds.
    In the last round $t = 5$, this arriving item has the cost of $\frac{1}{\epsilon}$ for both agents.
    In this case, it is easy to verify that no allocation is EF1 when it is allocated to one agent.
\end{itemize}

Combining the above cases, we can conclude that our statement holds.
\end{proof}

\clearpage

\end{document}